\documentclass[reqno,12pt]{amsart}
\makeatletter
\def\section{\@startsection{section}{1}%
	\z@{.7\linespacing\@plus\linespacing}{.5\linespacing}%
	{\bfseries
		\centering
}}
\def\@secnumfont{\bfseries}
\makeatother
\usepackage{amsmath}
\usepackage{amsfonts}
\usepackage{amssymb}
\usepackage{graphicx}
\usepackage{xcolor}
\usepackage{soul}
\usepackage{lineno}
\newtheorem{theorem}{Theorem}[section]

\newtheorem{corollary}[theorem]{Corollary}

\newtheorem{definition}[theorem]{Definition}

\newtheorem{lemma}[theorem]{Lemma}

\newtheorem{proposition}[theorem]{Proposition}
\newtheorem{remark}[theorem]{Remark}

\numberwithin{equation}{section}
\usepackage{lmodern}
\usepackage[colorlinks=true, allcolors=blue]{hyperref}
\colorlet{blu1}{blue!70!black}
\colorlet{blu2}{blue!50!black}
\colorlet{blu3}{blue!70!red}
\colorlet{blu4}{blue!60!green}
\colorlet{red1}{red!80}
\colorlet{red2}{red!50!black}
\colorlet{red3}{red!70!yellow}
\colorlet{red4}{red!50!yellow}
\colorlet{yel1}{yellow!50!black}
\colorlet{yel3}{yellow!20!blue}
\colorlet{gre1}{green!60!blue}
\colorlet{gre2}{green!60!black}
\colorlet{gre3}{green!40!black}

\usepackage{geometry} 

\title{Quantum Mechanics of Arc-Sine and Semi-Circle Distributions: A Unified Approach}
\date{}
\author{Luigi Accardi}
\address{Luigi Accardi: Centro Vito Volterra \\ Universit\`a di Roma "Tor Vergata"\\ Via Columbia 2, 00133  Roma \\ Italy}
\email{accardi@volterra.uniroma2.it}
	
\author{Tarek Hamdi}
\address{Tarek Hamdi: Department of Management Information Systems \\ College of Business and Economics \\ Qassim University  \\ Saudi Arabia
and Laboratoire d'Analyse Math\'ematiques et Applications LR11ES11 \\ Universit\'e de Tunis El-Manar \\ Tunisie}
\email{ t.hamdi@qu.edu.sa }

\author{Yun Gang Lu}
\address{Yun Gang Lu: Dipartimento di Matematica, Universit\`a di Bari, via Orabona 4, 70125 Bari,  Italy }
\email{yungang.lu@uniba.it}

\keywords{ Quantum Theory; Orthogonal polynomials; Semi--circle distribution; Arc--sine distribution;  Momentum operator; Kinetic Energy Operators; Evolution; Hilbert transform; Neumann series; Bessel functions}

\begin{document}
		
\maketitle

\begin{abstract}
This paper continues the program of applying beyond physics the technique of
\textbf{probabilistic quantization} and extending to the quantum mechanics
associated with the arc--sine distributions our previous results on the semi--circle distribution.
We derive analytical expressions for the momentum and kinetic energy operators using the arc--sine weighted Hilbert transform and express corresponding evolutions as Neumann series of Bessel functions. These series are applicable in various physical problems and in solving certain mixed difference equations and differential equations. 
Moreover, exploiting the similarity between the Jacobi sequences of the semi-circle and arc-sine measures, we establish a unified formulation of their quantum mechanics. We introduce the semicircle and arc--sine exponential vectors and the
corresponding coherent states and prove that, for both measures, the vacuum distributions of
the number operator in these states (arc--sine photon statistics) 
are a \textit{perturbation} of the geometric distribution (Gibbs states in Boson physics: see
the Introduction below for a discussion of the physical meaning of this perturbation).
The $*$--Lie algebra generated by canonical creation and annihilation operators
 of both probability measures is isomorphic to the $*$--Lie algebra generated by all rank-one operators in corresponding $L^2$-spaces.
The paper concludes with appendices that discuss the integral and Neumann series representations
of the $1$--parameter unitary groups generated by momentum in semi-circle and arc-sine cases.
\end{abstract}

\tableofcontents

\section{Introduction}
 The quantum decomposition of a classical random variable with all moments has revealed that quantization is not exclusive to physics, but rather a universal phenomenon in probability theory in the sense that every random variable with all moments canonically defines
its own quantum mechanics. The usual Boson quantization being associated with
the Gaussian distribution and the Fermi with the Bernoulli distribution (see the recent survey \cite{[AcLu20-Nigeria21]}). 
Probabilistic quantization has opened up new possibilities for modeling natural phenomena both
in physics and in many fields outside physics.
As canonical objects in mathematics often correlate with intriguing physical phenomena, studying the quantum mechanics (QM) associated with the simplest and most natural classes of probability measures significantly enhances our ability to construct models that appear highly nonlinear from the perspective of usual QM.\\
In this context, it is natural to examine the quantum mechanics canonically associated to the most important probability measures. Our previous work \cite{[AcHamLu22a],[AcHamLu22b]} outlined the structure of quantum mechanics associated with the semi-circle distribution
 (free quantum mechanics). In this paper, we extend these findings to the standard arc--sine distribution  (monotone quantum mechanics).
We derive analytical expressions for its momentum $P$ and kinetic energy $P^2/2$ operators using
the arc--sine weighted Hilbert transform and express corresponding evolutions in terms of
Neumann series of Bessel functions of the first kind.\\
These series have applications in various physical problems, such as describing internal gravity waves in a Boussinesq fluid and studying the propagation properties of diffracted light beams. A brief overview of these applications can be found in \cite{PS09}.
This suggests that monotone quantization might play a role in the
quantization of these phenomena of classical physics.
 Moreover, these series have proven useful for solving classes of mixed difference equations (see \cite[p. 530]{Watson}) and differential equations like the one-dimensional Schr\"{o}dinger equation, Sturm-Louiville equation, and perturbed Bessel equation (see \cite{KV}). \\

\noindent  This paper is organized as follows:
Section 2 provides a brief review of the canonical quantum decomposition of arc--sine and semi--circle random variables.
Section 3 revisits the orthogonal polynomials associated with these  probability distributions and their related Hilbert transforms.
Section 4 presents a unified formulation of semi--circle and arc--sine quantum mechanics and derives explicit expressions for  the monotone position and momentum groups $e^{itX}$ and $e^{itP}$.
Section 5 describes monotone evolutions generated by momentum and kinetic energy operators,
along with the  simplest case of a monotone harmonic oscillator.
Section 6 introduces exponential vectors associated with semi--circle and arc--sine measures and calculates the distributions of the number operator in their corresponding coherent states.

We find that, in the semi--circle case, this distribution is geometric, while in the
arc--sine case, it differs from the geometric only in the first and in the normalization term.
In the boson case, the geometric distribution describes equilibrium states.
So, in the latter case, the vacuum distribution of the number operator can be seen as a
non--equilibrium state in which the vector $\Phi_{0}$ plays a special role.
This suggests that arc--sine coherent states may play a role in those situations in which
there is a vector state playing a special role (as it happens in the Bose--Einstein condensation,
where this state is the ground state).
\\
The final section proves that for both measures, the Lie algebra generated by canonical creation and annihilation operators is equivalent to the $*$--algebra generated by all rank-one operators in corresponding $L^2$-spaces. The paper concludes with two appendices: Appendix A computes vacuum distributions of position operators in semi--circle and arc--sine cases using the \textit{quantum moments technique} developed in \cite{[AcLu16d-q-mom-prob]}
for symmetric classical random variables. In the case considered here, both results are implicit
in the quantum decomposition of the semi--circle and arc--sine random variables (see Sections
\ref{sec:The-arc--sine-case} and \ref{sec:The-sem-c-case}). However the technique used
here is applicable to the whole semi--circle--arc--sine class which, as shown in
\cite{[AccBarRha17Info-Compl]}, is characterized by having information complexity index equal to
$(0,K)$ for all $K\in\mathbb{N}$. For $k\le 3$, the distributions in this class are known to be the
semi--circle, the arc--sine and the Konno distributions. But, for $k\ge 4$, the distributions, hence
the classical moments, are not known and this technique allows to calculate them.
(This statement will be proved elsewhere).
 \\
Appendix B discusses  the integral and Neumann series representations of
 the momentum groups in both cases.

\section{Canonical quantum decomposition of the arc--sine and semi--circle random variables}
\label{Can-q-dec-arc-sin-semi-circl-RV}

In this section we recall some known facts, widely used in the following (see
\cite{[AcLu20-Nigeria21]} for details).\\
Let $\omega:=\left\{\omega_{n}\right\}_{n\geq1}$ be a principal Jacobi sequence (i.e.
$\omega_{n}\ge 0$ and $\omega_{n}=0 \Rightarrow \omega_{n+k}=0 \ , \ \forall k\in\mathbb{N}$).
Recall that the $1$--Mode Interacting Fock Space ($1$MIFS) associated to $\omega$ is
the Hilbert space defined by
\begin{equation}\label{df-Gamma-omega}
\Gamma_{\omega}(\mathbb{C})
:=\bigoplus_{n\in\mathbb{N}}\mathbb{C}\cdot\Phi_{n}
\end{equation}
where $\bigoplus_{n\in\mathbb{N}}$ denotes closed orthogonal sum and the orthogonal sequence
$(\Phi_{n})_{n\in\mathbb{N}}$ uniquely defines the scalar product in $\Gamma_{\omega}$ through the
prescription
\begin{equation}
\|\Phi_{n}\|^2 := \omega_{n}!= \prod_{j=1}^{n}\omega_{n}
\end{equation}
where, here and in the following, we use the convention
\begin{equation}\label{omega0=0}
\omega_{0}=0  \qquad;\qquad 0! := 1
\end{equation}
The unique linear extension of the map
\begin{equation}
a^{+}\Phi_{n} := \Phi_{n+1} \quad,\quad \forall n\in\mathbb{N}
\end{equation}
to the linear span of the $(\Phi_{n})_{n\in\mathbb{N}}$  is called \textit{creator} and the
properties of the principal Jacobi sequences imply that, on this domain $a^{+}$ has an adjoint
denoted $a$ (sometimes also $a^{-}$) and called \textit{annihilator}.
By Favard's Lemma, for any gradation preserving operator $a^{0}$, there exists a probability measure $\mu$ on $\mathbb{R}$ with all moments such that the map
$$
\Phi_{n} = a^{+n}\Phi_{0}\in\Gamma_{\omega}(\mathbb{C}) \mapsto  \Phi_{\mu;n}
\in L^2_{\mathbb{C}}(\mathbb{R}, \hbox{Borel}, \mu) =: L^2(\mathbb{R},\mu)
$$
where $\Phi_{\mu;n}$ is the $n$--th \textbf{monic} orthogonal polynomial of $\mu$, extends to an
isometry $U_{\mu}$ satisfying
\begin{equation}\label{q-dec-X}
U_{\mu}(a^{+}+a^{0}+a)U_{\mu}^* = X\colon  L^2(\mathbb{R},\mu)\to  L^2(\mathbb{R},\mu)
\end{equation}
where $X$ is the multiplication operator
$$
(Xf)(x) := xf(x) \quad,\quad x\in \mathbb{R} \ , \ f\in L^2(\mathbb{R},\mu)
\ , \ \int_{\mathbb{R}}x^2f(x)\mu(dx)<+\infty
$$
\begin{remark}
Notice that if $\mu$ is an arc-sine or semicircle measure, then the operator $U_{\mu}$ is a unitary isomorphism. This is because, for such measures, the associated orthogonal polynomials $(\Phi_{\mu;n})_{n \ge 0}$ form a \textbf{complete orthogonal system} in the Hilbert space $L^2([-2,2], \mu)$ (see, e.g., \cite{[MaHa03]}).

\end{remark}

In the following \textbf{we identify} $\Gamma_{\omega}(\mathbb{C}) $ with its image under $U_{\mu}$,
i.e. the closure of the linear span of $(\Phi_{\mu;n})_{n\in\mathbb{N}}$ in $L^2(\mathbb{R},\mu)$,
still denoted $\Gamma_{\omega}(\mathbb{C}) $ (sometimes also denoted $L^2_{pol}(\mathbb{R},\mu)$). With this identification, $\Phi_{n}$ is identified with $\Phi_{\mu;n}$, in particular the vacuum
vector $\Phi_{0}$ is identified with $\Phi_{\mu;0}$, which is the constant function identically
equal to $1$. Moreover \eqref{q-dec-X} becomes
\begin{equation}\label{q-dec-X-id}
a^{+}+a^{0}+a  = X\colon  L^2(\mathbb{R},\mu)\to  L^2(\mathbb{R},\mu)
\end{equation}
and is called the \textbf{quantum decomposition of $X$}. \\
$X$ is a self--adjoint operator on $L^2(\mathbb{R},\mu)$ and $\mu$ is its vacuum spectral measure.
With respect to this measure $X$ can be identified (up to stochastic equivalence) to a real valued random variable with all moments. Conversely, any real valued random variable $X$ with all moments and with probability distribution $\mu$ uniquely determines the $1$MIFS
$\Gamma_{\omega}(\mathbb{C})$, where $(\omega_{n})$ is the principal Jacobi sequence of $X$,
which is related to $L^2(\mathbb{R},\mu)$ in the way described above.
The operators $a^{+},a,a^{0}$ are called the CAP (creation, annihilation, preservation) operators
associated to $\mu$. $\mu$ is moment symmetric (i.e. all odd moment vanish) iff $a^{0}=0$. In this
case, one has
\begin{equation}\label{arcSin01a}
\omega_{1}=\left\langle\Phi_{0}, \left(a+a^{+}\right)^{2}\Phi_{0}\right\rangle
	=\left\langle\Phi_{0}, aa^{+}\Phi_{0}\right\rangle
= \hbox{ variance of $X$}
\end{equation}
The \textbf{number operator} associated to the gradation \eqref{df-Gamma-omega}
is characterized by the property
\begin{equation}\label{df-Lambda-om}
\Lambda\Phi_{n}:= n\Phi_{n}  \quad,\quad\forall n\in\mathbb{N}
\end{equation}

\subsection{The arc--sine case}\label{sec:The-arc--sine-case}

If $\omega_{1}=2\omega_{n}>0$ for any $n\geq2$, the operator $X=a+a^{+}$ has the
\textbf{arc--sine} distribution on the interval
$\left( -\sqrt{2\omega_{1}},\sqrt{2\omega_{1}}\right)$ (see Appendix \ref{sec:Vac-distr-arcsin}), where we recall that for any $c>0,$ the arc--sine
distribution on the interval  $\left(-c,c\right)  $ is
absolutely--continuous with density function
$$x\longmapsto\frac{1}
{\pi\sqrt{c^{2}-x^{2}}}\chi_{\left(-c,c\right)}\left(x\right)$$

 This distribution has clearly mean zero and the variance
$$
\int_{-c}^{c}\frac{x^{2}}{\pi\sqrt{c^{2}-x^{2}}}dx=\frac{2c^{2}}{\pi}\int
_{0}^{1}\frac{y^{2}}{\sqrt{1-y^{2}}}dy=\frac{2c^{2}}{\pi}\int_{0}^{\frac{\pi
}{2}}\sin^{2}\theta d\theta=\frac{c^{2}}{2}%
$$
So the vacuum distribution of $X=a+a^{+}$ is given by the density function
\begin{equation}
x\longmapsto\frac{1}{\pi\sqrt{2\omega_{1}-x^{2}}}\chi_{\left(-\sqrt
{2\omega_{1}},\sqrt{2\omega_{1}}\right)}\left(x\right)  \label{arcSin01b}%
\end{equation}
In particular, if  $\omega_{n}=1$ for any $n\ge 2$ (or equivalently, $\omega_1=2$), the vacuum distribution of $X=a+a^{+}$ is denoted
\begin{equation*}
	\mu_{as}(dx)=\frac{1}{\pi\sqrt{4-x^{2}}}\chi_{\left(-2,2\right)}\left(x\right)dx
\end{equation*}

\subsection{The semi--circle case}\label{sec:The-sem-c-case}
If $\omega_{n}=\omega_{1}$ for any $n\in\mathbb{N}^{\ast}$, the operator $X=a+a^{+}$ has the
\textbf{semi--circle} distribution on the interval $\left( -2\sqrt{\omega_{1}},2\sqrt{\omega_{1}}\right)$ (see Appendix \ref{sec:Vac-distr-arcsin}), where we recall that for any $c>0,$ the semi--circle
distribution on the interval $\left(-c,c\right)$ is absolutely--continuous
with the density function
$$x\longmapsto\frac{2}{\pi c^{2}}\sqrt{c^{2}-x^{2}%
}\chi_{\left(-c,c\right)}\left(x\right)
$$
 This distribution has clearly mean zero and variance
\begin{align*}
\int_{-c}^{c}\frac{x^{2}}{\pi c^{2}}\sqrt{c^{2}-x^{2}}dx  & =\frac{2c^{2}}%
{\pi}\int_{0}^{1}y^{2}\sqrt{1-y^{2}}dy=\frac{2c^{2}}{\pi}\int_{0}^{\frac{\pi
}{2}}\sin^{2}\theta\cos^{2}\theta d\theta\\
& =\frac{2c^{2}}{\pi}\left(\frac{\pi}{4}-\frac{3}{2^{2}\cdot2!}\cdot
\frac{\pi}{2}\right)  =\frac{c^{2}}{4}%
\end{align*}
So the vacuum distribution of $X=a+a^{+}$ is given by the density function
\begin{equation}
x\longmapsto\frac{1}{2\pi\omega_{1}}\sqrt{4\omega_{1}-x^{2}}\chi_{\left(
-2\sqrt{\omega_{1}},2\sqrt{\omega_{1}}\right)}\left(x\right)
\label{arcSin01c}%
\end{equation}
In particular, if  $\omega_{n}=1$ for any $n\in\mathbb{N}^{\ast}$, the vacuum distribution of $X=a+a^{+}$ is denoted
\begin{equation*}
\mu_{sc}(dx)=\frac{1}{2\pi}\sqrt{4-x^{2}}\chi_{\left(
	-2,2\right)}\left(x\right) dx.
\end{equation*}

\section{Chebyshev polynomials and associated Hilbert transforms}

In this section, we recall the orthogonal polynomials associated to the arc--sine and semi--circle distributions and the relationship between their associated Hilbert transforms.

\subsection{Chebyshev polynomials}

Let $(T_n)_{n\in\mathbb{N}}$ and $(\Phi_n)_{n\in\mathbb{N}}$ denote the monic Chebyshev polynomials of first and second kind associated respectively to  $\mu_{as}$ and $\mu_{sc}$ and given by (see, e.g., \cite{[AcHamLu22a]}):
\begin{equation}\label{df-Tn}
	T_n(x)=\begin{cases}
		2\cos\big(n\arccos{(x/2)}\big),& {\rm if\ } n\ge1\\ 1,& {\rm if\ } n=0
	\end{cases}{\,,\qquad \forall x\in [-2,2]}
\end{equation}
and
\begin{equation*}
	\Phi_n(x)=\frac{\sin\big((n+1)\arccos{(x/2)} \big)}{\sin\big(\arccos{(x/2)}\big)}
	{\,,\qquad \forall x\in [-2,2]}
\end{equation*}

The following relation between these polynomials is known, but we will provide a direct proof for the reader's convenience later.
\begin{lemma}\label{Chebyshev-relation}
For any $n\in\mathbb{N}$ and $x\in[-2,2]$, we have
\begin{align}\label{1s2dCheby00}
-(4-x^2)\Phi_n(x)=\begin{cases}
T_{n+2}(x)-T_n(x), &{\rm if\ } n\ge1\\
T_2(x)-2T_0(x), &{\rm if\ } n=0     \end{cases}
\end{align}

\begin{align}\label{1s2dCheby03}
(x^2-4)T_n(x)=\begin{cases}T_{n+2}(x)-2T_n (x)+2T_{n-2}(x), &\text{ if }n=2\\
T_{n+2}(x)-2T_n (x)+T_{n-2}(x),&\text{ if } n>2
\end{cases}
\end{align}
and
\begin{align}\label{connection2}
T_{n+1}(x)=\Phi_{n+1}(x)-\Phi_{n-1}(x),\quad\forall n\ge1
\end{align}
\end{lemma}
\begin{proof}
Since
$$\cos((n+2)\theta)=\cos((n+1)\theta+\theta)=\cos((n+1)\theta)\cos(\theta)- \sin((n+1)\theta)\sin(\theta)
$$
\begin{align*}
\cos(n\theta)&=\cos((n+1)\theta-\theta)=\cos((n+1)\theta)\cos(\theta)+\sin((n+1)\theta)\sin(\theta)
\end{align*}
subtracting, one gets
\begin{align}\label{1s2dCheby01}
\cos\big((n+2)\theta\big)-	\cos\big(n\theta\big)=-2\sin(\theta)	\sin\big((n+1)\theta\big)
\end{align}
For any $x\in[-2,2]$, by denoting $\theta:=\arccos(x/2)$, one gets
\begin{align}\label{1s2dCheby02} \sin^2(\arccos(x/2))=\sin^2\theta=1-\cos^2\theta
=1-{x^2\over4}
\end{align}
So, recalling that $\Phi_0=T_0=1$, the equality \eqref{1s2dCheby00} for $n=0$ is in fact
the following trivial equality:
$$
-4\sin^2\theta=2\cos(2\theta)-2
$$
Finally, for any $x\in[-2,2]$ and $n\ge1$,
\begin{align*}
T_{n+2}(x)=&2\cos((n+2)\arccos(x/2))\\
\overset{\eqref {1s2dCheby01}}=&2\cos(n\arccos(x/2)) -4\sin(\arccos(x/2)) \sin((n+1)\arccos(x/2))\\
=&T_{n}(x)-4\sin^2(\arccos(x/2)){\sin ((n+1)\arccos(x/2))\over\sin(\arccos(x/2))}\\
\overset{\eqref{1s2dCheby02}}=&T_{n}(x)-(4-x^2)
\Phi_n(x)
\end{align*}

Now we turn to prove the formula \eqref{1s2dCheby03}.
The notation $\theta:=\arccos(x/2)$ and the definition of $T_n$'s guarantee that
the expression in the right hand side of the equality in \eqref{1s2dCheby03} rewrites
\begin{align*}
&2\cos((n+2)\theta)-4\cos(n\theta)+2\cos((n-2)\theta)\\
&=2\Big(\cos(n\theta)\cos(2\theta)-\sin(n\theta) \sin(2\theta)-2\cos(n\theta)+\cos(n\theta)\cos(2\theta)+ \sin(n\theta)\sin(2\theta)\Big)\\
&=4\cos(n\theta)\Big(\cos(2\theta)-1\Big)=4\cos(n\theta)\Big(\cos^2(\theta)-\sin^2(\theta)-\cos^2(\theta)-\sin^2(\theta)\Big)\\
&=-8\sin^2(\theta)\cos(n\theta)
\end{align*}
which is nothing else than the expression in the left hand side of the equality in \eqref{1s2dCheby03} thanks to the formula \eqref{1s2dCheby02}.

Finally,  \eqref{connection2} is proved as follows: for any $n\ge1$,
\begin{align*}
&\Phi_{n+1}(x)-\Phi_{n-1}(x)={1\over \sin\theta}\Big( \sin((n+2)\theta)-\sin(n\theta)\Big)\\
&={1\over \sin\theta}\Big(\sin(n\theta)\cos(2\theta)+
\sin(2\theta)\cos(n\theta) -\sin(n\theta)\Big)\\
&={1\over \sin\theta} \Big(\sin(n\theta) \big(\cos^2(\theta)-\sin^2(\theta)\big)+
2\sin(\theta)\cos(\theta)\cos(n\theta) -\sin(n\theta) \big(\cos^2(\theta)+\sin^2(\theta)\big)\Big)\\
&=2\Big(\cos(\theta)\cos(n\theta)-\sin(\theta)\sin(n\theta)\Big)=2\cos((n+1)\theta)=T_{n+1}(x)
\end{align*}
%

\end{proof}

\subsection{$\mu_{as}$-Hilbert transform and $\mu_{sc}$-Hilbert transform}

Denote
\begin{align}\label{Hilbert-transforms}
H_{as} f(x)&:= \hbox{p.v.}\int_{-2}^2\frac{f(y)}{x-y} \mu_{as}(dy)=\hbox{p.v.}\int_{-2}^2\frac{f(y)}{\pi\sqrt{4-y^2}(x-y)} dy	 \notag\\
H_{sc} f(x)&:=2\hbox{p.v.}\int_{-2}^2\frac{f(y)}{x-y} \mu_{sc}(dy)=\hbox{p.v.}\int_{-2}^2
\frac{f(y)\sqrt{4-y^2}}{\pi(x-y)} dy
\end{align}
the Hilbert transforms with respect to the arc--sine and semi--circle distributions respectively on  $L^2([-2,2],\mu_{as})$ and $L^2([-2,2],\mu_{sc})$.
According to  \cite[Corollary 2.3]{[APS96]}, $H_{sc}$ is  surjective
on $L^2([-2,2],{\mu_{sc}})$ with a dense image, whereas $H_{as}$ is bounded-below injective operator on $L^2([-2,2],{\mu_{as}})$. \\
Moreover, according to \cite{[AcHamLu22a]}, $H_{sc}$ is  skew-adjoint (bounded) operator on $L^2([-2,2],{\mu_{sc}})$. 
We shall prove in this paragraph that for the Hilbert transform $H_{as}$, all polynomials $(T_n)_{n \ge 0}$ are in its domain. Furthermore, up to multiplication by a bounded support function, the operator $H_{as}$ is bounded. Additionally, a rank-one perturbation of the latter, results in a skew-adjoint operator on $L^2([-2,2], \mu_{as})$.
Notice that, since the Radon--Nikodym derivative
$$
\frac{d\mu_{sc}}{d\mu_{as}}(x) = \frac{1}{2}(4-x^2)
$$
is bounded, $L^2([-2,2],\mu_{as})\subset L^2([-2,2],\mu_{sc})$, where the inclusion is understood 
in the set--theoretical sense.
\begin{lemma}\label{Hilbert-Chebychev01}
For any $n\in\mathbb{N}$,
\begin{equation}\label{Hilbert-Chebychev}
	H_{sc} \Phi_n=T_{n+1} \ , \qquad  	H_{as} T_{n}=-\Phi_{n-1}
\end{equation}
with the convention $\Phi_{-1}=0$.  Moreover, on $L^2([-2,2],\mu_{as})$
\begin{align}\label{norm01}
\Vert T_j\Vert_2^2=\begin{cases}
1,&j=0\\  2,&j\ge1
\end{cases}
\end{align}
and on $L^2([-2,2],\mu_{sc})$
\begin{align}\label{norm02}
\Vert \Phi_j\Vert_2^2=1
\end{align}
\end{lemma}

\begin{proof}
According to \cite[p.180]{[Tricomi57]}, we have for any $n\in\mathbb{N}$,
\begin{equation*}
\hbox{p.v.}	\int_{0}^{\pi}\frac{\sin((n+1) \theta) \sin( \theta)}  {\cos\theta-\cos\varphi}d\theta=-\pi\cos((n+1) \varphi)
\end{equation*}
and
\begin{equation*}
	\hbox{p.v.}\int_{0}^{\pi}\frac{\cos(n \theta)}{\cos\theta-\cos\varphi}d\theta=\pi\frac{\sin(n \varphi)}{\sin( \varphi)}
\end{equation*}
Performing the variables change $\theta\mapsto \arccos \frac{y}{2}$ and $\varphi\mapsto \arccos \frac{x}{2}$, one gets
\begin{equation*}
\hbox{p.v.}	\int_{-2}^{2}\frac{\sin\big((n+1) \arccos \frac{y}{2}\big)}{y-x}dy=-\pi\cos\big((n+1) \arccos \frac{x}{2}\big)
\end{equation*}
and
\begin{equation*}
	\hbox{p.v.}\int_{-2}^{2}\frac{(1-\frac{y^2}{4})^{-\frac{1}{2}}\cos(n \arccos \frac{y}{2})}{y-x}dy=\pi\frac{\sin(n \arccos \frac{x}{2})}{\sin( \arccos \frac{x}{2})}
\end{equation*}
Then, using the definitions of $H_{as}$ and $H_{sc}$ from \eqref{Hilbert-transforms} one finds \eqref{Hilbert-Chebychev}.
Next, thanks to the fact $\sin(n\pi)=0$ for all $n\in{\mathbb Z}$ and the formula of integration by parts, one gets, for any $n\ge1$,
\begin{align*}
\int_0^\pi\cos^2(n\theta)d\theta={-1\over n} \sin(n\theta)\cos(n\theta)\Big\vert_0^\pi+ \int_0^\pi\sin^2(n\theta)d\theta=\int_0^\pi\sin^2(n\theta)d\theta
\end{align*}
and so
\begin{align}\label{norm01a}
\int_0^\pi\cos^2(n\theta)d\theta
=\int_0^\pi\sin^2(n\theta)d\theta ={\pi\over2},\quad \forall n\ge1
\end{align}

On $L^2([-2,2],\mu_{as})$ (and $L^2([-2,2],\mu_{sc})$ respectively), it is clear that
$$\Vert T_n\Vert_2^2=\int_{-2}^2 T_n^2(x){1\over\pi\sqrt{4-x^2} }dx;\quad
\Vert \Phi_n\Vert_2^2=\int_{-2}^2 \Phi_n^2(x){\sqrt{4-x^2}\over2\pi }dx$$
and so one gets trivially $\Vert T_0\Vert_2^2=\Vert \Phi_0\Vert_2^2=1$. For any $n\ge1$,
\begin{align*}
\Vert T_n\Vert_2^2&=4\int_{-2}^2 \cos^2(n\arccos(x/2)) {1\over\pi\sqrt{4-x^2}}dx\notag\\
&\overset{[\theta:=\arccos(x/2)]} = {4\over\pi}\int_\pi^0{\cos^2(n\theta)\over \sin(\theta)}d\cos(\theta)\\
&={4\over\pi}\int_0^\pi\cos^2(n\theta) d\theta\overset{\eqref{norm01a}} =2
\end{align*}
Similarly, on $L^2([-2,2],\mu_{sc})$,
\begin{align*}
\Vert \Phi_n\Vert_2^2 &=\int_{-2}^2\frac{\sin^2\big((n+1) \arccos{(x/2)} \big)}{\sin^2\big(\arccos{(x/2)}\big)} \cdot {\sqrt{4-x^2}\over2\pi}dx\notag\\
&\overset{[\theta:=\arccos(x/2)]} = {2\over\pi}\int_\pi^0{\sin^2((n+1)\theta)\over \sin^2(\theta)}\sin(\theta)d\cos(\theta)\notag\\
&={2\over\pi}\int_0^{\pi}\sin^2((n+1)\theta) d\theta\overset{\eqref{norm01a}} =1
\end{align*}

\end{proof}
\noindent We shall now state the connection between the Hilbert transforms $H_{sc}$ and $H_{as}$,
which will be of use in our subsequent analysis.
\begin{proposition}\label{prop:connection}
The identity
\begin{equation}\label{Hilbert-relationship}
(4-x^2)H_{as} f (x)+T_{1,0}f(x)
=H_{sc} f(x) - 2T_{0,1}f(x)
\end{equation}
holds everywhere on $(-2,2)$ for every $f\in L^2([-2,2],\mu_{as})$, where {for any $\{i,j\}\in\{0,1\}$, $T_{i,j}$} denotes the operator in $L^2([-2,2],\mu_{as})$ given by
\begin{equation*}
T_{i,j}(f)(x):= \frac{1}{||T_j||_2^2}\langle T_j, f\rangle_{L^2([-2,2],\mu_{as})} T_i(x)
= \left(\frac{1}{||T_j||_2^2}T_iT_j^*f\right) (x)
\end{equation*}
\end{proposition}
\begin{remark}
Before going through the proof, we point out that the equality \eqref{Hilbert-relationship} is well defined since
\begin{equation*}
	L^2([-2,2],\mu_{as}) \subset L^2([-2,2],\mu_{sc})
\end{equation*}
We notice also that  $T_{i,i}$ is the self-adjoint projection onto $\{T_i\}$ in $L^2([-2,2],\mu_{as})$.
 If $i\ne j$, the definition of $T_{i,j}$'s gives that $T_{i,j}T_n=\frac{1}{||T_j||_2^2}\langle T_j, T_n\rangle T_i=\delta_{j,n}T_i$. Consequently,
\begin{align*}
\langle T_{i,j}T_n, T_m\rangle&=\delta_{j,n}\langle T_i, T_m\rangle=\delta_{j,n}\delta_{i,m}\Vert T_i\Vert^2_2\\
\langle T_n, T_{j,i}T_m\rangle&=\delta_{i,m}\langle T_n, T_j\rangle=\delta_{j,n}\delta_{i,m}\Vert T_j\Vert^2_2
\end{align*}
Therefore
\[T_{i,j}=\frac{\Vert T_i\Vert^2_2}{\Vert T_j\Vert^2_2}T_{j,i}^*\]
In particular, the identity \eqref{Hilbert-relationship}, rewrites
\begin{equation}\label{Hilbert-relationship1}
H_{sc}=2\frac{d\mu_{sc}}{d\mu_{as}} H_{as} + (T_{1,0} + T_{1,0}^*)
\end{equation}
\end{remark}

\begin{proof}[Proof of proposition \ref{prop:connection}]
The  equality \eqref{Hilbert-relationship} is a straightforward consequence of the following identity (see, e.g., \cite[p. 179]{[Tricomi57]}):
\begin{equation}\label{integral-rep}
\int_{-2}^2\left[\sqrt{\frac{4-x^2}{4-y^2}}-\sqrt{\frac{4-y^2}{4-x^2}}\right]\frac{f(y)}{x-y}dy=\frac{-1}{\sqrt{4-x^2}}\int_{-2}^2\frac{x+y}{\sqrt{4-y^2}}f(y)dy
\end{equation}
Notice that,
\begin{equation*}
	\int_{-2}^2\frac{xf(y)}{\pi\sqrt{4-y^2}}dy=T_{1,0}(f)(x)
\end{equation*}
and
\begin{equation*}
	\int_{-2}^2\frac{yf(y)}{\pi\sqrt{4-y^2}}dy=2T_{0,1}(f)(x).
\end{equation*}
Now, in order to show that the formula \eqref{integral-rep} implies the formula \eqref{Hilbert-relationship}, we first multiply both sides of \eqref{integral-rep} by $1/\pi$. Then we use the definitions of $H_{as}$ and $H_{sc}$, which are given by \eqref{Hilbert-transforms}, to rewrite the right-hand side of \eqref{integral-rep} as
\begin{align*}
	\int_{-2}^2\left[\sqrt{\frac{4-x^2}{4-y^2}}-\sqrt{\frac{4-y^2}{4-x^2}}\right]\frac{f(y)}{x-y}dy=
	{\pi\over \sqrt{4-x^2}}\Big((4-x^2)(H_{as}f)(x) -(H_{sc}f)(x)\Big).
\end{align*}
Next, we recall that the left-hand side of this equation is equal to
$${-\pi\over \sqrt{4-x^2}}\Big((T_{1,0}f)(x)-2(T_{0,1}f)(x)\Big)$$
as shown in the previous step. Therefore, we can conclude that
\begin{align*}
{-\pi\over \sqrt{4-x^2}}\Big((T_{1,0}f)(x)-2(T_{0,1}f)(x) \Big)= {\pi\over \sqrt{4-x^2}}\Big((4-x^2)(H_{as}f)(x) -(H_{sc}f)(x)\Big),
\end{align*}
which is exactly formula \eqref{Hilbert-relationship}. This completes the proof.
\end{proof}


By definition, the scalar product on the Hilbert space $L^2([-2,2],\mu_{as})$ is
\begin{equation*}
\langle f,g\rangle_{L^2([-2,2],\mu_{as})}=\int_{-2}^2f(t)g(t)\mu_{as}(dt)
=\frac{1}{\pi}\int_{-2}^2\frac{f(t)g(t)}{\sqrt{4-t^2}}dt
\end{equation*}
Since $T_0\equiv1\in L^2([-2,2],\mu_{as})$, it follows that
\begin{equation*}
	L^2([-2,2],\mu_{as})=\{T_0\}^{\perp}\oplus \{T_0\}
\end{equation*}
where
\begin{equation*}
	\{T_0\}^{\perp}:=\{f\in L^2([-2,2],\mu_{as}): \langle f,T_0\rangle_{L^2([-2,2],\mu_{as})}=0\}
\end{equation*}
Now, recall from \cite{[MaHa03]} that $\{T_n, n\ge0\}$ forms a complete orthogonal system in $L^2([-2,2],\mu_{as})$. Then, any $f\in L^2([-2,2],\mu_{as})$ can be expressed as
\begin{equation*}
	f(x)=\sum_{n\ge0} f_nT_n(x)
\end{equation*}
where
\begin{equation*}
	f_n=\frac{1}{||T_n||_2^2}\langle f,T_n\rangle_{L^2([-2,2],\mu_{as})}
\end{equation*}
Let $Q$ denote the operator of multiplication by the coordinate in $L^2([-2,2],\mu_{as})$:
\begin{equation*}
	Q f(x)= x f(x), \quad x\in [-2,2]
\end{equation*}
\begin{theorem}
In the canonical identification of $\Gamma_{\mu_{as}}$ with  $L^2([-2,2],\mu_{as})$  described in Section \ref{Can-q-dec-arc-sin-semi-circl-RV},
\begin{itemize}
\item The position operator $X_{as}{:=a^+_{as}+a_{as}}$ is mapped into the the operator $X$ in $L^2([-2,2],\mu_{as})$ given by:
\begin{equation}\label{position-form}
   X = Q
\end{equation}
\item The momentum operator $P_{as}{:=i(a^+_{as}-a_{as})}$ is mapped into the operator $P$ in $L^2([-2,2],\mu_{as})$ given by:
\begin{equation}\label{momentum1}
  P f(x)= i(4-x^2)H_{as} f (x) +i T_{1,0}f(x)
\end{equation}
Or equivalently,
\begin{equation}\label{momentum2}
 P f(x)= iH_{sc} f(x)-2iT_{0,1}f(x).
\end{equation}
\item The arc--sine CAP operators $a^{\varepsilon}_{\mu_{as}}$ ($\varepsilon \in \{+,-\}$) are mapped into the following CAP operators

	\begin{align*}\label{mu-CAP-ops}
	A^{+} = \frac{1}{2}\big(Q +(4-x^2) H_{as}  +T_{1,0}\big)\\
A^{-}  = \frac{1}{2}\big(Q -(4-x^2) H_{as}  -T_{1,0}\big)
\end{align*}
\item The arc--sine kinetic energy operator $E_{as}$ is mapped into the operator $E{:=\frac{1}{2}P^2}$
in $L^2([-2,2],\mu_{as})$ given by:

	\begin{equation*}\label{us-kin-energ--op1}
Ef(x) =\frac{1}{2}(4-x^2) f(x)+T_{1,1}f(x)
\end{equation*}
\end{itemize}

\end{theorem}

\begin{proof}

Using the equalities
\begin{equation}\label{creation}
a^+_{\mu_{as}}T_{n}=T_{n+1},\qquad \forall n\in{\mathbb N}
\end{equation}
and
\begin{equation}\label{annihilation}
a_{\mu_{as}}T_{n}=\omega_{n}T_{n-1}=\begin{cases}
T_{n-1} &, n\ge2\\
2T_0&, n=1\\
0&, n=0
\end{cases}
\end{equation}

As a consequence, its conjugate equals to the annihilation operator $a$ introduced in \eqref{annihilation}:
\begin{align*}
&\big\langle T_n,a^+ T_0\big\rangle=\big\langle T_n, T_1\big\rangle=\delta_{n,1}\Vert T_1\Vert^2=2\delta_{n,1}, \qquad \forall n\in{\mathbb N}\\
&\big\langle aT_n,T_0\big\rangle=\omega_n\big\langle T_{n-1}, T_0\big\rangle=\delta_{n-1,0}\omega_{n}\Vert T_0\Vert^2=\omega_1\delta_{n,1}=2\delta_{n,1}, \qquad \forall n\in{\mathbb N}\\
\end{align*}	
and for any $m\ge 1$ and $n\in{\mathbb N}$,
\begin{align*}
&\big\langle T_n,a^+ T_m\big\rangle=\big\langle T_n, T_{m+1}\big\rangle=\delta_{n,m+1}\Vert T_{m+1}\Vert^2=2\delta_{n,m+1}\\
&\big\langle aT_n,T_m\big\rangle=\omega_n\big\langle T_{n-1}, T_m\big\rangle=\delta_{n-1,m}\omega_{n}\Vert T_m\Vert^2\overset{m\ge1}=2\omega_{m+1}\delta_{n,m+1}\overset{m\ge1}=2\delta_{n,m+1}\\
\end{align*}
 Now we deduce that $\big((a^++a)T_n\big)(x)=xT_n(x)$: Thanks to the facts $T_1(x)=x=xT_0(x)$ and $T_2(x)=x^2-2=xT_1(x)-2T_0(x)$
\begin{align*}
\big((a^++a)T_n\big)(x)=T_{n+1}(x)+\omega_nT_{n-1}
&=\begin{cases}
T_{n+1}(x)+T_{n-1}(x) &, n\ge2\\ T_{2}(x)+2T_0(x)&, n=1\\
xT_0(x)&, n=0
\end{cases}\\
&=\begin{cases}
		T_{n+1}(x)+T_{n-1}(x) &, n\ge2\\ xT_{1}(x)&, n=1\\
		xT_0(x)&, n=0
	\end{cases}
\end{align*}
By using the Jacobi $3$--diagonal relation for the monic Chebyshev polynomials of first kinds:
\begin{equation*}
	x T_{n+1}(x)=T_{n+2}(x)+T_{n}(x), \quad \forall n\ge1, x\in [-2,2]
\end{equation*}
the above equality becomes to
\[\big((a^++a)T_n\big)(x)=\begin{cases}
	xT_{n}(x) &, n\ge2\\ xT_{1}(x)&, n=1\\
	xT_0(x)&, n=0
\end{cases}=xT_n(x)\]
%
%
%

Next, we have
\begin{equation*}
P_{as}T_{n}=i(a^+_{\mu_{as}}-a_{\mu_{as}})T_{n}
=iT_{n+1}-i\omega_nT_{n-1}=\begin{cases}
		i(T_{n+1}-T_{n-1}) &, n\ge2\\
		i(T_{2}-2T_{0})&, n=1\\
			iT_{1}&, n=0
	\end{cases}
\end{equation*}
Which reduces, by the formula \eqref{1s2dCheby00},
to
\begin{equation*}
P_{as}T_{n}= 	\begin{cases}
-i(4-x^2)\Phi_{n-1}  &, n\ge1\\
	iT_{1}&, n=0\end{cases} \overset{\eqref{Hilbert-Chebychev}}=\begin{cases}
i(4-x^2)H_{as}T_n  &, n\ge1\\ iT_{1}&, n=0
\end{cases}
\end{equation*}
Hence,  the denseness of $(T_n)_{n\in\mathbb{N}}$ in $L^2([-2,2],\mu_{as})$  imply that
\begin{equation*}
P_{as}=i(4-x^2)H_{as}  +i T_{1,0}
\end{equation*}
i.e., the formula \eqref{momentum1} . The equality \eqref{momentum2} is a direct consequence of the identity \eqref{Hilbert-relationship}.
Next, for the identification of CAP operators, we use the following equalities:
\begin{equation*}
A^{+}=\frac{1}{2}(X -iP)\qquad;\qquad A^{-}=\frac{1}{2} (X +iP)
\end{equation*}  
Finally, the identification of the kinetic energy operator follows from its action on the $T_n$'s since we have

\begin{align*}
P_{as}^2T_n=&-\left[(a^+_{\mu_{as}})^2-a^+_{\mu_{as}}a_{\mu_{as}}-a_{\mu_{as}}a^+_{\mu_{as}}+a^2_{\mu_{as}}\right]T_n\\
=&\begin{cases}
-(T_{n+2}-2T_n+T_{n-2}),& n\ge3\\
-(T_{4}-2T_2+2T_{0}),& n=2\\
 -(T_{3}-3T_1), & n=1\\
 -(T_2-2T_0),&n=0
\end{cases}
 \\=&\begin{cases}  (4-x^2)T_n,& n\ge2\\
-(T_{3}-3T_1), & n=1\\
-(T_2-2T_0),&n=0
\end{cases}  \\
=&(4-x^2)T_n+
\begin{cases}
	2T_1,&n=1\\
	0,& otherwise
\end{cases}.
\end{align*}
where, the last the last equality holds due to the formula \eqref{1s2dCheby03}. Furthermore,
  we have
  \[-(T_2-2T_0)(x)=(4-x^2)T_0(x)\]
   by a simple calculation. Also,
 we can easily verify that $T_3(x)=x^3-3x$ and then \[-(T_{3}-2T_1)(x)=5x-x^3=(5-x^2)T_1(x)=(4-x^2)T_1(x)+T_1(x)\]
Hence
\[P_{as}^2T_n(x)=(4-x^2+\delta_{n,1})T_n(x)\]

This completes the proof.
\end{proof}

\begin{remark}
Using \eqref{momentum1}, the momentum $P_{as}$  operator may also expressed in integral form as follows:
 \begin{align*}P_{as} f(x)=& i\ \hbox{\rm p.v.}\int_{-2}^2\frac{(4-x^2)f(t)}{\pi(x-t) \sqrt{4-t^2}}dt+i\int_{-2}^2\frac{xf(t)}{\pi \sqrt{4-t^2}}dt\\
\overset{\eqref{integral-rep}}{=}&i\ \hbox{\rm p.v.}\int_{-2}^2
\frac{\sqrt{4-t^2}}{\pi(x-t)}f(t)dt-i\int_{-2}^2
\frac{tf(t)}{\pi\sqrt{4-t^2}}dt
\end{align*}
\end{remark}


\section{Unification of the monotone and free quantum mechanics}

\noindent In this section, we demonstrate the close relationship between the quantum mechanics
associated with the arc--sine and semi--circle distributions. To do this, we will first establish
a connection between their CAP operators. We will then derive a unified expression for
their evolutions generated by position and momentum operators.

\subsection{Relationship between monotone and free CAP operators}

Our main result here is the following.
\begin{theorem}
	The monotone and free CAP operators are related  on $L^2([-2,2],\mu_{as})$ via:

\begin{align}\label{CAP-relationship}
		\begin{cases}
			a_{\mu_{as}}^+=a_{\mu_{sc}}^++T_{1,0} \\
			a_{\mu_{as}}=a_{\mu_{sc}}+T_{0,1}
		\end{cases}
	\end{align}

In particular, the arc--sine position $X_{as}$ and momentum $P_{as}$ are related
to the free position and semi--circle momentum respectively on
$L^2([-2,2],\mu_{as})$ via:

\begin{equation}\label{Position-relationship}
	X_{as}=X_{sc}+T_{1,0}+T_{0,1}
\end{equation}
\begin{equation}\label{Momentum-relationship}
	P_{as}=P_{sc}+i (T_{1,0}-T_{0,1})
\end{equation}

\end{theorem}
\begin{proof}
It suffices to prove that the CAP operators coincide on the $(T_n)$'s.\\
It is clear that
\begin{equation*}
a^+_{\mu_{sc}}T_{0}=a^+_{\mu_{sc}}\Phi_{0}	=\Phi_1=T_1
\end{equation*}
\begin{equation*}
a^+_{\mu_{sc}}T_{1}=a^+_{\mu_{sc}}\Phi_{1}	=\Phi_2=T_2+T_0
\end{equation*}
and for any $n\ge2$,
\begin{equation*}
a^+_{\mu_{sc}}T_{n}\overset{\eqref{connection2}}=
a^+_{\mu_{sc}}(\Phi_{n}-\Phi_{n-2})=\Phi_{n+1}-
\Phi_{n-1}=T_{n+1}
\end{equation*}
which, by \eqref{creation}, coincides with
\begin{align*}
(a^+_{\mu_{as}}+T_{0,1})T_{n}
=&\begin{cases} T_{n+1} &, n\ge2\\
T_2+T_0&,n=1\\ T_1&, n=0  \end{cases}
\\=&T_{n+1}+\delta_{1,n}T_0
\end{align*}
Similarly, for the annihilation operators, we have
\begin{equation*}
a_{\mu_{sc}}T_{0}=a_{\mu_{sc}}\Phi_0=0
\end{equation*}
\begin{equation*}
	a_{\mu_{sc}}T_{1}=a_{\mu_{sc}}\Phi_1=\Phi_0=T_0	
\end{equation*}
\begin{equation*}
	a_{\mu_{sc}}T_{2}=a_{\mu_{sc}}(\Phi_2-\Phi_0)=\Phi_1=T_1	
\end{equation*}
	and for any $n\ge3$, we have
	\begin{equation*}
	a_{\mu_{sc}}T_{n}=a_{\mu_{sc}}(\Phi_{n}-\Phi_{n-2})=\Phi_{n-1}-\Phi_{n-3}=T_{n-1}
	\end{equation*}
which, by \eqref{annihilation}, coincides with
\begin{equation*}
(a_{\mu_{as}}-T_{0,1})T_{n}=\begin{cases}
	T_{n-1} &, n\ge1\\
	0&, n=0
\end{cases}
\end{equation*}
Hence  \eqref{CAP-relationship} is proved.
The identities \eqref{Position-relationship} sand \eqref{Momentum-relationship} are then a
simple consequences of the definitions of $X_{as}=  a_{\mu_{as}}+a_{\mu_{as}}^{+} $ and
$P_{as}=i\left(a_{\mu_{as}}^{+}-a_{\mu_{as}}\right) $.
\end{proof}

\subsection{Unification of the monotone and free evolutions}\label{subsec4-2}

Let $\omega:=\left(\omega_{n}\right)_{n\in\mathbb{N}}$ be a principal Jacobi sequence on the $1$MIFS
$\Gamma_{\omega}\left(\mathbb{C}\right)$, such that $\omega_j=\omega_2$ for all $j \geq 2$.
Let $X:=a+a^{+}$ and  $P:=i(a^{+}-a)$ denote respectively the position and momentum operators on $\Gamma_{\omega}\left(\mathbb{C}\right)$. Denote $\mu$ their vacuum distribution and $\{\Psi_n\}_{n\in\mathbb{N}}$ the orthogonal polynomials of $\mu$.
For the rest of this paper, we use the notation $\chi_{c)}(x)$ for any $c,x\in{\mathbb R}$ to denote the following function:
\[\chi_{c)}(x):=\begin{cases}1,&\text{ if }x<c\\0,&\text{ otherwise}\end{cases}\]

In this paragraph, we want to prove the following theorem, which holds for the arc--sine ($\omega_{1}=2\omega_{2}>0$) and the semi--circle  ($\omega_{1}=\omega_{2}>0$) cases, among others.
\begin{theorem}\label{unification}
The action on the $\mu$--orthogonal polynomials of the evolutions $e^{itX}$ and $e^{itP}$ are respectively given by
\begin{align*}
e^{itX}\Psi_h =&\sum_{p,m_{+},m_{-}\geq0}\frac{\left (it\right)^{m_{+}+m_{-}+2p}}{\left(m_{+}+m_{-}+2p\right)!}\omega_{2}^{p}\left(a^{+}\right)^{m_{+}}a^{m_{-}}
\Psi_{h}\\
&\sum_{\substack{p_0,p_1,\ldots,p_{m_{+}+m_{-}}\geq0\\
p_{0}+p_{1}+\ldots+p_{m_{+}+m_{-}}=p}} \left({\prod_{\substack{0\leq r\leq m_{+}+m_{-}}}}C_{p_r}\right)
\left(\delta_{h,m_-} \frac{\mathbf{E}\left(\xi^{2p_{m_{+}}}\right)}{\omega_
	{2}^{p_{m_{+}}}C_{p_{m_{+}}}}+\chi_{h)}(m_-)\right)
\end{align*}	
and
\begin{align*} e^{itP}\Psi_h=&\sum_{p,m_{+},m_{-}\geq0} \frac{\left( -1\right)^{p+m_+}  t^{m_{+}+m_{-}+2p}} {\left(m_{+}+m_{-}+2p\right) !}\omega_{2}^{p} \left(a^{+}\right)^{m_{+}}a^{m_{-}}\Psi_h \\
&\sum_{\substack{p_0,p_1,\ldots,p_{m_{+}+m_{-}}\geq0\\
p_0+p_1+\ldots+p_{m_{+}+m_{-}}=p}} \left({\prod_{\substack{0\leq r\leq m_{+}+m_{-}}}}C_{p_{r}}\right)
\left(\delta_{h,m_-} \frac{\mathbf{E}\left(\xi^{2p_{m_{+}}}\right)}{\omega_
	{2}^{p_{m_{+}}}C_{p_{m_{+}}}}+\chi_{h)}(m_-)\right)
\end{align*}
where $\xi$ is a  real-valued random variable with probability distribution $\mu$.
\end{theorem}

Before presenting the proof of this theorem, let us discuss the key elements that will be used in the proof.
One has formally
\begin{equation}
e^{itX}=\sum_{n=0}^{\infty}\frac{\left(it\right)^{n}}{n!}\sum
_{\varepsilon\in\left\{-1,1\right\}^{n}}a^{\varepsilon\left(1\right)
}\ldots a^{\varepsilon\left(n\right)}\label{arcSin07a}%
\end{equation}
and in particular, for any $h\in\mathbb{N}$%
\begin{equation}
e^{itX}\Psi_{h}=\sum_{n=0}^{\infty}\frac{\left(it\right)^{n}}{n!}%
\sum_{\varepsilon\in\left\{-1,1\right\}^{n}}a^{\varepsilon\left(
1\right)}\ldots a^{\varepsilon\left(n\right)}\Psi_{h}\label{arcSin07b}%
\end{equation}
$e^{itX}$ is uniquely determined by its matrix elements
\begin{equation}\label{arcSin02a}
\left\langle \Psi_{k},e^{itX}\Psi_{h}\right\rangle =\sum_{n=0}^{\infty}%
\frac{\left(it\right)^{n}}{n!}\sum_{\varepsilon\in\left\{-1,1\right\}
^{n}}\left\langle \Psi_{k},a^{\varepsilon\left(1\right)}\ldots
a^{\varepsilon\left(n\right)}\Psi_{h}\right\rangle \quad,\quad k,h\geq0
\end{equation}
The calculation of these matrix elements is done in two steps: first, one reduces the products
$a^{\varepsilon\left(1\right)}\ldots a^{\varepsilon\left(n\right)}$ to their normally ordered form,
i.e. to a product of the form
$\left(a^{+}\right)^{m_{+}}a^{m_{-}}F_{\varepsilon,\omega,n}\left(\Lambda\right)$
where $F_{\varepsilon,\omega,n}\colon \mathbb{N} \to\mathbb{R}$ is a function whose
expression is known and $m_{+}$, $m_{-}$ are natural integers, depending only on
$\varepsilon$ such that $m_{+} + m_{-} \le n$ (Theorem 2.11 of \cite{[AcHamLu22b]} which
holds in general $1$--mode--type interacting Fock spaces). In the special case considered here, this leads to
\begin{align}\label{arcSin07c}
a^{m_{-}}\Psi_h=
\begin{cases}  \displaystyle \prod_{j=0}^{m_{-}-1} \omega_{h-j}\Psi_{h-m_{-}}&; h\ge m_{-}\\
0&; h<m_{-}
\end{cases}= \chi_{h]}\left( m_{-}\right) \prod_{j=0}^{m_{-}-1} \omega_{h-j}\Psi_{h-m_{-}}
\end{align}
Consequently,
\begin{equation*}
\left(a^{+}\right)^{m_{+}}a^{m_{-}}F_{\varepsilon,\omega,n}\left(\Lambda\right)\Psi_{h}=
\begin{cases}
\displaystyle F_{\varepsilon,\omega,n}\left(h\right)
\prod_{j=0}^{m_{-}-1}\omega_{h-j}\Psi_{h+m_{+}-m_{-}}& \ ; h\ge m_{-}\\
0& \ ; h<m_{-}
\end{cases}
\end{equation*}
which can be equivalently formulated as \\
\begin{equation}\label{norm-ord-formula}
\left(a^{+}\right)^{m_{+}}a^{m_{-}}\Psi_{h}=\chi_{h]}\left(m_{-}\right) \Psi_{h+m_{+}-m_{-}} \prod_{j=0}^{m_{-}-1}\omega_{h-j}
\end{equation}
and implies that
\begin{align}\label{arcSin02b}
\left\langle\Psi_{k},\left(a^{+}\right)^{m_{+}}
a^{m_{-}}\Psi_{h}\right\rangle \overset{\eqref{norm-ord-formula}}=&
\chi_{h]}\left( m_{-}\right)  \left \langle\Psi_{k}, \Psi_{h+m_+-m_-} \right\rangle  \prod_{j=0}^{m_{-}-1}\omega_{h-j} \\
\overset{\eqref{arcSin07c}}=&  \chi_{h]}\left(m_{-}\right) \delta_{k-m_{+},h-m_{-}} \Big(\prod_{j=0}^{m_{-}-1} \omega_{h-j}\Big) \Big(\prod_{l=1}^k \omega_{l}\Big)  \notag
\end{align}
Notice that,  \eqref{arcSin02b} holds for any principal Jacobi sequence with
associated orthogonal polynomials $(\Psi_{k})$.
In the second step, one sums over all $\varepsilon$ corresponding to products whose normally ordered form coincide. This requires the solution of the inverse normal order theorem which was given in
Corollary 3.3 of \cite{[AcHamLu22b]} for the semi--circle law. In the rest of this paper we extend
this solution to the arc--sine law.

The inverse normal order theorem states that
$a^{\varepsilon\left(1\right)}\ldots a^{\varepsilon\left(n\right)}$ has the normally ordered form \eqref{norm-ord-formula} if and only
if there is $p\geq0,$ such that $\varepsilon$ belongs to the set
$\Theta_{m_{+}+m_{-}+2p}(m_{+},m_{-})$ defined as follows (see formula (3.12) of \cite{[AcHamLu22b]}):
the totality of $\varepsilon\in\{-,+\}^{m_{+}+m_{-}+2p}$ such that there exists $0=:j_{0}<j_{1}<\ldots <j_{m_{+}+m_{-}}<j_{m_{+}+m_{-}+1}:=m_{+}+m_{-}+2p+1$ and

$\bullet$ $\varepsilon(j_{h})=%
\begin{cases}
+1, & \hbox{if }h\in\{1,\ldots,m_{+}\}\\
-1, & \hbox{if }h\in\{m_{+}+1,\ldots,m_{+}+m_{-}\}
\end{cases}
$

$\bullet$ for any $h\in\{0,1,\ldots,m_{+}+m_{-}\}$, the cardinality of the set
$$(j_{h},j_{h+1}):=\left\{j_{h}+1,j_{h}+2,\ldots,j_{h+1}-1\right\}  $$
is even and $\varepsilon_{h}:=$ the restriction of $\varepsilon$ to $\ (j_{h},j_{h+1})\ $ is balanced
(in the sense of Definition 2.3 of \cite{[AcHamLu22b]});

$\bullet$ $\sum_{h=0}^{m_{+}+m_{-}}\left\vert (j_{h},j_{h+1})\right\vert =2p$

In the following, we denote
$$
p_{h}:=\frac{1}{2}\left\vert (j_{h},j_{h+1})\right\vert ,\hbox{ for any }h\in\{0,1,\ldots,m_{+}+m_{-}\}
$$
In terms of $p_h$'s, the fact ``$\varepsilon_h$ is balance'' means $\varepsilon_{h}\in\{-1,1\}^{2p_h}_+$.

Using (\ref{arcSin02b}), one can continue the development (\ref{arcSin07b}) of $e^{itX}\Psi_{h}$
as follows
\begin{align}\label{arcSin07d}
e^{itX}\Psi_{h} &  =\sum_{n=0}^{\infty}\frac{\left(it\right)^{n}}{n!}%
\sum_{\varepsilon\in\left\{-1,1\right\}^{n}}a^{\varepsilon\left(
1\right)}\ldots a^{\varepsilon\left(n\right)}\Psi_{h}\\
&  =\sum_{p,m_{+},m_{-}\geq0}\frac{\left(it\right)^{m_{+}+m_{-}+2p}%
}{\left(m_{+}+m_{-}+2p\right)  !}\sum_{\varepsilon\in\Theta_{m_{+}+m_{-}%
+2p}(m_{+},m_{-})}a^{\varepsilon\left(1\right)}\ldots a^{\varepsilon
\left(m_{+}+m_{-}+2p\right)}\Psi_{h}\nonumber\\
&  =\sum_{p,m_{+},m_{-}\geq0}\frac{\left(it\right)^{m_{+}+m_{-}+2p}}{\left(m_{+}+m_{-}+2p\right)  !}\sum_{\varepsilon\in\Theta_{m_{+}+m_{-}+2p}(m_{+},m_{-})}a^{\varepsilon\left(1\right)}\ldots a^{\varepsilon
\left(m_{+}+m_{-}+2p\right)}\Psi_{h}\chi_{h]} \left(m_{-}\right)\nonumber\\
&  =\sum_{p,m_{+}=0}^{\infty}\sum_{m_{-}=0}^{h}\frac{\left(it\right)^{m_{+}+m_{-}+2p}}{\left(m_{+}+m_{-}+2p\right)  !}\sum_{\varepsilon\in
\Theta_{m_{+}+m_{-}+2p}(m_{+},m_{-})}a^{\varepsilon\left(1\right)}\ldots
a^{\varepsilon\left(m_{+}+m_{-}+2p\right)}\Psi_{h}\nonumber
\end{align}

\bigskip

Now we consider the term
$$
\sum_{\varepsilon\in\Theta_{m_{+}+m_{-}+2p}(m_{+},m_{-})}a^{\varepsilon\left(1\right)}\ldots a^{\varepsilon\left(m_{+}+m_{-}+2p\right)}%
$$

\bigskip

\begin{proposition}\label{arcSin05}
For any $m_{\pm},p\geq0,$ $p_{0},p_{1},\ldots,p_{m_{+}+m_{-}}\geq0$ with
$p_{0}+p_{1}+\ldots+p_{m_{+}+m_{-}}=p,$ $h=0,1,\ldots,m_{+}+m_{-}$
and $\varepsilon_{h}\in\left\{-1,1\right\}_{+}^{2p_{h}},$ one denotes
$\left\{\left(l_{s}^{\left(h\right)},r_{s}^{\left(h\right)}\right)  \right\}_{s=1}^{p_{h}}
\in NCPP\left(2p_{h}\right) $ the unique non--crossing pair partition determined by $\varepsilon_{h}.$  Moreover, one introduces
\begin{align}\label{arcSin05op}
&A(m_{+},m_{-},\left\{p_{h}\right\}_{h=0}^{m_{+}+m_{-}}, \left\{\varepsilon_{h}\right\}_{h=0}^{m_{+}+m_{-}})
\notag\\
&:=\left(\prod_{s=1}^{p_0}\omega_{\Lambda+2s-l_s^{(0)}} \right)a^{+}\left(\prod_{s=1}^{p_1}\omega_{\Lambda+2s-l_s^{(1)}}\right)a^{+}\ldots\bigg(\prod_{s=1}^{p_{m_{+}-1}} \omega_{\Lambda+2s-l_{s}^{\left(m_{+}-1\right)}}
\bigg) a^{+}\bigg(\prod_{s=1}^{p_{m_{+}}}\omega_{\Lambda+2s-l_{s}^ {\left(m_{+}\right)}}\bigg)  a\notag\\
&\bigg(\prod_{s=1}^ {p_{m_{+}+1}}\omega_{\Lambda+2s-l_{s}^{\left(m_{+}+1 \right)}}\bigg) a\ldots\notag
\bigg(\prod_{s=1}^{p_{m_{+}+m_{-}-1}}\omega_{\Lambda+2s-l_{s}^{\left(m_{+}+m_{-}-1\right)}}\bigg) a\bigg(\prod _{s=1} ^{p_{m_{+}+m_{-}}}\omega_ {\Lambda+2s-l_{s} ^{\left(m_{+}+m_{-}\right)}}\bigg)\nonumber\\
&=\prod_{j=1}^{m_+}\left(\prod_{s=1}
^{p_{j-1}}\omega_{\Lambda+2s-l_{s}^{(j-1)}}\,a^+\right) \cdot\bigg(\prod_{s=1}^{p_{m_{+}+m_{-}-1}}\omega_
{\Lambda+2s-l_{s}^{\left(m_{+}+m_{-}-1\right)}}\bigg)
\cdot\prod_{j=1} ^{m_-}\bigg(a\prod_{s=1} ^{p_{m_{+}+j }}\omega_{\Lambda+2s-l_{s}^{\left(m_{+}+j\right)}}\bigg)
\end{align}
Then
\begin{align*}
&\sum_{\varepsilon\in\Theta_{m_{+}+m_{-}+2p}(m_{+},m_{-})}
a^{\varepsilon\left(1\right)}\ldots a^{\varepsilon\left(m_{+}+m_{-}+2p\right)}\\
&=\sum_{\substack{p_{0},p_{1},\ldots,p_{m_{+}+m_{-}}\geq0\\
p_{0}+p_{1}+\ldots+p_{m_{+}+m_{-}}=p}}\sum_{\varepsilon_{h}
\in\left\{-1,1\right\}_{+}^{2p_{h}},\forall h=0,1,\ldots,m_{+}+m_{-}} A(m_{+},m_{-},\left\{
p_{h}\right\}_{h=0}^{m_{+}+m_{-}},\left\{\varepsilon_{h}\right\}
_{h=0}^{m_{+}+m_{-}})
\end{align*}

\end{proposition}

\bigskip

\begin{proof}By the definition, one has%
\begin{align}
&\sum_{\varepsilon\in\Theta_{m_{+}+m_{-}+2p}(m_{+},m_{-})}
a^{\varepsilon(1)}\ldots a^{\varepsilon(m_{+}+m_{-}+2p)} \notag\\
&=\sum_{\substack{p_0,p_{1},\ldots,p_{m_{+}+m_{-}}\geq0\\
p_{0}+p_{1}+\ldots+p_{m_{+}+m_{-}}=p}}\sum_{\varepsilon_h\in
\{-1,1\}_{+}^{2p_h},\forall h=0,1,\ldots,m_{+}+m_{-}}\notag\\
&\prod_{j=1}^{m_+}\left(\prod_{s=1}^{p_{j-1}}\omega_{\Lambda+2s-l_{s}^{(j-1)}}\,a^+\right) \cdot \bigg(\prod_{s=1}^{p_{m_{+}+m_{-}-1}}\omega_{\Lambda+2s-l_{s}^{\left(m_{+}+m_{-}-1\right)}}\bigg)\cdot
\prod_{j=1} ^{m_-}\bigg(a\prod_{s=1} ^{p_{m_{+}+j }} \omega_{\Lambda+2s-l_{s}^{\left(m_{+}+j\right)}}\bigg)
\end{align}

For any $h\in\left\{0,1,\ldots,m_{+}+m_{-}\right\}  ,$ Proposition 1 of
\cite{[AcHamLu22b]} says that for any $\varepsilon_{h}\in\left\{-1,1\right\}
_{+}^{2p_{h}}$ with the unique determined $\left\{\left(l_{s}^{\left(
h\right)},r_{s}^{\left(h\right)}\right)  \right\}_{s=1}^{p_{h}}\in
NCPP\left(2p_{h}\right)  $,%
\begin{equation}\label{arcSin05c}
\prod_{s=1}^{2p_{h}}a^{\varepsilon_{h}(s)} =\prod_{s=1}^{p_{h}}
\omega_{\Lambda+2s-l_{s}^{(h)}}
\end{equation}
and so we proved the thesis.
\end{proof}

\bigskip

In the following discussion, we are interested in calculating the terms given by:
\[
\sum_{\varepsilon\in\Theta_{m_{+}+m_{-}+2p}(m_{+},m_{-})}
a^{\varepsilon(1)}\ldots a^{\varepsilon\left(m_{+}+m_{-}+2p\right)}\Psi_{h}
\]
 for any $h\in\mathbb{N}$.

\bigskip

\begin{proposition}\label{arcSin06} For any $n\in\mathbb{N}^{\ast}$ and $\varepsilon\in\{-1,1\}_{+}^{2n}$ with the unique $\left\{\left(l_s,r_s\right)\right\}_{s=1}^{n}\in NCPP(2n)$, one has
\begin{equation}\label{arcSin06a}
2s-l_s=\sum_{k=r_s}^{2n}\varepsilon(k)\geq1,\ \ \forall s\in\{1,\ldots,n\}
\end{equation}
and
\begin{equation}\label{arcSin06b}
\prod_{s=1}^n\omega_{\Lambda+2s-ls}\Psi_m=\prod_{s=1}^n
\omega_{m+2s-l_s}\Psi_m,\ \ \forall m\in\mathbb{N}
\end{equation}
In particular,
\begin{equation}\label{arcSin06c}
\prod_{s=1}^{n}\omega_{\Lambda+2s-l_{s}}\Psi_{0}=\prod_{s=1}^{n}\omega_{2s-l_{s}}\Psi_{0}=\bigg\langle \Psi_{0},\prod_{s=1}^{n}\omega
_{\Lambda+2s-l_{s}}\Psi_{0}\bigg\rangle \Psi_{0}
\end{equation}
and

\begin{itemize}
\item in the case of $\omega_{j}=\omega_{1}$ for all $j\geq1$ (semi--circle
case)
\begin{equation}\label{arcSin06d}
\prod_{s=1}^n\omega_{\Lambda+2s-l_s}\Psi_m=\omega_{1}^n\Psi_m, \ \ \forall m\geq1
\end{equation}

\item in the case of $\omega_{j}=\omega_{2}$ for all $j\geq2$ (generalized arc--sine case)
\begin{equation}\label{arcSin06e}
\prod_{s=1}^n\omega_{\Lambda+2s-l_s}\Psi_m
=\omega_2^n\Psi_m,\ \ \forall m\geq1
\end{equation}
\end{itemize}
\end{proposition}

\begin{proof}The definition of the number operator gives (\ref{arcSin06b}) and in particular (\ref{arcSin06c}).
Moreover, one gets (\ref{arcSin06d}) (respectively,
(\ref{arcSin06e})) in the case of semi--circle (respectively,
arc--sine) since $2s-l_s\geq1$ (so $m+2s-l_s\geq2$ whenever $m\geq1$) for all $s\in\left\{1,\ldots,n\right\} .$ Finally, in (\ref{arcSin06a}), the equality is nothing else than the formula (2.55) of \cite{[AcHamLu22b]}; the inequality is a direct conclusion of the fact: $\left\{\left(l_s,r_s\right)  \right\}_{s=1}^n$ is the unique non--crossing pair
partition determined by  $\varepsilon\in\{-1,1\}_+^{2n}$.
\end{proof}
\begin{theorem}\label{arcSin08}If $\omega_{j}=\omega_{2}$ for any $j\geq2$ (in particular,
arc--sine case: $\omega_{1}=2\omega_{2}>0;$ semi--circle case: $\omega_{1}=\omega_{2}>0$), for any $m_{\pm}\geq0$ and $p_{0},p_{1},\ldots,p_{m_{+}+m_{-}}\geq0, \varepsilon_{r}\in\left\{-1,1\right\}_{+}^{2p_{r}}$
with $r=0,1,\ldots,m_{+}+m_{-}$, for any $h\in\mathbb{N}$
\begin{align}
&A(m_{+},m_{-},\left\{p_{r}\right\}_{r=0}^{m_{+}+m_{-}},\left\{
\varepsilon_{r}\right\}_{h=0}^{m_{+}+m_{-}})\notag\\
=&\left(a^{+}\right)^{m_{+}}a^{m_{-}}\prod_{r=0}^{m_+}
\left(\prod_{s=1}^{p_r}\omega_{\Lambda+2s-l_{s}^{(r)}+m_++m_--r}
\right)\prod_{r=m_++1}^{m_++m_-}\left(\prod_{s=1}^{p_r}\omega_
{\Lambda+2s-l_{s}^{(r)}+r-m_+-m_-}\right)
\end{align}
In particular,  for any $h\in\mathbb{N}$,
\begin{align}\label{arcSin08a}
& A(m_{+},m_{-},\left\{p_{r}\right\}_{r=0}^{m_{+}+m_{-}},
\left\{\varepsilon_r\right\}_{h=0}^{m_{+}+m_{-}})\Psi_h\notag\\
=& \left(a^{+}\right)^{m_{+}}a^{m_{-}}\Psi_{h}
\begin{cases}
0, & \hbox{if }m_{-}>h\\
\omega_{2}^{p-p_{m_{+}}}\prod_{s=1}^{p_{m_+}}
\omega_{2s-l_s^{(m_+)}} , & \hbox{if }m_{-}=h\\
\omega_{2}^{p}, & \hbox{if }m_{-}<h
\end{cases}
\end{align}
\end{theorem}

\begin{proof}
It is proved by combining Proposition \ref{arcSin05} and Proposition \ref{arcSin06}.  %
In particular,  for any $h\in\mathbb{N}$
\begin{align*}
& A(m_{+},m_{-},\left\{p_{r}\right\}_{r=0}^{m_{+}+m_{-}},  \left\{\varepsilon_{r}\right\}_{h=0}^{m_{+}+m_{-}}) \Psi_{h}\\
=& \left(a^{+}\right)^{m_{+}}a^{m_{-}}\Psi_{h}\cdot
\begin{cases}
0, & \hbox{if }m_{-}>h\\
\left\langle \Psi_{0},a^{\varepsilon_{m_{+}}(1)} \ldots a^{\varepsilon_{m_{+}}\left(2p_{m_{+}}\right)}\Psi_{0}
\right\rangle\omega_{2}^{p_{0}+p_{1}+\ldots+p_{m_{+}+m_{-}}-p_{m_{+}}}
 , & \hbox{if }m_{-}=h\\
\omega_{2}^{p_{0}+p_{1}+\ldots+p_{m_{+}+m_{-}}}, & \hbox{if }m_{-}<h
\end{cases}
\end{align*}

\end{proof}

Summing up above results, we have

\begin{theorem}\label{theorem-summing}
If $\omega_{j}=\omega_{2}$ for any $j\geq2$ (in particular,
arc--sine case: $\omega_{1}=2\omega_{2}>0;$ semi--circle case: $\omega_1=\omega_2>0$), for any $m_{\pm},p,h\geq0,$ for any such $p_{0}%
,p_{1},\ldots,p_{m_{+}+m_{-}}\geq0$
\begin{align*}
& \sum_{\substack{\varepsilon_{r}\in\left\{-1,1\right\}_{+}
^{2p_{r}}\\ \forall r=0,1,\ldots,m_{+}+m_{-}}}A(m_{+},m_{-},
\left\{p_{r}\right\}_{r=0}^{m_{+}+m_{-}},\left\{\varepsilon_{r}\right\}_{h=0}^{m_{+}+m_{-}})\Psi _{h}\\
=&\left(a^{+}\right)^{m_{+}}a^{m_{-}}\Psi_{h}
\cdot\omega_{2}^{p-p_{m_{+}}}\cdot \bigg(
\prod_{\substack{0\leq r\leq m_{+}+m_{-}\\r\neq m_{+}}}
C_{p_{r}}\bigg)
\begin{cases}
0, & \hbox{if }m_{-}>h\\
\displaystyle\sum_{\varepsilon_{m_{+}}\in\left\{ -1,1\right\}_{+}^{2p_{m_{+}}}}\prod_{s=1}^{p_{m_+}}
\omega_{2s-l_{s}^{(m_+)}}  , & \hbox{if }m_{-}=h\\
C_{p_{m_{+}}}\omega_{2}^{p_{m_{+}}}, & \hbox{if }m_{-}<h
\end{cases}
\end{align*}
where, $C_n$ is the n-th Catalan number for all $n\in \mathbb{N}$.
\end{theorem}

In the following corollary, we apply this theorem to our special cases: the arc-sine case and the semi-circle case.

\begin{corollary}\label{arcSin09}
Assume that $\omega_{j}=\omega_{2}$ for any $j\geq2$. Then, for any $m_{\pm},p,h\geq0$ and any such $p_{0}
,p_{1},\ldots,p_{m_{+}+m_{-}}\geq0$ and $p_{0}
+p_{1}+\ldots+p_{m_{+}+m_{-}}=p$
\begin{align*}
&\sum_{\varepsilon_{r}\in\{-1,1\}_{+}^{2p_{r}},\forall
r=0,1,\ldots,m_{+}+m_{-}}A(m_{+},m_{-},\{p_r\}_{r=0}
^{m_{+}+m_{-}},\{\varepsilon_r\}_{h=0}^{m_{+}+m_{-}})
\Psi_{h}\\
=&\left(a^{+}\right)^{m_{+}}a^{m_{-}}\Psi_{h}\cdot \omega_{2}^{p}\cdot \bigg(\prod\limits_{\substack{0\leq r\leq m_{+}+m_{-}}}C_{p_{r}}\bigg)  \cdot
\begin{cases}
0, & \hbox{if }m_{-}>h\\
\frac{\mathbf{E}\left(\xi^{2p_{m_+}}\right)}{\omega_2^{p_{m_{+}}} C_{p_{m_{+}}}}  , & \hbox{if }m_{-}=h\\
 1, & \hbox{if }m_{-}<h
\end{cases}
\end{align*}
In particular,
\begin{itemize}
\item In the arc--sine case (we assume $\omega_{1}=2\omega_{j}$):
\begin{align*}
&\sum_{\varepsilon_{r}\in\left\{-1,1\right\}_{+}^{2p_{r}},\forall r=0,1,\ldots,m_{+}+m_{-}}A(m_{+},m_{-},\left \{p_{r}\right\}_{r=0}^{m_{+}+m_{-}},\left\{\varepsilon
_{r}\right\}_{h=0}^{m_{+}+m_{-}})\Psi_{h}\\
=&\omega_{2}^{p}\left(a^{+}\right)^{m_{+}}a^{m_{-}}\Psi
_{h}\bigg(\prod_{\substack{0\leq r\leq m_{+}+m_{-}}}
C_{p_{r}}\bigg)
\begin{cases}
0, & \hbox{if }m_{-}>h\\
p_{m_+}+1  , & \hbox{if }m_{-}=h\\
 1, & \hbox{if }m_{-}<h
\end{cases}
\end{align*}
\item In the  semi--circle case (we assume $\omega_{1}=\omega_{j}$):
\begin{align*}
&\sum_{\varepsilon_{r}\in\left\{-1,1\right\}_{+}^{2p_{r}},\forall
r=0,1,\ldots,m_{+}+m_{-}}A(m_{+},m_{-},\left\{p_{r}\right\}_{r=0}
^{m_{+}+m_{-}},\left\{\varepsilon_{r}\right\}_{h=0}^{m_{+}+m_{-}})\Psi_{h}\\
=&\omega_{2}^{p}\left(a^{+}\right)^{m_{+}}a^{m_{-}}\Psi_{h}\  \left(\prod_{\substack{0\leq r\leq m_{+}+m_{-}}}
C_{p_{r}}\right)
\begin{cases}0, & \hbox{if }m_{-}>h\\ 1, & \hbox{if }m_{-}\le h
\end{cases}
\end{align*}
\end{itemize}
\end{corollary}
\begin{proof}
	This corollary follows directly from Theorem \ref{theorem-summing} and the following fact that is proved in the Appendix \ref{sec:Vac-distr-arcsin}:
	\begin{align*} 
		\sum_{\varepsilon_{m_{+}}\in\left\{ -1,1\right\}_{+} ^{2p_{m_{+}}}}\prod_{s=1}^{p_{m_+}}\omega_{2s-l_{s}^{(m_+)}}
		=&\mathbf{E}\left(\xi^{2p_{m_{+}}}\right)\\	
		=&\begin{cases}
			\omega_2^{p_{m_{+}}}C_{p_{m_{+}}},& 	\hbox{if }	\xi\hbox{ has the variance }\omega_{1}=\omega_{j}, j\geq2\\
			\omega_2^{p_{m_{+}}}(p_{m_{+}}+1)C_{p_{m_{+}}} ,& 	\hbox{if }	\xi\hbox{ has the variance }\omega_{1}=2\omega_{j}, j\geq2
		\end{cases}	
	\end{align*}
\end{proof}
We proceed now to the proof of the main result of this paragraph.
\begin{proof}[Proof of Theorem \ref{unification}]
	
Let $\omega_{j}=\omega_{2}$ for any $j\geq2$ (in particular,
arc--sine case: $\omega_{1}=2\omega_{2}>0;$ semi--circle case: $\omega_{1}=\omega_{2}>0$). According to the above discussion, we have
\begin{align*}
	e^{itX}\Psi_{h} =&  \sum_{n=0}^{\infty}\frac{\left(it\right)^{n}}{n!}%
	\sum_{\varepsilon\in\left\{-1,1\right\}^{n}}a^{\varepsilon\left(
		1\right)}\ldots a^{\varepsilon\left(n\right)}\Psi_{h}\label{arcSin07d}\\
	=&  \sum_{p,m_{+},m_{-}\geq0}\frac{\left(it\right)^{m_{+}+m_{-}+2p}%
	}{\left(m_{+}+m_{-}+2p\right)  !}\sum_{\varepsilon\in\Theta_{m_{+}+m_{-}%
			+2p}(m_{+},m_{-})}a^{\varepsilon\left(1\right)}\ldots a^{\varepsilon
		\left(m_{+}+m_{-}+2p\right)}\Psi_{h}\nonumber
	\\=&  \sum_{p,m_{+},m_{-}\geq0}\frac{\left(it\right)^{m_{+}+m_{-}+2p}%
	}{\left(m_{+}+m_{-}+2p\right)  !}\sum_{\substack{p_{0},p_{1},\ldots,p_{m_{+}+m_{-}}\geq0\\p_{0}+p_{1}%
			+\ldots+p_{m_{+}+m_{-}}=p}} {\displaystyle\prod\limits_{\substack{0\leq r\leq m_{+}+m_{-}}}}
	C_{p_{r}}
	\\& \omega_{2}^{p}\left(a^{+}\right)^{m_{+}}a^{m_{-}}\Psi_{h}\  \left\{
	\begin{array}
		[c]{ll}%
		0, & \hbox{if }m_{-}>h\\
	1, & \hbox{if }m_{-}<h\\
	\frac{\mathbf{E}\left(\xi^{2p_{m_{+}}}\right)}{\omega_{2}^{p_{m_{+}}}C_{p_{m_{+}}}}  , & \hbox{if }m_{-}=h
	\end{array}
	\right.
\end{align*}
Similarly, we have
\begin{align*}
	e^{itP}\Psi_{h} =&\sum_{m\ge 0}\sum_{\varepsilon\in\{-1,1\}^{m}}\frac{t^{m}}{m!}
	(-1)^{\nu_+(\varepsilon)}a^{\varepsilon(1)}\cdots a^{\varepsilon( m)}\Psi_{h}
	\\=&\sum_{p,m_{+},m_{-}\geq0}\frac{t^{m_{+}+m_{-}+2p}%
	}{\left(m_{+}+m_{-}+2p\right)  !}\sum_{\varepsilon\in\Theta_{m_{+}+m_{-}%
			+2p}(m_{+},m_{-})}(-1)^{\nu_+(\varepsilon)}a^{\varepsilon\left(1\right)}\ldots a^{\varepsilon
		\left(m_{+}+m_{-}+2p\right)}\Psi_{h}\nonumber
	\end{align*}
where $\nu_+$ takes value $p+m_+$ on $\Theta_{m_++m_-+2p}(m_+,m_-)$, for any $m_+,m_-,p\ge 0$. Therefore
\begin{multline*}
	e^{itP}\Psi_{h} = \sum_{p,m_{+},m_{-}\geq0}\frac{\left( -1\right)^{p+m_+}  t^{m_{+}+m_{-}+2p}%
	}{\left(m_{+}+m_{-}+2p\right)  !}\sum_{\substack{p_{0},p_{1},\ldots,p_{m_{+}+m_{-}}\geq0\\p_{0}+p_{1}%
			+\ldots+p_{m_{+}+m_{-}}=p}} {\displaystyle\prod\limits_{\substack{0\leq r\leq m_{+}+m_{-}}}}
	C_{p_{r}}
	\\ \omega_{2}^{p}\left(a^{+}\right)^{m_{+}}a^{m_{-}}\Psi_{h}\  \left\{
	\begin{array}
		[c]{ll}%
		0, & \hbox{if }m_{-}>h\\
		1, & \hbox{if }m_{-}<h\\
		\frac{\mathbf{E}\left(\xi^{2p_{m_{+}}}\right)}{\omega_{2}^{p_{m_{+}}}C_{p_{m_{+}}}}  , & \hbox{if }m_{-}=h
	\end{array}
	\right.
\end{multline*}
\end{proof}

\subsection{ The monotone evolutions $e^{itX}$ and $e^{itP}$ }

We are interested here in the arc--sine case, which means that we set
$\omega_{1}=2\omega_{j}=2$ for any $j\geq2$.
In this case, and {\bf for all that follows}, the monotone position and momentum operators $X_{as}$ and $P_{as}$ will be referred to simply as $X$ and $P$.
Let's start with the following technical result.
\begin{lemma}\label{Catalan}
We have
\begin{align*}
\sum_{\substack{p_{0},p_{1},\ldots,p_{k}\geq0\\
p_{0}+p_{1}+\ldots+p_{k}=p}}\bigg(\displaystyle \prod_{\substack{0\leq r\leq k}}C_{p_{r}}\bigg) =\frac{k+1}{2p+k+1}\binom{2p+k+1}{p}
\end{align*}
and for any $i\in\{1,\ldots,k\}$
\begin{align*}
\sum_{\substack{p_{0},p_{1},\ldots,p_{k}\geq0
\\p_{0}+p_{1}+\ldots+p_{k}=p}} (p_{i}+1)
\bigg(\displaystyle\prod_{\substack{0\leq r\leq k}}
C_{p_{r}}\bigg)  =\binom{2p+k}{p}
\end{align*}
\end{lemma}

\begin{proof}
The first identity is proved in \cite[Proposition 3]{[AcHamLu22b]}. For the second identity,
consider the following generating functions:
\begin{equation*}
G_1(z)= \sum_{n\geq 0} (n+1)C_n z^n = \sum_{n\geq 0}\binom{2n}{n}z^n=\frac{1}{\sqrt{1-4z}}
	\end{equation*}
	
\begin{equation*}
G_2(z)= \sum_{n\geq 0} C_n z^n = \sum_{n\geq 0}\binom{2n}{n}\frac{z^n}{n+1}=\frac{2}{1+\sqrt{1-4z}}
\end{equation*}
	Then number
	\begin{align*}
	\sum_{\substack{p_{0},p_{1},\ldots,p_{k}\geq0\\p_{0}+p_{1}%
			+\ldots+p_{k}=p}}\left(
	{\displaystyle\prod\limits_{\substack{0\leq r\leq k}}}
	C_{p_{r}}\right) (p_{i}+1)
	\end{align*}
	is interpreted as the coefficient of $z^p$ in the product $G_1(z)G_2(z)^{k}$ which has the following series expansion (see, e.g., \cite[proof of corollary 3.2]{NiHamHmi12}):
	
\begin{equation*}
G_1(z)G_2(z)^{k}=\frac{2^{k}}{\sqrt{1-4z}(1+\sqrt{1-4z})^{k}}=\sum_{n \geq 0}
\binom{2n+k}{n}z^n
\end{equation*}
\end{proof}

\begin{proposition}\label{arc--sine_case}
	We have
\begin{align*}
&\sum_{\varepsilon\in\Theta_{m_{+}+m_{-}+2p}(m_{+},m_{-})}a^{\varepsilon(1)}\ldots a^{\varepsilon\left(m_{+}+m_{-}+2p\right)}T_{h}\\
& =\left(a^{+}\right)^{m_{+}}a^{m_{-}}T_{h}\
\begin{cases} 0, & \hbox{if }m_{-}>h\\
\binom{2p+m_{+}+m_{-}}{p}, & \hbox{if }m_{-}=h\\
\frac{m_{+}+m_{-}+1}{2p+m_{+}+m_{-}+1}\binom{2p+m_{+} +m_{-}+1}{p} , & \hbox{if }m_{-}<h   \end{cases}
\end{align*}
\end{proposition}

\begin{proof}
This follows from the Lemma \ref{Catalan} since we have
\begin{align*}
&\sum_{\varepsilon\in\Theta_{m_{+}+m_{-}+2p}(m_{+},m_{-})}a^{\varepsilon\left(1\right)}\ldots a^{\varepsilon\left(m_{+}+m_{-}+2p\right)}T_{h}\\
&=\sum_{\substack{p_{0},p_{1},\ldots,p_{m_{+}+m_{-}}
\geq0\\p_{0}+p_{1}+\ldots+p_{m_{+}+m_{-}}=p}}
\sum_{\varepsilon_{h}\in\left\{-1,1\right\}_{+}^{2p_{h}},\forall h=0,1,\ldots,m_{+}+m_{-}}A(m_{+},m_{-},
\left\{p_{h}\right\}_{h=0}^{m_{+}+m_{-}},
\left\{\varepsilon_{h}\right\}_{h=0}^{m_{+}+m_{-}})\\
&=\left(a^{+}\right)^{m_{+}}a^{m_{-}}T_{h}\ \sum_{\substack{p_{0},p_{1},\ldots,p_{m_{+}+m_{-}}\geq0\\p_{0}+p_{1}+\ldots+p_{m_{+}+m_{-}}=p}}\bigg(
\prod_{\substack{0\leq r\leq m_{+}+m_{-}}}
C_{p_{r}}\bigg) \begin{cases} 0, & \hbox{if }m_{-}>h\\
p_{m_+}+1, & \hbox{if }m_{-}=h\\ 1, & \hbox{if }m_{-}<h
\end{cases}
\end{align*}
\end{proof}

From this one can deduce
	
\begin{theorem}
	The action on the $\mu_{as}$--orthogonal polynomials of the monotone evolutions $e^{itX}$ and $e^{itP}$ are respectively given by
	\begin{align}\label{position}
		e^{itX}T_{h} = \frac{1}{t}\sum_{n=0}^{h-1} \sum_{m\ge n}\left(i\right)^{m}(m+1)J_{m+1}(2t) T_{h+m-2n}
		+2\sum_{m\ge0} \left(i\right)^{m+h}J_{m+h}(2t)  T_{m}.
	\end{align}	
and
\begin{align}\label{momentum}
	e^{itP}T_{h} = \frac{1}{t}\sum_{n=0}^{h-1} \sum_{m\ge n}\left(-1\right)^{m}(m+1)J_{m+1}(2t) T_{h+m-2n}
	+2\sum_{m\ge0} \left(-1\right)^{m}J_{m+h}(2t)  T_{m}.
\end{align}	
for any $h\in \mathbb{N}$ and any $t\in\mathbb{C}$ where an empty sum is zero. 
	\end{theorem}
\begin{proof}
By Proposition \ref{arc--sine_case}, we have
\begin{align*}
	e^{itX}T_{h}
	=&  \sum_{p,m_{+}=0}^{\infty}\sum_{m_{-}=0}^{h}\frac{\left(it\right)
		^{m_{+}+m_{-}+2p}}{\left(m_{+}+m_{-}+2p\right)  !}\sum_{\varepsilon\in
		\Theta_{m_{+}+m_{-}+2p}(m_{+},m_{-})}a^{\varepsilon\left(1\right)}\ldots
	a^{\varepsilon\left(m_{+}+m_{-}+2p\right)}T_{h}
	\\=&  \sum_{p,m_{+}=0}^{\infty}\sum_{m_{-}=0}^{h-1}\frac{\left(it\right)
		^{m_{+}+m_{-}+2p}}{\left(m_{+}+m_{-}+2p\right)  !}  \frac{m_{+}+m_{-}+1}{2p+m_{+}+m_{-}+1}\binom{2p+m_{+}+m_{-}+1}{p}  \left(a^{+}\right)^{m_{+}}a^{m_{-}}T_{h}
	\\&+\sum_{p,m_{+}=0}^{\infty}\frac{\left(it\right)^{m_{+}+h+2p}}{\left(m_{+}+h+2p\right)  !}  \binom{2p+m_{+}+h}{p}  \left(a^{+}\right)^{m_{+}}a^{h}T_{h}
	\\=&  \sum_{p,m_{+}=0}^{\infty}\sum_{m_{-}=0}^{h-1}\frac{\left(it\right)
		^{m_{+}+m_{-}+2p}}{p !}  \frac{m_{+}+m_{-}+1}{(p+m_{+}+m_{-}+1)!}  \left(a^{+}\right)^{m_{+}}a^{m_{-}}T_{h}
	\\&+\sum_{p,m_{+}=0}^{\infty}\frac{\left(it\right)^{m_{+}+h+2p}}{p!\left(m_{+}+h+p\right)  !}   \left(a^{+}\right)^{m_{+}}a^{h}T_{h}
	\\=&  \sum_{m_{+}=0}^{\infty}\sum_{m_{-}=0}^{h-1}\frac{\left( i\right)^{m_{+}+m_{-}}J_{m_{+}+m_{-}+1}(2t)}{t}  (m_{+}+m_{-}+1)  \left(a^{+}\right)^{m_{+}}a^{m_{-}}T_{h}
	\\&+\sum_{m_{+}=0}^{\infty} \left(i\right)^{m_{+}+h}J_{m_{+}+h}(2t)  \left(a^{+}\right)^{m_{+}}a^{h}T_{h}
\end{align*}

Then, using
\begin{equation*}
	\left(a^{+}\right)^{m_{+}}a^{m_{-}}T_{h}=
	\begin{cases}
		T_{h+m_{+}-m_{-}}&; h>m_{-}\\
		2T_{m_{+}}&; h=m_{-}\\
		0&; h<m_{-}
	\end{cases}
\end{equation*}
we are led to
\begin{align*}
	e^{itX}T_{h}
=& \sum_{m_{-}=0}^{h-1} \sum_{m_{+}=0}^{\infty}\frac{\left( i\right)^{m_{+}+m_{-}}J_{m_{+}+m_{-}+1}(2t)}{t}  (m_{+}+m_{-}+1) T_{h+m_{+}-m_{-}}
	\\&+2\sum_{m_{+}=0}^{\infty} \left(i\right)^{h+m_{+}}J_{h+m_{+}}(2t)  T_{m_{+}}
\end{align*}
Performing the index change $m\mapsto m_++m_-$, we end up with the equality \eqref{position}. Similarly, we have
\begin{align*}
	e^{itP} =&\sum_{m\ge 0}\frac{(it)^{m}}{m!}P^{m}=
	\sum_{m\ge 0}\sum_{\varepsilon\in\{-,+\}^{m}}\frac{(i)^{m}t^{m}}{m!}(-i)^{m}
	(-1)^{\nu_+(\varepsilon)}a^{\varepsilon(1)}\cdots a^{\varepsilon( m)}
	\\=& \sum_{m\ge 0}\sum_{\varepsilon\in\{-,+\}^{m}}\frac{t^{m}}{m!}
	(-1)^{\nu_+(\varepsilon)}a^{\varepsilon(1)}\cdots a^{\varepsilon( m)}
	\\=&\sum_{m_+,m_-,p\ge 0}\frac{t^{m_++m_-+2p}}{(m_++m_-+2p)!}
	\sum_{\varepsilon\in\Theta_{m_++m_-+2p}(m_+,m_-)}
	(-1)^{\nu_+(\varepsilon)}a^{\varepsilon(1)}\cdots a^{\varepsilon(m_++m_-+2p)}
\end{align*}
where $\nu_+$ takes value $p+m_+$ on $\Theta_{m_++m_- +2p}(m_+,m_-)$, for any $m_+,m_-,p\ge 0$. Therefore

\begin{align*}
&e^{itP}T_{h}\\
=&\sum_{m_{-}=0}^{h-1} \sum_{m_+,p\ge 0} \frac{(-1)^{p+m_+}t^{m_++m_-+2p}}{(m_++m_-+2p)!} \frac{m_{+}+m_{-}+1}{2p+m_{+}+m_{-}+1}\binom{2p+m_{+}+
m_{-}+1}{p}  \left(a^{+}\right)^{m_{+}}a^{m_{-}}T_{h}\\
&+\sum_{p,m_{+}=0}^{\infty}\frac{(-1)^{p+m_+} t^{m_{+}+h+2p}}{\left(m_{+}+h+2p\right)!}\binom
{2p+m_++h}{p} \left(a^{+}\right)^{m_{+}}a^{h}T_{h}\\
=&\sum_{m_{-}=0}^{h-1} \sum_{p,m_{+}=0}^{\infty}
\frac{(-1)^{p+m_+}t^{m_{+}+m_{-}+2p}}{p !}  \frac{m_{+}+m_{-}+1}{(p+m_{+}+m_{-}+1)!}  \left(a^{+}\right)^{m_{+}}a^{m_{-}}T_{h}\\
&+\sum_{p,m_{+}=0}^{\infty}\frac{(-1)^{p+m_+}
t^{m_{+}+h+2p}}{p!\left(m_{+}+h+p\right)  !}   \left(a^{+}\right)^{m_{+}}a^{h}T_{h}\\
=&\sum_{m_{+}=0}^{\infty}\sum_{m_{-}=0}^{h-1}\frac{(-1) ^{m_{+}}J_{m_{+}+m_{-}+1}(2t)}{t}(m_{+}+ m_{-}+1) \left(a^{+}\right)^{m_{+}}a^{m_{-}}T_{h}\\
&+\sum_{m_{+}=0}^{\infty} (-1)^{m_{+}}J_{m_{+}+h}(2t)  \left(a^{+}\right)^{m_{+}}a^{h}T_{h}
\end{align*}
and the theorem is proved.
\end{proof}


In particular, for $h=0$ we get the following identities.
\begin{corollary}
	For any $t\in \mathbb{C}$ and $x\in[-2,2]$, one has
	\begin{align*}
		e^{itX}T_{0}(x) &=2e^{itx}
	\end{align*}
	and
	\begin{align*}
	e^{itP}T_{0}(x) &=J_{0}(2t)+\cos(t\sqrt{4-x^2})
-\frac{1}{\pi}\hbox{p.v.}\int_{-2}^{2}\frac{\sin(t\sqrt{4-y^2})}{x-y}dy
\end{align*}
	\end{corollary}
\begin{proof}

These are consequences of their analogous results in \cite{[AcHamLu22a]} (see Lemma 5 and its proof).
\end{proof}

\section{Monotone evolutions generated by the momentum and kinetic energy operators and monotone harmonic oscillator}

As before, we use the notation $P$ in this paragraph to denote the operator $P_{as}$.
\subsection{The monotone Hamiltonian group $e^{itP^2}$}\label{kinetic_energy_evolution}
\begin{theorem}
	For any $h\in \mathbb{N}$ and any $t\in\mathbb{C}$, we have
	\begin{align}\label{hamiltonian}
		e^{itP^{2}}T_{h} =&\frac{e^{2it}}{t}\sum_{n=0}^{h-1} \sum_{\substack{m\ge 0\\m+n\  even}}(-1)^m(-i)^{\frac{m+n}{2}}\frac{m+n}{2}J_{\frac{m+n}{2}}(2t)T_{h+m-n}
		\\&+2e^{2it} \sum_{\substack{m\ge 0\\m+h\  even}}(-1)^m(-i)^{\frac{m+h}{2}}J_{\frac{m+h}{2}}(2t)T_{m}\nonumber
	\end{align}
 where an empty sum is zero. In particular,
 \begin{equation*}
 		e^{itP^{2}}T_{0}(x)=2e^{it(4-x^2)}
 \end{equation*}
\end{theorem}
\begin{proof}
Expanding $e^{itP^{2}}$ we obtain
\begin{align*}
e^{itP^{2}}=\sum_{m\geq0}\frac{\left(it\right)^{m}}{m!}P^{2m}\overset{P=-i\left(a-a^{+}\right)}{=}
\sum_{m\geq0}\frac{\left(-it\right)^{m}}{m!}
\sum_{\varepsilon\in\{-,+\}^{2m}}(-1)^{\nu_{+}
(\varepsilon)}a^{\varepsilon(1)}\cdots a^{\varepsilon(2m)}
\end{align*}
Thus,
\begin{align*}
e^{itP^{2}}T_{h}&= \sum_{m_-=0}^{h-1} \sum_{\substack{m_+,p\ge 0\\m_++m_-\ even}} \frac{(-1)^{p+m_+}(-it)^{\frac{m_++m_-}{2}+p}}
{(\frac{m_++m_-}{2}+p)!} \frac{m_{+}+m_{-}+1}{2p+m_{+}+m_{-}+1}\\
&\hspace{3.6cm}\binom{2p+m_{+}+m_{-}+1}{p}  \left(a^{+}\right)^{m_{+}}a^{m_{-}}T_{h}\\
&+\sum_{\substack{p,m_{+}=0\\m_++h\ even}}^{\infty}\frac{(-1)^{p+m_+} (-it)^{\frac{m_++h}{2}+p}}{\left(\frac{m_++h}{2}+p\right)  !}  \binom{2p+m_{+}+h}{p}  \left( a^{+}\right)^{m_{+}}a^{h}T_{h}
\end{align*}
Then, using the Legendre duplication formula, we are led to, for any $h\in \mathbb{N}$ and any $t\in\mathbb{C}$,
\begin{align*}
e^{itP^{2}}T_{h} =&\sum_{n=0}^{h-1} \sum_{\substack{m\ge 0\\
m+n even}}\frac{(-1)^m(-it)^{\frac{m+n}{2}}}
{\Gamma\left(1+\frac{m+n}{2}\right)}{}_1
F_1\left(\frac{m+n+1}{2};m+n+2;4it\right)T_{h+m-n}\\
+&2 \sum_{\substack{m\ge 0\\m+h even}} \frac{(-1)^m(-it)^{\frac{m+h}{2}}}
{\Gamma\left(1+\frac{m+h}{2}\right)}{}_1
F_1\left(\frac{m+h+1}{2};m+h+1;4it\right)T_{m}
\end{align*}
	Finally, using the identities
	\begin{equation*}
{}_1F_1\left(\frac{m+n+1}{2};m+n+2;4it\right)=\frac{m+n}{2t}e^{2it}
\Gamma\left(1+\frac{m+n}{2}\right)t^{-\frac{m+n}{2}}J_{\frac{m+n}{2}}(2t)
	\end{equation*}
and
\begin{equation*}
	{}_1F_1\left(\frac{m+h+1}{2};m+h+1;4it\right)
=e^{2it}\Gamma\left(1+\frac{m+h}{2}\right)t^{-\frac{m+h}{2}}J_{\frac{m+h}{2}}(2t)
\end{equation*}
	the expression \eqref{hamiltonian} is proved.
	Finally, when $h=0$, expression \eqref{hamiltonian} becomes
	\begin{align*}
	e^{itP^{2}}T_{h}(x) =&2e^{2it}	\sum_{k\ge 0}(-i)^kJ_k(2t)T_{2k}(x)
	\\\overset{[\theta=\arccos(x/2)]}{=}&2e^{2it}J_0(2t)	+4e^{2it}\sum_{k\ge 0}(-i)^kJ_k(2t)\cos(2k\theta)
	\\=&2e^{2it}J_0(2t)	+4e^{2it}\sum_{k\ge 0}i^kJ_k(2t)\cos[k(2\theta+\pi)]
		\\=&2e^{2it}e^{2it\cos(2\theta+\pi)}
		\\=&e^{it(4-x^2)}
	\end{align*}
\end{proof}

\subsection{The $1$--parameter $*$--automorphisms groups associated to $P$ and $P^2$}

In this subsection, we determine the action of the $*$--automorphisms
$e^{itP}( \ \cdot \ )e^{-itP}$ and $e^{itP^2}( \ \cdot \ )e^{-itP^2}$ on the algebra
generated by creation and annihilation operators. Since $a^*=a^+$, it is then sufficient
to determine this action on $a^+$.

\subsubsection{The $*$--automorphism group associated to $P$}

Set
\begin{equation*}
	a^{+}_t:=e^{itP}a^{+}e^{-itP}
\end{equation*}
We have,
\begin{equation*}
	\partial_ta^{+}_t	=ie^{itP}[P,a^{+}]e^{-itP}
\end{equation*}
and
\begin{equation*}
[P,a^{+}]=-i[a,a^+]	
\end{equation*}
where, in the arc--sine case, we have
\begin{equation}\label{commutation}
	a^+aT_h=2\delta_{1,h}T_1+\chi_{[2}(h)T_{h}
	\qquad {\rm and \ }
	aa^+T_h=2\delta_{0,h}T_0+\chi_{[1}(h)T_{h}
\end{equation}
Thus,
\begin{equation*}
	[P,a^{+}]=-i[a,a^+]	\overset{\eqref{commutation}}=-i(2T_{0}T_0^*-\frac{1}{2}T_{1}T_1^*).
\end{equation*}
and hence
\begin{equation*}
	\partial_ta^{+}_t	=2\left(e^{itP}T_0\right)\left(e^{itP}T_0\right)^*
-\frac{1}{2}\left(e^{itP}T_1\right)\left(e^{itP}T_1\right)^*
\end{equation*}
or equivalently
\begin{equation*}
	a^{+}_t	=a^++\int_0^tds \left(2\left(e^{isP}T_0\right)\left(e^{isP}T_0\right)^*
-\frac{1}{2}\left(e^{isP}T_1\right)\left(e^{isP}T_1\right)^*\right)
\end{equation*}

\subsubsection{The $1$--parameter $*$--automorphism group associated to $P^2$}
\hfill\break

In this paragraph we prove that the action of the $1$--parameter $*$--automorphism
group  $e^{isP^2}(\, \cdot \,)e^{isP^2}$ on the quantum algebra associated to the
classical random variable $X$ is uniquely determined by the action of $e^{isP^2}$ on
the three vectors $T_0,T_1$ and $T_2$. In fact, setting
\begin{equation*}
	a^{+}_t:=e^{itP^2}a^{+}e^{-itP^2}
\end{equation*}
one has,
\begin{equation*}
\partial_ta^{+}_t	=ie^{itP^2}[P^2,a^{+}]e^{-itP^2}
\end{equation*}
 where
\begin{equation*}
	[P^2,a^{+}]=[a,(a^+)^2]+[a^+,a^2]
\end{equation*}

Using \eqref{commutation}, we have
\begin{equation*}
(a^+)^2aT_h=
2\delta_{1,h}T_2+\chi_{[2}(h)T_{h+1} ,
\qquad a(a^+)^2T_h=T_{h+1},
\end{equation*}
\begin{equation*}
a^+a^2T_h=
2\delta_{2,h}T_1
+\chi_{[3}(h)T_{h-1} \qquad{\rm and}\
a^2a^+T_h= 
2\delta_{1,h}T_0+\chi_{[2}(h)T_{h-1}
\end{equation*}
Then, we get
\begin{align*}
[P^2,a^{+}]T_h=[a,(a^+)^2]T_h+[a^+,a^2]T_h&=\delta_{0,h}T_1-2\delta_{1,h}T_0-\delta_{1,h}T_2+\delta_{2,h}T_1
\end{align*}
Thus,
\begin{align*}
	[P^2,a^{+}]&=(T_{1,0}+T_{1,2})-(T_{1,0}+T_{1,2})^*
	\\&=T_1\big(T_0+\frac{1}{2}T_2\big)^*-\big(T_0+\frac{1}{2}T_2\big)T_1^*
\end{align*}

and hence
\begin{equation*}
	a^{+}_t	=a^++i\int_0^tds e^{isP^2}T_1\big(e^{isP^2}T_0+\frac{1}{2}e^{isP^2}T_2\big)^*-\big(e^{isP^2}T_0+\frac{1}{2}e^{isP^2}T_2\big)(e^{isP^2}T_1)^*
\end{equation*}

\subsection{The monotone harmonic oscillator}

Recall from \cite{[AcHamLu22a]} that, for $\hat{\omega}\in\mathbb{R}$, the
$\hat{\omega}$--harmonic oscillator in generalized quantum mechanics is given by:
\begin{equation*} H_{\hat{\omega}}=\frac{1}{2}(\hat{\omega}^2-1)({a^+}^2+a^2)+\frac{1}{2}(\hat{\omega}^2+1)(a^+a+aa^+).
\end{equation*}
In particular, putting $\hat{\omega}^2=1$, one finds
\begin{equation*}
	H_{1}=a^+a+aa^+,
\end{equation*}
Using \eqref{commutation}, we get
\begin{equation*}
	H_1T_h=2\delta_{0,h}T_0+3\delta_{1,h}T_1+2\chi_{[2}(h)T_{h}
\end{equation*}
For any $n\ge0$, we have
\begin{equation*}
	H_1^nT_h=2^n\delta_{0,h}T_0+3^n\delta_{1,h}T_1+2^n\chi_{[2}(h)T_{h}
\end{equation*}
and hence
\begin{equation*}
	e^{itH_1}=e^{2it}Id+(e^{3it}-e^{2it})T_{1,1}
\end{equation*}

\begin{remark}
	We expect that the case $\hat{\omega}^2\ne1$ can be dealt with using techniques similar to those in section \ref{kinetic_energy_evolution}. However, this case will not be considered in the present paper.
\end{remark}

\section{Exponential vectors }

\subsection{Exponential vectors in $1$--mode interacting Fock spaces ($1$MIFS)}

\hfill\break

\noindent Recalling that the theory of orthogonal polynomials in $1$ variable can be identified with the theory of $1$--mode interacting Fock spaces ($1$MIFS), in this section we use the
notation \eqref{df-Gamma-omega} for the $1$MIFS associated to the
principal Jacobi sequence $\omega=\left(\omega_{n}\right)_{n\in\mathbb{N}}$.
Exponential (and coherent) vectors in $1$MIFS were introduced in \cite{[Das98]}
(see the bibliography in \cite{[Das02b]} for further references).
Here we briefly recall their definition and some properties in the general frame of $1$MIFS, then specialize to the semi-circle and monotone case.

Denote $\Phi_{\omega;0}$ be the vacuum vector of $\Gamma_{\omega}\left(\mathbb{C}\right)$
 and $\Phi_{\omega;n}:=a^{+n}\Phi_{\omega;0}$ the $n$--particle vector ($n\ge 1$).
The \textit{normalized} $n$--\textit{particle vector} (which in the above mentioned identification is the normalized orthogonal polynomial of degree $n$) is
$$
\tilde{\Phi}_{\omega;n}:=
\begin{cases}
\frac{1}{\sqrt{\omega_{n}!}}\Phi_{\omega;n}, & \mbox{if } \omega_{n}!\neq0 \\
0, & \mbox{if } \omega_{n}! = 0
\end{cases}
\quad , \quad \forall n\in\mathbb{N}^{\ast}
$$
and hereinafter, for any $n\in{\mathbb N}$
$$
\omega_n!:=\begin{cases}1,&\text{ if }n=0\\
\prod_{k=1}^n\omega_k,&\text{ if }n\ge1
\end{cases}
$$
With these notations, each $z\in\mathbb{C}$, for which the series
\begin{equation}\label{df-om-z-coher-vec}
\Psi_{\omega; z}:=\sum_{n=0}^{\infty} \frac{z^{n}}{\sqrt{\omega_{n}!}}\tilde{\Phi}_{\omega;n}
\end{equation}
is norm--convergent in $\Gamma_{\omega}\left(\mathbb{C}\right)$, defines a vector in this space.
Since
$$
\left\langle\sum_{n=0}^{\infty}\frac{u^{n}}{\sqrt{\omega_{n}!}}\tilde{\Phi}_{\omega;n},
\sum_{n=0}^{\infty}\frac{z^{n}}{\sqrt{\omega_{n}!}}\tilde{\Phi}_{\omega;n}\right\rangle
=\sum_{n=0}^{\infty}\frac{(\overline{u}z)^{n}}{\omega_{n}!}
$$
this is the case if an only if $z$ is such that
\begin{equation}\label{conv-cond-om-z-coher-vec}
\sum_{n=0}^{\infty}\frac{|z|^{2n}}{\omega_{n}!} < +\infty
\end{equation}
Denote
\begin{equation}\label{df-C(omega)}
\mathbb{C}_{\omega}:=\left\{z\in\mathbb{C}\colon \eqref{conv-cond-om-z-coher-vec}
\hbox{ is satisfied}\right\}
\end{equation}
In the following, when $\omega$ is fixed and no confusion is possible, \textbf{we simply omit}
the symbol $\omega$, e.g. we write $\Psi_{z}$ instead of $\Psi_{\omega; z}$.\\

\begin{remark}
Notice that $\mathbb{C}_{\omega}$ in \eqref{df-C(omega)} can be reduced to $\{0\}$.
Since, for $n\ge 1$,
$$
\frac{|z|^{2n}}{\omega_{n}!} =\prod_{k=1}^{n}\frac{|z|^{2}}{\omega_{k}}
$$
this happens for example if $\omega_{n}\to 0$ as $n\to +\infty$.
\end{remark}
\begin{definition}
\upshape
For any $z\in\mathbb{C}$,
the vector \eqref{df-om-z-coher-vec} is called the
$\omega$--\textbf{exponential vector} associated to $z$ if the series in \eqref{df-om-z-coher-vec} converges.  A normalized exponential vector is also called a \textbf{coherent vector}.
The state associated to $\Psi_ z:=\Psi_{\omega; z}$, namely
\begin{equation}\label{df-coher-st}
\psi _{z}\left(\cdot\right)
:=\frac{1}{\left\Vert \Psi_{z}  \right\Vert ^{2}}\left\langle \Psi_{z}  ,(\, \cdot \, )\Psi_{z}  \right\rangle
\end{equation}
is called the \textbf{coherent state} associated to $z$. The set of these states is denoted
$\hbox{Exp}(\omega)$.

If $X$ is a \textbf{symmetric} random variable with probability distribution $\mu$, we also say
that $\Psi_{z}$ is the $X$--exponential vector (or the $\mu$--exponential vector)
associated to $z$ and similarly for the associated coherent states and
for $\mathbb{C}_{\omega}$.
\end{definition}
\noindent Recall that the number operator $\Lambda\equiv \Lambda_{\omega}$ on $\Gamma\left(\omega\right)$ is characterized by
$$
\Lambda\Phi_{\omega;n}:= n\Phi_{\omega;n}\quad,\quad\forall n\in\mathbb{N}
$$
\begin{remark}
For any $z\in\mathbb{C}_{\omega}$ and $t\in\mathbb{C}_{\omega}$,
\begin{equation}\label{e(u-Lambda)Psi-z}
e^{t\Lambda}\Psi_{z}
=\sum_{n=0}^{\infty}\frac{z^{n}}{\sqrt{\omega_{n}!}}e^{tn}\tilde{\Phi}_{\omega;n}
=\sum_{n=0}^{\infty}\frac{(e^{t}z)^{n}}{\sqrt{\omega_{n}!}}\tilde{\Phi}_{\omega;n}
=\Psi_{e^{t}z}
\end{equation}
For example, if $\omega_{n}=1$ for $n$ sufficiently big (e.g. both arc--sine and semicircle for $\omega_2=1$), one has $\mathbb{C}_{\omega}=\{z\in \mathbb{C}:\vert z\vert<1\}$ and so $e^{t\Lambda}$ is well--defined if and only if $\vert e^t\vert\le1$
\end{remark}
\begin{remark}
From now on we will only deal with principal Jacobi sequences $\omega$ such that
\begin{equation}\label{supp-om-inf}
n_{\omega}{:=\max\{k:\omega_k>0\} }= +\infty \ \ \hbox{or equivalently } \ \omega_{n}>0 \ , \ \forall n\in\mathbb{N}
\end{equation}
\end{remark}
Notice that the definition of  $\omega$--\textbf{exponential vector} makes sense if,
keeping the same sequence $\omega$, one replaces $\Gamma_\omega$ by any
Hilbert space $\mathcal{K}$ and the $\tilde{\Phi}_{\omega;n}$ by any ortho--normal
basis $(e_{\mathcal{K},n})$ of $\mathcal{K}$.\\
\noindent In particular, one can take $\mathcal{K}:=\Gamma_{\omega'}$, where
$\omega'\ne\omega$ is another principal Jacobi sequence (also satisfying \eqref{supp-om-inf}),
$e_{\mathcal{K},n}= \tilde{\Phi}_{\omega';n}$ (for each $n\in\mathbb{N}$) and define
\begin{equation}\label{df-om-om'-z-coher-vec}
\Psi_{\omega, \omega';z}:=\sum_{n=0}^{\infty}\frac{z^{n}} {\sqrt{\omega_{n}!}}\tilde{\Phi}_{\omega';n}
\in\Gamma_{\omega'}
\end{equation}
Under condition \eqref{supp-om-inf}, the map
\begin{equation}\label{df-U(omega,omega')}
\tilde{\Phi}_{\omega;n}\mapsto\tilde{\Phi}_{\omega';n}
\end{equation}
is a $1$--to--$1$ correspondence between ortho--normal bases, hence it
uniquely extends by linearity and continuity to a unitary isomorphism
$$U_{\omega, \omega'}\colon \Gamma_{\omega}\to\Gamma_{\omega'}.$$
Since the convergence condition \eqref{df-om-z-coher-vec} is independent of the choice of the
ortho--normal basis, the two series \eqref{df-om-z-coher-vec} and \eqref{df-om-om'-z-coher-vec}
converge for the same set of $z$, i.e. $\mathbb{C}_{\omega}$. \\
The map $U_{\omega', \omega}\colon\Gamma_{\omega'}\to\Gamma_{\omega}$ defined
by \eqref{df-U(omega,omega')} with the roles of $\omega$ and $\omega'$ exchanged satisfies
\begin{equation}\label{df-U(omega,omega')b}
U_{\omega', \omega}\Psi_{\omega';z}
=\sum_{n=0}^{\infty}\frac{z^{n}}{\sqrt{\omega_{n}!}}U_{\omega', \omega}\tilde{\Psi}_{\omega';n}
=\sum_{n=0}^{\infty}\frac{z^{n}}{\sqrt{\omega_{n}!}}\tilde{\Psi}_{\omega;n}
\end{equation}
Since $U_{\omega, \omega'}$ is a unitary isomorphism, it follows that
$$
\langle\Psi_{ \omega';z}, \Psi_{\omega';z'}\rangle
=\langle\Psi_{\omega; z},\Psi_{\omega; z'}\rangle
\quad,\quad\forall z,z'\in\mathbb{N}
$$
In conclusion: the vectors $(\Psi_{\omega';z})_{z\in\mathbb{C}_{\omega}}$ are an
isomorphic copy, in $\Gamma_{\omega'}$, of the vectors $\Psi_{\omega; z}\in\Gamma_{\omega}$.\\


\subsection{Boson exponential vectors}

It is known (see \cite{[AcLu20-Nigeria21]})
that the $1$--mode boson Fock space $\Gamma_{Bos}$ is the $1$MIFS
corresponding to the principal Jacobi sequence $\omega_{n}= n$. In this case, the
$\tilde{\Phi}_{Bos,n}$ are the normalized Hermite polynomials and the normalized
exponential vectors take the familiar form
$$
\ \Psi _{Bos}\left(z\right)
:=\sum_{n=0}^{\infty}\frac{z^{n}}{\sqrt{n!}}\tilde{\Phi}_{Bos,n}
$$
They belong to $\Gamma_{Bos}$ for any $z\in\mathbb{C}$ and satisfy
\begin{equation}\label{scal-prod-AS-coher-vects10}
\left\langle \Psi _{Bos}\left(u\right)  ,\Psi _{Bos}\left(z\right)  \right\rangle
=\sum_{n=0}^{\infty}
\frac{\left(\overline{u}z\right)^{n}}{n!}=e^{\overline{u}z}\quad,\quad \forall u,z\in\mathbb{C}
\end{equation}
in particular
$$
\left\Vert \Psi _{Bos}\left(z\right)  \right\Vert ^{2}
=e^{\left\vert z\right\vert ^{2}}
$$
The corresponding coherent states are denoted
\begin{equation}\label{AS-coher-sts11}
\psi _{Bos;z}(\, \cdot \,)
\overset{\eqref{df-coher-st}}{=} e^{-\left\vert z\right\vert ^{2}}
\left\langle \Psi _{Bosz} , (\, \cdot \,)\Psi_{Bosz}  \right\rangle
\end{equation}

\subsubsection{Distribution of the number operator in a Boson coherent state }
\hfill\break

\begin{proposition}
If $\omega_{n}=n$ for any $n\in\mathbb{N}^{\ast}$, then the $\psi_{z}$--distribution
of $\Lambda$ is:\\
-- the 1 point distribution at the origin (i.e. $\delta_{0}$) if $z=0$;

\noindent-- the Poisson distribution with parameter
$\left\vert z\right\vert^{2}$ if $z\in\mathbb{C}\setminus\left\{0\right\}$.
\end{proposition}
\begin{proof} One has
\begin{equation}\label{num01b-2}
e^{it\Lambda}\Psi_{z} =\sum_{n=0}^{\infty} \frac{z^{n}}{\sqrt{n!}}e^{itn}\tilde{\Phi}_{\omega;n}
=\Psi_{e^{it}z}  ,\ \ \forall t\in\mathbb{R}\hbox{ and }z\in\mathbb{C}
\end{equation}

\begin{align}
& \psi _{z}\left(e^{it\Lambda}\right)  \overset
{{(\ref{AS-coher-sts11})}}{=}e^{-\left\vert z\right\vert ^{2}}\left\langle
\Psi_{z}  ,e^{it\Lambda}\Psi_{z}\right\rangle \overset{{(\ref{e(u-Lambda)Psi-z})}}{=}
e^{-\left\vert z\right\vert ^{2}}\left\langle \Psi_{z}
,\Psi_{e^{it}z}  \right\rangle \label{num01e}\\
&\overset{{(\ref{scal-prod-AS-coher-vects10})}}{=}
e^{-\left\vert z\right\vert ^{2}}e^{\left\vert z\right\vert ^{2}e^{it}}
=e^{\left\vert z\right\vert^{2}\left(e^{it}-1\right)},\ \ \forall t\in\mathbb{R}\nonumber
\end{align}
The thesis follows because the function
$t\longmapsto e^{\left\vert z\right\vert^{2}\left(e^{it}-1\right)}$
is the characteristic function of
$$
\begin{cases}
  \delta_{0}, & \mbox{if } z=0 \\
\hbox{the Poisson distribution with parameter} \left\vert z\right\vert ^{2}, & \mbox{if } z\neq0
\end{cases}
$$
\end{proof}

\subsection{Semi--circle exponential vectors}

\begin{corollary}
\upshape
If $\omega_{n}=\omega_{1}>0$ is constant for $n\geq 1$ (this case corresponds to the Semi--circle
distribution), then
$\Psi_{\omega;z}  $ is well defined for any
$$
z\in\mathbb{C}_{Sc,\omega_{1}}:=\left\{z\in\mathbb{C}:\left\vert z\right\vert <\sqrt{\omega_{1}}\right\}
$$
and one has, with the convention $x^0:=1$ for any $x\in\mathbb{C}$,
\begin{equation}\label{S-c-coher-vects}
\Psi_{z}:=\Psi_{\omega;z}
=\sum_{n=0}^{\infty}\left(\frac{z}{\sqrt{\omega_{1}}}\right)^{n}\tilde{\Phi}_{\omega;n}
\quad,\quad \forall z\in\mathbb{C}_{Sc,\omega_{1}}
\end{equation}
\begin{equation}\label{scal-prod-S-c-coher-vects}
\left\langle \Psi_{u}  ,\Psi_{z}  \right\rangle
=\frac{1}{1-\frac{\overline{u}z}{\omega_{1}}}
\quad, \quad \forall u,z\in\mathbb{C}_{Sc,\omega_{1}}
\end{equation}
\begin{equation}\label{norm-S-c-coher-vects}
\left\Vert \Psi_{z}  \right\Vert ^{2}
=\frac{1}{1-\frac{|z|^2}{\omega_{1}}}
\quad, \quad \forall z\in\mathbb{C}_{Sc,\omega_{1}}
\end{equation}
The corresponding coherent states are
\begin{equation}\label{S-c-coher-sts}
\psi _{z}\left(\cdot\right)
=\left(1-\frac{|z|^2}{\omega_{1}}\right)
\left\langle \Psi_{z},(\, \cdot \, )\Psi_{z}  \right\rangle
\quad,\quad\forall z\in\mathbb{C}_{Sc,\omega_{1}}
\end{equation}
\end{corollary}
\begin{proof}
\eqref{S-c-coher-vects} follows thanks to \eqref{df-om-z-coher-vec} and the fact that
$\omega_n!= (\omega_1)^n$. As a consequence, one obtains
$$
\left\langle \Psi_{u}  ,\Psi_{z}  \right\rangle
=\sum_{n=0}^{\infty}\left(\frac{\overline{u}z}
{\omega_{1}}\right)^{n}=\frac{1}{1-\frac
{\overline{u}z}{\omega_{1}}}
\quad,\quad \forall u,z\in\mathbb{C}_{Sc,\omega_1}
$$
which is \eqref{scal-prod-S-c-coher-vects} and, from this, \eqref{norm-S-c-coher-vects}
and \eqref{S-c-coher-sts} follow immediately.
\end{proof}

\subsubsection{Distribution of the number operator in a $\mu_{sc}$--coherent state }

\begin{corollary}
\upshape
If $\omega_{n}=\omega_1>0$ for all $n\geq1$, the $\phi_{Sc;z}$--distribution of $\Lambda$ is
\begin{itemize}
\item the $1$-point distribution at the origin (i.e. $\delta_{0}$), if $z=0$;

\item the geometric distribution with parameter $\frac{|z|^2}{\omega_{1}}$, if $z\not=0$.
\end{itemize}
\end{corollary}
\begin{proof}
$$
\psi _{z}\left(e^{it\Lambda}\right)
\overset{\eqref{S-c-coher-sts}}{=}
\left(1-\frac{|z|^2}{\omega_{1}}\right) \left\langle \Psi_{z}, (e^{it\Lambda})\Psi_{z}  \right\rangle
\overset{\eqref{e(u-Lambda)Psi-z}}{=}
\left(1-\frac{|z|^2}{\omega_{1}}\right)\left\langle \Psi_{z},\Psi_{e^{it}z}  \right\rangle
$$
$$
\overset{\eqref{scal-prod-S-c-coher-vects}}{=}
\left(1-\frac{|z|^2}{\omega_{1}}\right)\frac{1}{1-\frac{|z|^2}{\omega_{1}}e^{it}}
= \int_{\mathbb{R}}e^{itz}d\mu_z
$$
where,
	$$
	\mu_{z}({\{k\}}):= \left(1-\frac{|z|^2}{\omega_{1}}\right)\left(\frac{|z|^{2}}{\omega_{1}}\right)^{k}
	\quad,\quad\forall k\in\mathbb{N}
	$$
	i.e. the geometric distribution with parameter $\frac{|z|^2}{\omega_{1}}$.
\end{proof}

\subsection{Generalized Arc--sine distributions and their exponential vectors}

The family of generalized Arc--sine distributions is parametrized by two different strictly
positive numbers $\omega_{1}\ne  \omega_{2}$.
\begin{definition}\label{df:gen-gen-Arc-sin-exp-vects}

The \textbf{generalized Arc--sine distribution} with parameters $(\omega_{1}, \omega_{2})$
is a probability distribution whose principal Jacobi sequence $\omega\equiv (\omega_{n})$
satisfies the condition
\begin{equation}\label{df-gen-Arc-sin-PJS}
\omega_{1} \ne \omega_{2} = \omega_{n} > 0 \quad,\quad  \forall n\geq2
\end{equation}
If, in addition to \eqref{df-gen-Arc-sin-PJS}, also the condition
\begin{equation}\label{Arc-sin-PJS}
\omega_{1} = 2\omega_{2}
\end{equation}
is satisfied, then one speaks of the \textbf{Arc--sine distribution} with parameter $\omega_{1}$.
\end{definition}

\begin{proposition}\label{prop:gen-Arc-sin-exp-vects}
The exponential vectors $\Psi_{z}\equiv \Psi_{\omega_{1}, \omega_{2};z}$ of the generalized
Arc--sine distribution with parameters $(\omega_{1}, \omega_{2})$ are well defined for any
$$
z\in\mathbb{C}_{As,\omega_{1}, \omega_{2}}
:=\left\{z\in\mathbb{C}:\left\vert z\right\vert <\sqrt{\omega_{2}}\right\}
$$
and one has, writing simply $\tilde{\Phi}_{\omega;n}$ for the orthogonal polynomial
$\tilde{\Phi}_{\omega_{1}, \omega_{2};n}$,
\begin{equation}\label{A-s-coher-vects}
\Psi_{z}
=\tilde{\Phi}_{\omega;0} + \frac{z}{\sqrt{\omega_{1}}}\left(\tilde{\Phi}_{\omega;1}
+\sum_{n=2}^{\infty}\left(\frac{z}{\sqrt{\omega_{2}}}\right)^{n-1}\tilde{\Phi}_{\omega;n}\right)
\end{equation}
\begin{equation}\label{scal-prod-AS-coher-vects}
\left\langle \Psi_{u}  ,\Psi_{z}  \right\rangle
=\frac{\omega_{1}\omega_{2}-\left(\omega_{1}-\omega_{2}\right)
\overline{u}z}{\omega_{1}\left(\omega_{2}-\overline{u}z\right)}
\quad, \quad \forall u,z\in\mathbb{C}_{As}
\end{equation}
\begin{equation}\label{norm-A-s-coher-vects}
\left\Vert \Psi_{z}  \right\Vert ^{2}
=\frac{\omega_{1}\omega_{2}-\left(\omega_{1}-\omega_{2}\right)
\left\vert z\right\vert ^{2}}{\omega_{1}\left(\omega_{2}-\left\vert z\right\vert^{2}\right)}
\quad, \quad \forall z\in\mathbb{C}_{As}
\end{equation}
The corresponding coherent states are
\begin{equation}\label{AS-coher-sts}
\psi _{z}\left(\cdot\right)
=\frac{\omega_{1}\left(\omega_{2}-\left\vert z\right\vert ^{2}\right)}
{\omega_{1}\omega_{2}-\left(\omega_{1}-\omega_{2}\right)
\left\vert z\right\vert ^{2}}\left\langle \Psi_{z},(\, \cdot \, )\Psi_{z}  \right\rangle
\quad,\quad\forall z\in\mathbb{C}_{As}
\end{equation}
\end{proposition}
\begin{proof}
\eqref{A-s-coher-vects} follows from
$$
\ \Psi_{z}
=\tilde{\Phi}_{\omega;0}+\frac{z}{\sqrt{\omega_{1}}}\tilde{\Phi}_{\omega;1}
+\sum_{n=2}^{\infty}\frac{z^{n}}{\sqrt{\omega_{1}}\sqrt{\omega_{2}}^{n-1}}\tilde{\Phi}_{\omega;n}
$$
$$
=\tilde{\Phi}_{\omega;0}+\frac{z}{\sqrt{\omega_{1}}}\left(\tilde{\Phi}_{\omega;1}
+\sum_{n=2}^{\infty}\left(\frac{z}{\sqrt{\omega_{2}}}\right)^{n-1}\tilde{\Phi}_{\omega;n}\right)
$$
From it one deduces that
\begin{align}
\left\langle \Psi_{u}  ,\Psi_{z} \right\rangle
& =1+\frac{\overline{u}z}{\omega_{1}}\left(1+\sum_{n=2}^{\infty}
\frac{\left(\overline{u}z\right)^{n-1}}{\omega_{2}^{n-1}}\right)  =1+\frac{\overline{u}z}{\omega_{1}}\sum_{m=0}^{\infty}
\frac{\left(\overline{u}z\right)^{m}}{\omega_{2}^{m}}
=1+\frac{\frac{\overline{u}z}{\omega_{1}}}{1-\frac{\overline{u}z}{\omega_{2}}}\nonumber\\
& =1+\frac{\omega_{2}}{\omega_{1}}\frac{\overline{u}z}{\omega_{2}-\overline
{u}z}=\frac{\omega_{1}\omega_{2}-\left(\omega_{1}-\omega_{2}\right)\overline{u}z}
{\omega_{1}\left(\omega_{2}-\overline{u}z\right)},\ \ \forall u,z\in\mathbb{C}_{As}\nonumber
\end{align}
and this proves \eqref{scal-prod-AS-coher-vects}. In particular
$$
\left\Vert \Psi_{z}  \right\Vert ^{2}
=1+\frac{\frac{\left\vert z\right\vert ^{2}}
{\omega_{1}}}{1-\frac{\left\vert z\right\vert ^{2}}{\omega_{2}}}
=1+\frac{\omega_{2}}{\omega_{1}}\frac{\left\vert z\right\vert ^{2}}
{\omega_{2}-\left\vert z\right\vert ^{2}}
=\frac{\omega_{1}\omega_{2}-\left(\omega_{1}-\omega_{2}\right)  \left\vert
z\right\vert ^{2}}{\omega_{1}\left(\omega_{2}-\left\vert z\right\vert^{2}\right)}
$$
which is \eqref{norm-A-s-coher-vects}.
For the corresponding coherent state, one has, for any $z\in\mathbb{C}_{As}$,
\begin{equation*}
\psi _{z}\left(\cdot\right)
\overset{\eqref{df-coher-st}}{=}
\frac{1}{\left\Vert \Psi_{z}  \right\Vert ^{2}}\left\langle
\Psi_{z}  ,\cdot\Psi_{z}  \right\rangle
\overset{\eqref{norm-A-s-coher-vects}}{=}
\frac{\omega_{1}\left(\omega_{2}-\left\vert z\right\vert ^{2}\right)}
{\omega_{1}\omega_{2}-\left(\omega_{1}-\omega_{2}\right)  \left\vert
z\right\vert ^{2}}\left\langle \Psi_{z}
,\cdot\Psi_{z}  \right\rangle
\end{equation*}
which proves \eqref{norm-A-s-coher-vects}.  \end{proof}

\subsubsection{Distribution of the number operator in a generalized $\mu_{as}$--coherent state}

\hfill\break

\noindent We need the following result.
\begin{lemma}\label{lm:Four-transf-gen-AS}
Suppose that $\mu_{z}$ is a probability measure on $\mathbb{R}$, with support equal
to $\mathbb{N}$, whose Fourier transform has the form
\begin{equation}\label{Four-transf-discr-meas}
\hat{\mu_{z}}(t):=\int_{\mathbb{R}}e^{itx}\mu_{z}(dx)
=a\left(1+\frac{be^{it}}{\omega_{1}\left(1-ce^{it}\right)}\right)
\end{equation}
with
\begin{equation}\label{conds-on-par-discr-meas}
b > 0 \quad;\quad c\in (0,1)\quad;\quad \omega_{1}> 0
\end{equation}
Then
\begin{equation}\label{struc-a-discr-meas}
\mu_{z}(\{0\}) = a=\frac{1}{1+\frac{b}{\omega_{1}\left(1-c\right)}}
\end{equation}
\begin{equation}\label{struc-mu-z-discr-meas}
\mu_{z}({\{n\}}) =\frac{b}{\omega_{1}\left(1+\frac{b} {\omega_{1}\left(1-c\right)}\right)}c^{n-1}
\quad ,\quad n\ge 1
\end{equation}
\end{lemma}
\begin{proof}
Since $|c|<1$,
\begin{equation}\label{num01f-1}
\hat{\mu_{z}}(t):=\int_{\mathbb{R}}e^{itx}\mu_{z}(dx)
= \mu_{z}({\{0\}}) + \sum_{n=1}^{\infty}e^{itn}\mu_{z}({\{n\}})
\end{equation}
$$
=a\left(1+\frac{be^{it}}{\omega_{1}\left(1-ce^{it}\right)}\right)
= a+\frac{abe^{it}}{\omega_{1}} \frac{1}{\left(1-ce^{it}\right)}
$$
$$
= a+\frac{abe^{it}}{\omega_{1}}\sum_{n=0}^{\infty}
c^{n}e^{itn}
= a+\frac{ab}{\omega_{1}}\sum_{n=0}^{\infty}
c^{n}e^{it(n+1)}
= a+\sum_{n=1}^{\infty}\frac{ab}{\omega_{1}}c^{n-1}e^{itn}
$$
Identifying the coefficients of $e^{itn}$ in \eqref{num01f-1}, one finds
	\begin{equation}\label{num01f-3}
	\mu_{z}(0) = a \quad;\quad \mu_{z}(n) = \frac{ab}{\omega_{1}}c^{n-1} \quad,\quad n\ge 1
	\end{equation}
and \eqref{struc-a-discr-meas} follows using the condition
	$$
	\hat{\mu_{z}}(0)=1
	\iff 1 = a\left(1+\frac{b}{\omega_{1}\left(1-c\right)}\right)
	\iff \frac{1}{\left(1+\frac{b}{\omega_{1}\left(1-c\right)}\right)}  = a
	$$
Using this, the second equality in \eqref{num01f-3} becomes \eqref{struc-mu-z-discr-meas}.
\end{proof}

\begin{remark}
The measure $\mu_{z}$, defined by \eqref{struc-mu-z-discr-meas} is indeed a probability
measure, in fact it is positive because of \eqref{conds-on-par-discr-meas} and
\begin{align*}
\mu_{z}({\{0\}}) + \sum_{n=1}^{\infty}\mu_{z}(\{n\})
&= \frac{1}{1+\frac{b}{\omega_{1}\left(1-c\right)}}
+\frac{b}{\omega_{1}\left(1+\frac{b}{\omega_{1}\left(1-c\right)}\right)}\sum_{n=1}^{\infty}c^{n-1}
\\&= \frac{1}{1+\frac{b}{\omega_{1}\left(1-c\right)}}
+\frac{b}{\omega_{1}\left(1+\frac{b}{\omega_{1}\left(1-c\right)}\right)}\frac{1}{1-c}
\\&=  \frac{1}{1+\frac{b}{\omega_{1}\left(1-c\right)}}
\left(1+\frac{b}{\omega_{1}(1-c)}\right)
=1
\end{align*}
\end{remark}
\begin{proposition}\label{prop:pert-geom-distr}
Let $\omega\equiv (\omega_{n})\equiv (\omega_{1},\omega_{2})$ be as in
\eqref{df-gen-Arc-sin-PJS} (i.e. strictly positive constant for $n\ge 2$) and denote
$\Lambda\equiv \Lambda_{\omega}$ the number operator on $\Gamma_{\omega}$.
Then, for any $\ t\in\mathbb{R}$ and $z\in\mathbb{C}_{As,\omega_{1}, \omega_{2}}$,
\begin{equation}\label{num01b}
e^{it\Lambda}\Psi_{\omega;z} = \Psi_{\omega;e^{it}z}
\end{equation}
\begin{equation}\label{num01f}
\psi_{\omega;z}\left(e^{it\Lambda}\right)
=\frac{\omega_{1}\left(\omega_{2}-\left\vert z\right\vert ^{2}\right)}
{\omega_{1}\omega_{2}-\left(\omega_{1}-\omega_{2}\right)\left\vert z\right\vert ^{2}}
\left(1+\frac{\left\vert z\right\vert ^{2}e^{it}}
{\omega_{1}\left(1-\frac{\left\vert z\right\vert ^{2}}{\omega_{2}}e^{it}\right)}\right)
\end{equation}
and the $\psi_{\omega;z}$--distribution of $\Lambda$ is
\begin{equation}\label{Lamb-psi-z-distr-Gen-AS0}
\begin{cases}
\mu_{z}({\{0\} })=
\frac{\omega_{1}\left(\omega_{2}-\left\vert z\right\vert ^{2}\right)}
{\omega_{1}\omega_{2}-\left(\omega_{1}-\omega_{2}\right)\left\vert z\right\vert ^{2}}\\
\mu_{z}({\{n\} }) =
\frac{\omega_{1}\left(\omega_{2}-\left\vert z\right\vert ^{2}\right)}
{\omega_{1}\omega_{2}-\left(\omega_{1}-\omega_{2}\right)\left\vert z\right\vert ^{2}}
\frac{\left\vert z\right\vert ^{2}}{\omega_{1}}
\left(\frac{\left\vert z\right\vert ^{2}}{\omega_{2}}\right)^{n-1}
\quad, \quad  n\ge 1
            \end{cases}
\end{equation}
\end{proposition}
\begin{proof}
\eqref{num01b} follows from \eqref{e(u-Lambda)Psi-z}.
To prove \eqref{num01f}, note that
\begin{align*}
\psi_{\omega;z}\left(e^{it\Lambda}\right)
&\overset{\hbox{(\ref{AS-coher-sts})}}{=}
\frac{\omega_{1}\left(\omega_{2}-\left\vert z\right \vert^{2}\right)}{\omega_{1}\omega_{2}-\left(\omega_{1} -\omega_{2}\right)\left\vert z\right\vert ^{2}}
\left\langle \Psi_{\omega;z} , e^{it\Lambda} \Psi_{\omega;z}  \right\rangle\nonumber\\
& \overset{\hbox{(\ref{num01b})}}{=}
\frac{\omega_{1}\left(\omega_{2}-\left\vert z\right\vert ^{2}\right)}{\omega_{1}\omega_{2}
-\left(\omega_{1}-\omega_{2}\right)  \left\vert z\right\vert ^{2}}\left\langle\Psi_{\omega;z}  ,\Psi_{\omega;e^{it}z} \right\rangle \nonumber\\
&  \overset{\hbox{(\ref{scal-prod-AS-coher-vects})}}{=}
\frac{\omega_{1}\left(\omega_{2}-\left\vert z\right \vert ^{2}\right)}{\omega_{1}\omega_{2}-\left( \omega_{1} -\omega_{2}\right)\left\vert z\right\vert ^{2}}\frac{\omega_{1}\omega_{2}-\left(\omega_{1}
-\omega_{2}\right)  \left\vert z\right\vert^{2}e^{it}}{\omega_{1}\left(\omega_{2}
-\left\vert z\right\vert ^{2}e^{it}\right)}\nonumber\\
&  =\frac{\omega_{1}\left(\omega_{2}-\left\vert z\right \vert^{2}\right)}{\omega_{1}\omega_{2}-\left(\omega_1 -\omega_2\right)\left\vert z\right\vert ^2} \left(1+\frac{\left\vert z\right\vert ^{2}e^{it}}
{\omega_{1}\left(1-\frac{\left\vert z\right\vert ^{2}}{\omega_{2}}e^{it}\right)}\right)  \nonumber
\end{align*}
To prove \eqref{Lamb-psi-z-distr-Gen-AS0} notice that, because of \eqref{num01f}, we are in
the conditions of Lemma \ref{lm:Four-transf-gen-AS} with the following choice of
the parameters:
$$
a=\frac{1}{1+\frac{b}{\omega_{1}\left(1-c\right)}
}\ ;\quad b=\left\vert z\right\vert ^{2}\ ;\quad
c=\frac{\left\vert z\right\vert ^{2}}{\omega_{2}}
$$
Then \eqref{struc-mu-z-discr-meas} becomes
\begin{equation}\label{Lamb-psi-z-distr-Gen-AS}
\begin{cases}
\mu_{z}({\{0\} })=  \frac{1}{1+\frac{\left\vert z\right\vert ^{2}}
{\omega_{1}\left(1-\frac{\left\vert z\right\vert ^{2}}{\omega_{2}}\right)}} \\
\mu_{z}(\{n\}) =
\frac{1}{1+\frac{\left\vert z\right\vert ^{2}}
{\omega_{1}\left(1-\frac{\left\vert z\right\vert ^{2}}{\omega_{2}}\right)}}
\frac{\left\vert z\right\vert ^{2}}{\omega_{1}}
\left(\frac{\left\vert z\right\vert ^{2}}{\omega_{2}}\right)^{n-1} , \  n\ge 1
            \end{cases}
\end{equation}
and, since
$$
1+\frac{\left\vert z\right\vert ^{2}}
{\omega_{1}\left(1-\frac{\left\vert z\right\vert ^{2}}{\omega_{2}}\right)}
=\frac{\omega_{1}\left(1-\frac{\left\vert z\right\vert ^{2}}{\omega_{2}}\right)
+\left\vert z\right\vert ^{2}}
{\omega_{1}\left(1-\frac{\left\vert z\right\vert ^{2}}{\omega_{2}}\right)}
=\frac{\omega_{1}
\left(\frac{\omega_{2}-\left\vert z\right\vert ^{2}}{\omega_{2}}\right)
+\left\vert z\right\vert ^{2}}
{\omega_{1}\left(\frac{\omega_{2}-\left\vert z\right\vert ^{2}}{\omega_{2}}\right)}
$$
$$
=\frac{\omega_{1}
\left(\omega_{2}-\left\vert z\right\vert ^{2}\right)
+\omega_{2}\left\vert z\right\vert ^{2}}
{\omega_{1}\left(\omega_{2}-\left\vert z\right\vert ^{2}\right)}
=\frac{\omega_{1}\omega_{2}
-\left(\omega_{1}-\omega_{2}\right)\left\vert z\right\vert ^{2}}
{\omega_{1}\left(\omega_{2}-\left\vert z\right\vert ^{2}\right)}
$$
\eqref{Lamb-psi-z-distr-Gen-AS} is equivalent to \eqref{Lamb-psi-z-distr-Gen-AS0}.
\end{proof}
\begin{corollary}
(\textbf{Arc-sine case}) If
$$
0<\omega_{1}/2=\omega_{2}= \omega_{n}   \quad,\quad \forall n\geq2
$$
then, for any
\begin{equation}\label{df-C(As)}
z\in\mathbb{C}_{As,\omega_{2}}
:=\left\{z\in\mathbb{C}\colon\left\vert z\right\vert<\sqrt{\omega_{2}}\right\}
\end{equation}
one has
\begin{equation}\label{psi(As;z)-char-fctn-Lambda}
\psi_{As;z}\left(e^{it\Lambda}\right)
=\frac{2\left(1-\frac{\vert z\vert^2} {\omega_2}\right)}
{2-\frac{\vert z\vert^2} {\omega_2}}
\left(1+\frac{\frac{\vert z\vert^2} {\omega_2} e^{it} }
{2\left(1-\frac{\vert z\vert^2} {\omega_2} e^{it}\right)}\right)
\end{equation}
and the $\psi_{As;z}$--distribution of $\Lambda$ is
\begin{equation}\label{Lamb-psi-z-distr-AS}
\mu_{z}({\{0\} })
=\frac{2\left(1-\frac{\vert z\vert^2} {\omega_2}\right)}{2- \frac{\vert z\vert^2} {\omega_2}}\ \ ;\quad
\mu_{z}({\{n\} })
=\frac{\left(1-{ \frac{\left\vert z\right\vert ^{2}}{\omega_2}}\right)}
{2 - \frac{\left\vert z\right\vert ^{2}}{\omega_2}}
\left(\frac{\left\vert z\right\vert ^{2}}{\omega_2}\right)^{n}
\end{equation}
\end{corollary}
\begin{proof}
If $\omega_{1}/2=\omega_{2}$,  \eqref{num01f} becomes
\begin{align*}
\psi _{z}\left(e^{it\Lambda}\right)
& =\frac{2\left(\omega_{2}-\left\vert z\right\vert ^{2}\right)}{2\omega_{2}
-\left\vert z\right\vert ^{2}}
\frac{2\omega_{2}-\left\vert z\right\vert ^{2}e^{it}} {2\left(\omega_{2}-\left\vert z\right\vert ^{2}e^{it}\right)}\\
&=\frac{2\left(\omega_{2}-\left\vert z\right\vert
^{2}\right)}{2\omega_{2}-\left\vert z\right\vert ^{2}}\left(1+\frac{\left\vert
z\right\vert ^{2}e^{it}}{2\left(\omega_{2}-\left\vert z\right\vert ^{2}e^{it}\right)
}\right)  \nonumber\\
& =\frac{2\left(1-\frac{\left\vert z\right\vert ^{2}}{\omega_{2}}\right)}
{2-\frac{\left\vert z\right\vert ^{2}}{\omega_{2}}}
\left(1+\frac{\frac{\left\vert z\right\vert ^{2}}{\omega_{2}}e^{it}}
{2\left(1-\frac{\left\vert z\right\vert ^{2}}{\omega_{2}} e^{it}\right)}\right) \nonumber
\end{align*}
which is \eqref{psi(As;z)-char-fctn-Lambda}.
If $2\omega_{2} = \omega_{1}$, then \eqref{Lamb-psi-z-distr-Gen-AS} becomes
\begin{align*}
\mu_{z}({\{0\} })&=\frac {\omega_{1}\left(\omega_{2}-\left\vert z\right\vert^{2} \right)} {\omega_{1}\omega_{2}-\left(\omega_{1}-
\omega_{2}\right)\left\vert z\right\vert ^{2}}
=\frac{2\omega_{2}\left(\omega_{2}-\left\vert z\right\vert ^{2}\right)}{2\omega_{2}^2- \omega_{2}\left\vert z\right\vert ^{2}} \\
&=\frac{2\left(\omega_{2}-\left\vert z\right\vert ^{2}\right)}{\omega_{2} + (\omega_{2}- \left\vert z\right\vert ^{2})}
=\frac{2\left(1-{\frac{\left\vert z\right\vert ^2}{\omega_2}} \right)}
{2- \frac{\left\vert z\right\vert ^{2}}{\omega_2} }
\end{align*}
Similarly, for $n\ge 1$,
$$
\mu_{z}(\{n\})
= \frac{\omega_{1}\left(\omega_{2}-\left\vert z\right\vert ^{2}\right)}{\omega_{1}\omega_{2} -\left(\omega_{1}-\omega_{2}\right)\left\vert z\right\vert ^{2}}
\frac{\left\vert z\right\vert ^{2}}{\omega_{1}}
\left(\frac{\left\vert z\right\vert ^{2}} {\omega_{2}}\right)^{n-1}
$$
$$
=\frac{\left(\omega_{2}-\left\vert z\right\vert ^{2}\right)}{2\omega_{2} -\left\vert z\right\vert^2}
\left(\frac{\left\vert z\right\vert ^{2}}{\omega_{2}}\right)^{n}
=\frac{\left(1-{ \frac{\left\vert z\right\vert ^{2}}{\omega_2}}\right)}
{2 - \frac{\left\vert z\right\vert ^{2}}{\omega_2}}
\left(\frac{\left\vert z\right\vert ^{2}}{\omega_2}\right)^{n}
$$
\end{proof}

\begin{remark}
Let $\mu$ be a probability measure on $\mathbb{R}$ with support equal
to $\mathbb{N}$. A \textbf{perturbation of $\mu$ at the origin} is a new probability measure
on $\mathbb{N}$ of the form
\begin{equation}\label{df-nu-x,mu}
\nu_{x,\mu}(\{k\} )
:=x\delta_{0,k} + y\mu(\{k\} ) \chi_{\mathbb{N}^*}(k)\ ,\quad x\in [0,1] \,, \ y\in (0,+\infty)
\end{equation}
 $\nu_{x,\mu}(\{k\})$ is positive and it is a probability measure if and only if
\begin{align}\label{y-fctn-of-x}
y= \frac{1-x}{1-\mu(\{0\})}
\end{align}
because
$$
\sum_{k\in {\mathbb N}}\nu_{x,\mu}(\{k\})=x + y\sum_{k=1}^{\infty}\mu(\{k\} )
=x + y(1-\mu({\{0\} }))
$$
and which equals to $1$ if and only if \eqref{y-fctn-of-x} holds. In this case, \eqref{df-nu-x,mu} becomes
\begin{equation}\label{df-nu-x,mu-2}
\nu_{x,\mu}(\{k\} )
:=x\delta_{0,k}+\frac{1-x}{1-\mu({\{0\} })}\mu(\{k\} )\chi_{\mathbb{N}^*}(k) \ ,\quad x\in [0,1]
\end{equation}
If $\mu$ is geometric with parameter $c\in [0,1]$, then
\begin{equation}\label{df-mu-c}
\mu(\{k\})= (1-c)c^{k} \quad,\quad\forall k\in\mathbb{N}
\end{equation}
and \eqref{df-nu-x,mu-2} becomes
$$
\nu_{x,\mu}(\{k\})
=x\delta_{0,k} +(1-x)\frac{1-c}{c}c^{k}\chi_{\mathbb{N}^*}(k) \quad,\quad x\in [0,1]
$$
This is the case discussed in Proposition \ref{prop:pert-geom-distr}. One recognizes that
$\nu_{x,\mu}$ is a convex combination of the atomic measure at $0$ with the normalized
restriction on $\mathbb{N}^*$ of the geometric distribution of parameter $c$.
\end{remark}

\section{$*$--Lie algebras generated by the CAP operators}

\subsection{Notations}

In this section, we use the following notations.

\begin{equation}\label{Phi-n-n<0}
\Phi_{n}:=\begin{cases}0, & \mbox{if } n<0 \\
\Phi_{0}{:=\text{the vacuum vector,} }& {\text{if } n=0 }\\
a^{+n}\Phi_{0}, & \mbox{if } n\in \mathbb{N}^*
\end{cases}
\end{equation}
is a sequence of orthogonal polynomials;
	
--  $\omega=\{\omega_n\}_{n\in \mathbb{N} }$ is the principal Jacobi sequence associated to
$(\Phi_{n})$;
	
--  $\Gamma^{0}_\omega $ is the linear span of the $(\Phi_{n})$;

--  Defining,
$$
\Phi_{m}\Phi_{n}^*\xi:=\langle\Phi_{n}, \xi\rangle\Phi_{m} \ , \ \forall \xi\in\Gamma^{0}_\omega
 \ , \ \forall m,n\in\mathbb{N}
$$
the family
\begin{equation}\label{df-Phi(m,n)}
\Phi_{m,n}:=\begin{cases}
\Phi_{m}\Phi_{n}^*, &\mbox{ if }0\leq n,m\le n_{\omega}\\
0, & \mbox{ if either }n<0,\hbox{ or }m<0,\mbox{ or }m>n_{\omega},\mbox{ or }n>n_{\omega}
\end{cases}
\end{equation}
is the system of (non--normalized) matrix units associated to the orthogonal
basis $(\Phi_{n})$ ;
	
-- $\Lambda$ is the number operator associated to the basis $(\Phi_{n})$;

--  the function $\big(\omega!\big)^{-1}\colon\mathbb{Z}\longrightarrow \mathbb{R}$
is defined as follows,
\begin{equation}\label{LieAlg01}
\big(\omega!\big)^{-1}(k) :=\begin{cases}
\big(\omega_{k}!\big) ^{-1}, & \hbox{ if }0\leq k\le n_{\omega}\\
0, & \hbox{ otherwise}\end{cases}
\, ,\quad \forall k\in\mathbb{Z} %
\end{equation}
and, for any function $F\colon\mathbb{Z}\longrightarrow \mathbb{C}$,
\begin{equation}\label{LieAlg01a}
F_{\Lambda}\Phi_{n}:= F_{n}\Phi_{n}
\end{equation}
	
--  $\mathcal{L}\big( a,a^{+} \big)$ denotes the $*$--Lie--algebra generated by the annihilator
$a$ and the creator $a^{+}$;

--  $\mathcal{L}\big( a,a^{+},\Phi_{0,0} \big)$ denotes the $*$--Lie--algebra generated by
the vacuum projection $\Phi_{0,0}$, the annihilator $a$ and the creator $a^{+}$.\\

\begin{remark}
For all $0\leq n\leq n_{\omega}$, one has $a\Phi_{n}=\omega_{n}\Phi_{n-1}$ (recall that $\Phi_{-1}:=0$) and $\left\langle \Phi_{n}, \Phi_{n}\right\rangle =\omega_{n}!$. Moreover
\begin{equation}\label{LieAlg01b}
\big( \omega_{\Lambda}!\big) ^{-1}\Phi_{k}:=\begin{cases}
\big(\omega_{k}!\big)^{-1}\Phi_{k}, & \hbox{ if }0\le k\leq n_{\omega}\\ 0, & \hbox{ otherwise}
\end{cases} \,  ,\quad \forall k\in\mathbb{Z}
\end{equation}
\end{remark}
\begin{proposition}\label{LieAlg02}
On $\Gamma^{0}_\omega$ the following equalities hold:
\begin{align}\label{LieAlg02a}
a=&\sum_{n\in\mathbb{N}}\Phi_{n-1}\Phi_{n}^*
\big(\omega_{\Lambda-1}!\big)^{-1}\notag\\
=& \big(\omega_{\Lambda}!\big)^{-1} \sum_{n\in\mathbb{N}}\Phi_{n-1}\Phi_{n}^*
=\sum_{n\in\mathbb{N}^*}\sqrt{\omega_{n}} \tilde{\Phi}_{n-1}\tilde{\Phi}_{n}^*
\end{align}
\begin{align}\label{LieAlg02b}
a^{+}=&\big(\omega_{\Lambda-1}!\big)^{-1}
\sum_{n\in\mathbb{N}}\Phi_{n}\Phi_{n-1}^*\notag\\
=&\sum_{n\in\mathbb{N}}\Phi_{n}\Phi_{n-1}^*
\big(\omega_{\Lambda}!\big) ^{-1}
=\sum_{n\in\mathbb{N}^*}\sqrt{\omega_{n}}
\tilde{\Phi}_{n}\tilde{\Phi}_{n-1}^*
\end{align}
\begin{equation}\label{comp-aa+-a+a}
aa^{+} =\omega_{\Lambda+1}\ ;\quad a^{+}a = \omega_{\Lambda}\ ;\quad \big[  a,a^+\big]=\omega_{\Lambda+1}-\omega_{\Lambda}=: \partial \omega_{\Lambda}
\end{equation}
\begin{equation}\label{LieAlg03g}
\big(\Phi_{m,n}\big)^*=\Phi_{n,m}
\quad,\quad \forall n,m\in\mathbb{N}
\end{equation}
\begin{equation}\label{LieAlg03e}
a\Phi_{m,n}=\omega_m\Phi_{m-1,n}\quad,\quad a^+\Phi_{m,n}=\Phi_{m+1,n}
\quad,\quad \forall n,m\in\mathbb{N}
\end{equation}
\begin{equation}\label{LieAlg03f}
\Phi_{m,n}a^+ = \omega_n\Phi_{m,n-1} \quad, \quad\Phi_{m,n}a = \Phi_{m,n+1}
\quad,\quad \forall n,m\in\mathbb{N}
\end{equation}
\begin{align}
\big[  a,\Phi_{m,n}\big] &=\omega_{m}\Phi_{m-1,n}  -\Phi_{m,n+1}
\quad,\quad \forall n,m\in\mathbb{N}
\label{LieAlg02c1}\\
\big[ a^{+},\Phi_{m,n}\big]
&=\Phi_{m+1,n} - \omega_{n}\Phi_{m,n-1}
\quad,\quad \forall n,m\in\mathbb{N}\label{LieAlg02c2}
\end{align}
\begin{equation}\label{LieAlg02d}
\big[ \Phi_{k,h},\Phi_{m,n}\big]
=\delta_{h,m}\big(  \omega_{m}!\big)  \Phi_{k,n}-\delta_{k,n}\big(\omega_{n}!\big)\Phi_{m,h}
\,,\quad \forall k,h,m,n\in\mathbb{N}
\end{equation}
\end{proposition}
\begin{proof}
The last equality in \eqref{LieAlg02a} follows from the definition $\tilde \Phi_n:={1\over\sqrt{\omega!}}\Phi_n$ for any $n$. One gets the first two equalities in \eqref{LieAlg02a} as follows: for any $k\in{\mathbb N}$,
$$a \Phi_{k}=\omega_{k}\Phi_{k-1}  = \big(\omega_{k-1}!\big)^{-1}\omega_{k}!\Phi_{k-1}
$$	
and, on the one hand,
\begin{align*}
&\sum_{n\in\mathbb{N}}\Phi_{n-1}\Phi_{n}^*
\big(\omega_{\Lambda-1}!\big)^{-1}\Phi_{k}
=\big(\omega_{k-1}!\big)^{-1}\sum_{n\in\mathbb{N}}
\Phi_{n-1}\Phi_{n}^*\Phi_{k} \\
=&\big(\omega_{k-1}!\big)^{-1}\sum_{n\in\mathbb{N}}
\Phi_{k-1}\delta_{n,k}\|\Phi_{k}\|^2=\big(\omega_{k-1}! \big)^{-1} \ \omega_{k}!\Phi_{k-1}
\end{align*}
and on the other hand
\begin{align*}
&\big(\omega_{\Lambda}!\big)^{-1}\sum_{n\in\mathbb{N}}
\Phi_{n-1}\Phi_{n}^*\Phi_{k}
=\big(\omega_{\Lambda}!\big)^{-1}\sum_{n\in\mathbb{N}}
\Phi_{k-1}\delta_{n,k}\|\Phi_{k}\|^2\\
=&\big(\omega_{\Lambda}!\big)^{-1}\Phi_{k-1}\omega_{k}!
=\big(\omega_{k-1}!\big)^{-1}\omega_{k}!\Phi_{k-1}
\end{align*}
\eqref{LieAlg02b} is obtained from \eqref{LieAlg02a} taking the adjoint of both sides.\\
The first two equalities of \eqref{comp-aa+-a+a} are known and they imply the third.\\
\eqref{LieAlg03g} and \eqref{LieAlg03e} follow from \eqref{df-Phi(m,n)} and \eqref{LieAlg03f}
is the adjoint of \eqref{LieAlg03e}.\\
\eqref{LieAlg02c1} follows from the fact that, for any $n,m\in\mathbb{N}$,
\begin{align*}
&\left[a,\Phi_{m,n}\right] =a\Phi_{m,n}-\Phi_{m,n}a
\overset{\eqref{df-Phi(m,n)}}{=}   a\Phi_{m}\Phi_{n}^* - \Phi_{m}\Phi_{n}^*a\notag\\
=&\omega_{m}\Phi_{m-1}\Phi_{n}^*-\Phi_{m}(a^+\Phi_{n})^*
=\omega_{m}\Phi_{m-1}\Phi_{n}^*-\Phi_{m}\Phi_{n+1}^*\\
=&\omega_{m}\Phi_{m-1,n} - \Phi_{m,n+1}
\end{align*}
Taking the adjoint of both sides of \eqref{LieAlg02c1} one obtains
\begin{equation*}
\big[\Phi_{n,m}, a^{+}\big]
=\omega_{m}\Phi_{n,m-1} - \Phi_{n+1,m}
\end{equation*}
which, exchanging the roles of $m$ and $n$, is equivalent to \eqref{LieAlg02c2}.\\
Finally, for any $k,h,m,n\in\mathbb{N}$
\begin{align*}
&\left[\Phi_{k,h},\Phi_{m,n}\right] =\Phi_{k,h} \Phi_{m,n}-\Phi_{m,n}\Phi_{k,h}\\
=&\Phi_{k}\Phi_{n}^*\left\langle \Phi_{h},\Phi_{m} \right\rangle -\Phi_{m}\Phi_{h}^*
\left\langle \Phi_{n},\Phi_{k}\right\rangle
=\delta_{h,m}\left(  \omega_{m}!\right)  \Phi_{k,n}-\delta_{k,n}\left(\omega_{n}!\right)  \Phi_{m,h}
\end{align*}
which is \eqref{LieAlg02d}.
\end{proof}

\begin{theorem}\label{LieAlg03}
\begin{equation}\label{thm:Phi(m,n)-gen-cal-L(a,a+,Phi(0,0))}
\mathcal{L}\big( a,a^{+},\Phi_{0,0}\big)
=\hbox{lin-sp.}\big\{  a,a^{+},\Phi_{m,n}:n,m\in\mathbb{N}\big\}
\end{equation}
\end{theorem}
\begin{proof}
Clearly
\begin{equation}\label{LieAlg02d2}
\mathcal{L}\big(  a,a^{+},\Phi_{0,0} \big)\subseteq \hbox{lin-sp.}\big\{  a,a^{+},\Phi_{m,n}:n,m\in\mathbb{N}\big\}
\end{equation}
So to prove \eqref{thm:Phi(m,n)-gen-cal-L(a,a+,Phi(0,0))} it is sufficient to prove that
$\big\{\Phi_{m,n}\colon n,m\in\mathbb{N}\big\}\subseteq
\mathcal{L}\big( a,a^{+},\Phi_{0,0}\big)$. To this goal notice that \eqref{LieAlg03f} implies
\begin{equation}\label{LieAlg03f-0}
a\Phi_{0,n}=0=\Phi_{n,0}a^+ \qquad\forall n\in\mathbb{N}
\end{equation}
Therefore, $\Phi_{0,1} \in \mathcal{L}\big( a,a^{+},\Phi_{0,0}\big)$ because
$$
\mathcal{L}\big( a,a^{+},\Phi_{0,0}\big)\ni
[a,\Phi_{0,0}]= a\Phi_{0,0} - \Phi_{0,0}a
\overset{\eqref{LieAlg03f},\eqref{LieAlg03f-0}}{=} - \Phi_{0,1}
$$

Suppose by induction that, for $k\le n$,
$\Phi_{0,k} \in \mathcal{L}\big( a,a^{+},\Phi_{0,0}\big)$. Then $\Phi_{0,n+1} \in \mathcal{L}\big( a,a^{+},\Phi_{0,0}\big)$ because
$$
\mathcal{L}\big( a,a^{+},\Phi_{0,0}\big)\ni
[\Phi_{0,n},a]= a\Phi_{0,n} - \Phi_{0,n}a
\overset{\eqref{LieAlg03f}, \eqref{LieAlg03f-0}}{=} - \Phi_{0,n+1}
$$
Therefore, by induction, for $n\in\mathbb{N}$,
$\Phi_{0,n} \in \mathcal{L}\big( a,a^{+},\Phi_{0,0}\big)$. Because of \eqref{LieAlg03g}
and of the fact that $\mathcal{L}\big( a,a^{+},\Phi_{0,0}\big)$ is a $*$--Lie algebra, also
$\Phi_{n,0} \in \mathcal{L}\big( a,a^{+},\Phi_{0,0}\big)$. Consequently for each
$m, n\in\mathbb{N}$,
\begin{align*}
\mathcal{L}\big( a,a^{+},\Phi_{0,0}\big)\ni&
[\Phi_{0,m},\Phi_{n,0}]
= \Phi_{0,m}\Phi_{n,0} - \Phi_{n,0}\Phi_{0,m} \\
=& \Phi_{0}\Phi_{m}^* \Phi_{n}\Phi_{0}^*
- \Phi_{n}\Phi_{0}^*\Phi_{0}\Phi_{m}^*
{=\delta_{m,n}\omega_n!\Phi_{0}\Phi_{0}^*
- \Phi_{n}\Phi_{m}^* }
\end{align*}

Taking $m\ne n$, one finds $\Phi_{n,m}\in\mathcal{L} \big( a,a^{+}, \Phi_{0,0}\big)$ because
$$
\mathcal{L}\big( a,a^{+},\Phi_{0,0}\big)\ni
- \Phi_{n}\Phi_{m}^* = - \Phi_{n,m}
$$
and, for $m = n$, one gets
$$
\mathcal{L}\big( a,a^{+},\Phi_{0,0}\big)\ni
[\Phi_{0,n},\Phi_{n,0}]
= \omega_{n}! \Phi_{0,0}- \Phi_{n,n}
$$
and since $\omega_{n}! \Phi_{0,0} \in\mathcal{L}\big(a,a^{+},\Phi_{0,0}\big)$,
this implies that, for any $n\in\mathbb{N}$, $\Phi_{n,n}\in\mathcal{L}\big( a,a^{+},\Phi_{0,0}\big)$.
 In conclusion, for any
$m, n\in\mathbb{N}$, $$\Phi_{n,m}\in\mathcal{L}\big( a,a^{+},\Phi_{0,0}\big).$$
\end{proof}

\begin{remark}
 From
$$
\big(\omega_{\Lambda+1}-\omega_{\Lambda}\big)\Phi_n
=\big(\omega_{n+1}-\omega_{n}\big)\Phi_n\ ,\quad\forall n\in\mathbb{N}
$$
one deduces
\begin{align}\label{LieAlg03a}
\omega_{\Lambda+1}-\omega_{\Lambda}
\overset{\eqref{comp-aa+-a+a}}=[a,a^+]{=}
\sum_{n=0}^{n_\omega}(\omega_{n}!)^{-1}\big(\omega_{n+1}-\omega_{n}\big)\Phi_{n,n}
\end{align}
In particular, if $n_\omega=\infty$ (equivalently, $\omega_n>0$ for any $n\ge1$),
\begin{align}\label{LieAlg03d}
[a,a^+]=\omega_{\Lambda+1}-\omega_{\Lambda}
&=\sum_{n=0}^\infty(\omega_{n}!)^{-1}\big(\omega_{n+1} -\omega_{n}\big)\Phi_{n,n}\notag\\
&=\sum_{n=0}^\infty \big(\omega_{n+1} -\omega_{n}\big)\tilde{\Phi}_{n,n}
\end{align}
Suppose that there exists $K_{\omega}\in\mathbb{N}^*$ such that $K_{\omega}\le n_{\omega}$ and
$$
\omega_{n}= \omega_{K_{\omega}} \ ,\quad\forall
n\ge K_{\omega}
$$
then \eqref{LieAlg03a} becomes
$$
\omega_{\Lambda+1}-\omega_{\Lambda}
=[a,a^+]
=\sum_{n=0}^{n_\omega}(\omega_{n}!)^{-1}\big(\omega_{n+1}-\omega_{n}\big)\Phi_{n,n}
$$
\begin{equation}\label{LieAlg03d-a}
=\sum_{n=0}^{K_{\omega}-1}(\omega_{n}!)^{-1}\big(\omega_{n+1} -\omega_{n}\big)\Phi_{n,n}
=\sum_{n=0}^{K_{\omega}-1}\big(\omega_{n+1} -\omega_{n}\big)\tilde{\Phi}_{n,n}
\end{equation}
\end{remark}
\begin{lemma}
\begin{equation}\label{[a,.]n(F(Lambda)}
[a, \ \cdot \ ]^{n}(F_{\Lambda})=\partial^{n} F_{\Lambda}a^{n}\quad,\quad \forall n\in{\mathbb N}^*
\end{equation}
\end{lemma}
\begin{proof}
It is known that
\begin{equation}\label{[FLambda,a+-]}
[F_{\Lambda}, a^{+}] = a^{+}\partial F_{\Lambda} \ ;\quad [a, F_{\Lambda}] = \partial F_{\Lambda}a
\end{equation}
The second equality in \eqref{[FLambda,a+-]} implies that, for any $k$, one has
$F_\Lambda a\Phi_k=F_{k-1}a\Phi_k$, and so
$$[a,F_\Lambda]\Phi_k=aF_\Lambda\Phi_k-F_\Lambda a\Phi_k
=F_k a\Phi_k-F_{k-1}a\Phi_k= \big(F_{\Lambda+1}- F_{\Lambda}\big)a\Phi_k=\partial F_{\Lambda}a\Phi_k
$$
The second equality in \eqref{[FLambda,a+-]} follows by taking the adjoints of both sides of the first.
Suppose by induction that \eqref{[a,.]n(F(Lambda)} holds for some $n$, then
\begin{align*}
[a, \ \cdot \ ]^{n+1}(F_{\Lambda})
=[a,[a, \ \cdot \ ]^{n}(F_{\Lambda})]
=& [a,\partial^{n} F_{\Lambda}a^{n}]\\
= &[a,\partial^{n} F_{\Lambda}]a^{n}\\
\overset{\eqref{[FLambda,a+-]}}{=}&\partial^{n+1} {F}_{\Lambda}a^{n+1}
\end{align*}
Thus, by induction, \eqref{[a,.]n(F(Lambda)} holds for every $n\in\mathbb{N}$.
\end{proof}

\begin{remark}
Since $\partial\omega_{\Lambda}\in\mathcal{L}\big( a,a^{+}\big)$, for any $n\in\mathbb{N}$,
$$
\mathcal{L}\big( a,a^{+}\big)\ni [a, \ \cdot \ ]^{n}(\omega_{\Lambda})
=\partial^{n}\omega_{\Lambda}a^{n}
$$
\end{remark}
\subsection{The Semi--circle and Generalized Arc--sine case}

\hfill\break

In this section we apply Theorem \ref{LieAlg03} to the following cases:\\
-- $\omega_n=\omega_1$ for any $n\ge1$ (Semi--circle case);\\
-- $\omega_1\ne\omega_2=\omega_n$ for any $n\ge 2$ (Generalized Arc--sine case)\\
-- $\omega_n=\omega_2=\frac{\omega_1}{2}$ for any $n\ge 2$ (Arc--sine case);

\subsubsection{The semi--circle case}

\hfill\break

\begin{proposition}\label{prop:LieAlg-semi-circ}
If $\omega_n=\omega_1$ for any $n\ge1$ (semi-circle case),
\begin{equation}\label{LieAlg03c}
\mathcal{L}\big(a,a^{+} \big)=\hbox{lin-sp.}\big\{  a,a^{+},\Phi_{m,n}: n,m\in\mathbb{N}\big\}
\end{equation}
\end{proposition}
\begin{proof}
If $\omega_n=\omega_1$ for any $n\ge1$, \eqref{LieAlg03a} gives
\begin{equation}\label{LieAlg03b}
\mathcal{L}\big(a,a^{+} \big)\ni [a,a^+]=\omega_{\Lambda+1}-\omega_{\Lambda}=
\sum_{n=0}^\infty(\omega_n!)^{-1}\big(\omega_{n+1}-\omega_{n}\big)\Phi_{n,n}=\omega_1\Phi_{0,0}
\end{equation}
So $\mathcal{L}\big(a,a^{+} \big)= \mathcal{L}\big(a,a^{+}, \Phi_{0,0}\big)$
and \eqref{LieAlg03c} follows from Theorem \ref{LieAlg03}.
\end{proof}

\subsubsection{The generalized Arc--sine case}

\begin{proposition}\label{prop:LieAlg-GAS}
\eqref{LieAlg03c} holds if $\omega_1\ne\omega_2= \omega_n$ for any $n\ge2$ (generalized arc--sine case) and if $\omega_2$ is  neither $2\omega_1$ nor ${3\over2}\omega_1$.
\end{proposition}
\begin{proof}
In this case, $K_{\omega}=2$ and so \eqref{LieAlg03d-a} gives
\begin{align}\label{LieAlg04-GAS}
\mathcal{L}\big(a,a^{+}\big)\ni\partial\omega_{\Lambda}
&=[a,a^+]=\sum_{n=0}^{K_{\omega}-1}(\omega_{n}!)
^{-1}(\omega_{n+1}-\omega_{n})\Phi_{n,n}\notag\\
&= \omega_1\Phi_{0,0} +\omega_{1}^{-1}(\omega_{2}-\omega_{1})\Phi_{1,1}
\end{align}
Therefore
\begin{align}\label{LieAlg04-GAS1}
\mathcal{L} \big(a,a^{+} \big)\ni& [a, \partial\omega_{\Lambda}]
=[a,\omega_1\Phi_{0,0}+\omega _{1}^{-1}(\omega_{2}-\omega_{1})\Phi_{1,1}]\notag\\
=&\omega_1[a,\Phi_{0,0}] +\omega_{1}^{-1} (\omega_{2}-\omega_{1})[a,\Phi_{1,1}]\notag\\
\overset{\eqref{LieAlg02c1}}{=}&
-\omega_1\Phi_{0,1} + \omega_{1}^{-1}(\omega_{2} -\omega_{1})(\omega_1\Phi_{0,1}-\Phi_{1,2})\notag\\
=& -\omega_1\Phi_{0,1}+(\omega_{2}-
\omega_{1})\Phi_{0,1}- \omega_{1}^{-1} (\omega_{2}-\omega_{1})\Phi_{1,2}\notag\\
=& (\omega_2-2\omega_1)\Phi_{0,1} - \omega_{1}^{-1}(\omega_{2} - \omega_{1})\Phi_{1,2}
\end{align}
\eqref{LieAlg04-GAS} and \eqref{LieAlg04-GAS1} imply that
\begin{align}\label{LieAlg04-GAS2a}
&\mathcal{L} \big(a,a^{+} \big)\ni
\big[\partial\omega_{\Lambda},[a, \partial\omega_{\Lambda}]\big] \notag\\
&=[\omega_1\Phi_{0,0} +\omega_{1}^{-1}(\omega_{2} -\omega_{1})\Phi_{1,1},
(\omega_2-2\omega_1)\Phi_{0,1} - \omega_{1}^{-1}(\omega_{2} - \omega_{1})\Phi_{1,2}]\notag\\
&=\omega_1 (\omega_2-2\omega_1)[\Phi_{0,0},\Phi_{0,1}]+\omega_{1}^{-1}(\omega_{2}-\omega_{1})
(\omega_2-2\omega_1)[\Phi_{1,1},\Phi_{0,1}]\notag\\
& -(\omega_2-\omega_1)[\Phi_{0,0},\Phi_{1,2}]
-\omega_{1}^{-2}(\omega_{2}- \omega_{1})^{2}[\Phi_{1,1},\Phi_{1,2}]
\end{align}
\eqref{LieAlg02d} says
$[\Phi_{0,0},\Phi_{0,1}]=\Phi_{0,1}$,
$[\Phi_{1,1},\Phi_{0,1}]=-\omega_1\Phi_{0,1}$,
$[\Phi_{1,1},\Phi_{1,2}]=\omega_1\Phi_{1,2}$ and
$[\Phi_{0,0},\Phi_{1,2}]=0$. So \eqref{LieAlg04-GAS2a} becomes to
\begin{align}\label{LieAlg04-GAS2}
&\mathcal{L} \big(a,a^{+} \big)\ni
\big[\partial \omega_{\Lambda},[a, \partial\omega_{\Lambda}]\big] \notag\\
=&\omega_1(\omega_2-2\omega_1)\Phi_{0,1}-
(\omega_{2}-\omega_{1}) (\omega_2- 2\omega_1)\Phi_{0,1}-\omega_{1}^{-1}(\omega_{2}- \omega_{1})^{2}\Phi_{1,2}\notag\\
=&-(\omega_2-2\omega_1)^2\Phi_{0,1}
-\omega_{1}^{-1}(\omega_{2}-
\omega_{1})^{2}\Phi_{1,2}
\end{align}
Taking a linear combination of \eqref{LieAlg04-GAS1} and \eqref{LieAlg04-GAS2}, one finds
\begin{align}\label{LieAlg04-GAS3}
\mathcal{L} \big(a,a^{+} \big)&\ni(\omega_2-2\omega_1) \Phi_{0,1}-\omega_{1}^{-1}(\omega_{2}- \omega_{1}) \Phi_{1,2}\notag\\
&+\frac{1}{\omega_{2}-\omega_{1}}
\Big((\omega_2-2\omega_1)^{2}\Phi_{0,1}+\omega_{1}^{-1} (\omega_{2}-\omega_{1})^{2}\Phi_{1,2}\Big)\notag\\
&=\frac{(\omega_2-2\omega_1)(2\omega_2-3\omega_1)}{\omega_{2}-\omega_{1}}\Phi_{0,1}
\end{align}
In conclusion, $\Phi_{0,1}\in\mathcal{L} \big(a,a^{+} \big)$ if $(\omega_2-2\omega_1)(2\omega_2-3\omega_1)\ne0$ (equivalently, $2\omega_1\ne\omega_2\ne{3\over2}\omega_1$; in other words, by denoting $\omega_2=(1+q)\omega_1$, if $q\notin\{1,{1\over2}\}$); consequently, $\Phi_{1,0}\in\mathcal{L} \big(a,a^{+} \big)$ since $\mathcal{L} \big(a,a^{+} \big)$ is a $*$--Lie algebra.
Finally, since $\mathcal{L} \big(a,a^{+} \big)\ni \Phi_{0,1} \ , \ \Phi_{1,0}$, then also
\begin{equation}\label{LieAlg04b1}
\mathcal{L} \big(a,a^{+} \big)\ni [\Phi_{0,1}, \Phi_{1,0}]\overset{\eqref{LieAlg02d}}=  \omega_{1}\Phi_{0,0} - \Phi_{1,1}
\end{equation}
Taking again linear combinations of \eqref{LieAlg04b1} and \eqref{LieAlg04-GAS}, one finds,
\begin{align}\label{LieAlg04b1a}
&\mathcal{L} \big(a,a^{+} \big)
\ni \omega_{1}^{-1}(\omega_{2}-\omega_{1})[\Phi_{0,1}, \Phi_{1,0}]+ [a,a^+]\notag\\
&=\omega_{1}^{-1}(\omega_{2}-\omega_{1})(\omega_{1}
\Phi_{0,0}- \Phi_{1,1})+ \omega_1\Phi_{0,0} +\omega_ {1}^{-1}(\omega_{2}-\omega_{1})\Phi_{1,1} \notag\\
&=(\omega_{2}-\omega_{1})\Phi_{0,0} + \omega_1\Phi_{0,0}
=\omega_{2}\Phi_{0,0}
\end{align}
and so $\Phi_{0,0}\in\mathcal{L}\big(a, a^{+}\big)$.
\end{proof}
Now we consider the case $\omega_n=c\omega_1$ with $c\in \big\{2,{3\over2}\big\}$.
\begin{proposition}\label{prop:LieAlg-GAS02}
\eqref{LieAlg03c} holds if  ${3\over2}\omega_1 =\omega_2= \omega_n$ for any $n\ge2$.
\end{proposition}
\begin{proof}
It follows from
\eqref{LieAlg04-GAS1} and its conjugate form that
\begin{align*}
(\omega_2-2\omega_1) \Phi_{0,1} - \omega_{1}^{-1} (\omega_{2} - \omega_{1})\Phi_{1,2}  =[a, \partial\omega_{\Lambda}]\in\mathcal{L} \big(a,a^{+} \big)
\end{align*}
and
\begin{align}
(\omega_2-2\omega_1) \Phi_{1,0} - \omega_{1}^{-1}(\omega_{2} - \omega_{1})\Phi_{2,1}
=[a, \partial\omega_{\Lambda}]^*
\in\mathcal{L} \big(a,a^{+} \big)
\end{align}
So
\begin{align}\label{LieAlg04-a2}
\mathcal{L} \big(a,a^{+} \big)\ni& \big[a,[a, \partial\omega_{\Lambda}]^*\big]= (\omega_2-2\omega_1) [a,\Phi_{1,0}] - \omega_{1}^{-1}(\omega_{2} - \omega_{1})[a,\Phi_{2,1}]\notag\\  \overset{\eqref{LieAlg02c1}}=&
(\omega_2-2\omega_1)(\omega_1\Phi_{0,0}-\Phi_{1,1})
- \omega_{1}^{-1}(\omega_{2} - \omega_{1})
(\omega_2\Phi_{1,1}-\Phi_{2,2})\notag\\
=&(\omega_2-2\omega_1)\omega_1\Phi_{0,0}
-(\omega_2-2\omega_1+\omega_{1}^{-1}\omega_2(\omega_{2} - \omega_{1}) )\Phi_{1,1}\notag\\
&+\omega_{1}^{-1}(\omega_{2} - \omega_{1}) \Phi_{2,2}
\end{align}
and
\begin{align}\label{LieAlg04-a3}
&\mathcal{L} \big(a,a^{+} \big)\ni \big[[a, \partial \omega_{\Lambda}],[a, \partial\omega_{\Lambda}]^*\big]
\notag\\
&=(\omega_2-2\omega_1)^2[\Phi_{0,1},\Phi_{1,0}]
-(\omega_2-2\omega_1)\omega_1^{-1}(\omega_2-\omega_1) [\Phi_{0,1},\Phi_{2,1}]\notag\\
&-(\omega_2-2\omega_1)\omega_1^{-1}(\omega_2-\omega_1) [\Phi_{1,2},\Phi_{1,0}]
+\omega_1^{-2}(\omega_2-\omega_1)^2 [\Phi_{2,1},\Phi_{1,2}]\notag\\
&\overset{\eqref{LieAlg02d}}=
(\omega_2-2\omega_1)^2(\omega_1\Phi_{0,0}-\Phi_{1,1})
+\omega_1^{-2}(\omega_2-\omega_1)^2
(\omega_1\Phi_{2,2}-\omega_1\omega_2\Phi_{1,1}) \notag\\
&=(\omega_2-2\omega_1)^2\omega_1\Phi_{0,0}
-\big((\omega_2-2\omega_1)^2 +\omega_1^{-1}\omega_2 (\omega_2-\omega_1)^2\big)\Phi_{1,1}\notag\\
&\ +\omega_1^{-1}(\omega_2-\omega_1)^2\Phi_{2,2}
\end{align}
Therefore
\begin{align}\label{LieAlg04-a4}
&\mathcal{L} \big(a,a^{+} \big)\ni (\omega_2-\omega_1) \big[a,[a, \partial\omega_{\Lambda}]^*\big]
-\big[[a, \partial \omega_{\Lambda}],[a, \partial\omega_{\Lambda}]^*\big]\notag\\
=&\omega_1^2(\omega_2-2\omega_1)\Phi_{0,0}
-\omega_1(\omega_2-2\omega_1)\Phi_{1,1}\notag\\
=&(\omega_2-2\omega_1)(\omega_1\Phi_{0,0}-\Phi_{1,1})
\end{align}
i.e.
\begin{align}\label{LieAlg04-a5}
\omega_1\Phi_{0,0}-\Phi_{1,1}\in \mathcal{L} \big(a,a^{+} \big)
\end{align}
if $\omega_2\ne 2\omega_1$ (in particular, if $\omega_2={3\over2} \omega_1$). Consequently,
\begin{align}
&\mathcal{L} \big(a,a^{+} \big)\ni \omega_1^{-1}(\omega_2-\omega_1)(\omega_1\Phi_{0,0}-\Phi_{1,1})+ [a,a^+]\notag\\
\overset{\eqref{LieAlg04-GAS}}=&
\omega_1^{-1}(\omega_2-\omega_1)(\omega_1\Phi_{0,0}
-\Phi_{1,1})+\omega_1\Phi_{0,0} +\omega_{1}^{-1} (\omega_{2}-\omega_{1})\Phi_{1,1} \notag\\
=&\omega_{2}\Phi_{0,0}
\end{align}
This gives $\Phi_{0,0}\in\mathcal{L}\big(a, a^{+}\big)$.
\end{proof}
Finally, we consider the case $\omega_2 = 2\omega_1$.
\begin{proposition}\label{prop:LieAlg-GAS-om2=2om1}
\eqref{LieAlg03c} holds if
\begin{equation}\label{cond-om2=2om1}
\omega_2 = 2\omega_1= \omega_n > 0 \quad,\quad \forall n\ge2
\end{equation}
\end{proposition}
\begin{proof}
Under assumption \eqref{cond-om2=2om1}, \eqref{LieAlg04-GAS} becomes
\begin{equation}\label{LieAlg04-GAS-om2=2om1}
\mathcal{L}\big(a,a^{+}\big)\ni [a,a^{+}]=\partial\omega_{\Lambda}
= \omega_1\Phi_{0,0} + \Phi_{1,1}
\end{equation}
Therefore
$$
\mathcal{L}\big(a,a^{+}\big)\ni
[a, \ , \cdot \, ]^{2}(a^{+})
= [a, \omega_1\Phi_{0,0} + \Phi_{1,1}]
= \omega_1[a, \Phi_{0,0}] + [a, \Phi_{1,1}]
$$
$$
= \omega_1a\Phi_{0,0} - \omega_1\Phi_{0,0}a + a\Phi_{1,1} - \Phi_{1,1}a
= - \omega_1\Phi_{0,1} + \omega_1\Phi_{1,1} - \Phi_{1,2}
$$
\begin{equation}\label{LieAlg04-GAS-om2=2om1b}
= -\Phi_{1,2}
\end{equation}
Hence also
$$
\mathcal{L}\big(a,a^{+}\big)\ni
[a, \ , \cdot \, ]^{3}(a^{+}) = -[a, \Phi_{1,2}]
= - \omega_1\Phi_{0,2} + \Phi_{1,3}
$$
$$
= (-1)^{3-1} \left(- (3-2)\omega_1\Phi_{0,3-1} + \Phi_{1,3}\right)
$$
Suppose by induction that, for a given $n\in\mathbb{N}$ with $n\ge 2$, one has
\begin{equation}\label{LieAlg04-GAS-om2=2om1c}
\mathcal{L}\big(a,a^{+}\big)
\ni [a, \ , \cdot \, ]^{n}(a^{+})
= (-1)^{n-1} \left(- (n-2)\omega_1\Phi_{0,n-1} + \Phi_{1,n}\right)
\end{equation}
Then
$$
[a, \ , \cdot \, ]^{n+1}(a^{+})
= (-1)^{n-1} [a,\left(- (n-2)\omega_1\Phi_{0,n-1} + \Phi_{1,n}\right)]
$$
$$
= (-1)^{n-1}\left(- (n-2)\omega_1 [a,\Phi_{0,n-1}] + [a,\Phi_{1,n}]\right)
$$
$$
= (-1)^{n-1}\left( (n-2)\omega_1 \Phi_{0,n} + \omega_1\Phi_{0,n} - \Phi_{0,n+1}\right)
$$
$$
= (-1)^{n-1}\left( ((n+1)-2)\omega_1 \Phi_{0,n} - \Phi_{0,n+1}\right)
$$
$$
= (-1)^{n}\left(- ((n+1)-2)\omega_1 \Phi_{0,n} + \Phi_{0,n+1}\right)
$$
Thus by induction, \eqref{LieAlg04-GAS-om2=2om1c} holds for each $n\in\mathbb{N}$ with $n\ge 2$ and
implies that
\begin{equation}\label{LieAlg04-GAS-om2=2om1d}
- (n-2)\omega_1\Phi_{0,n-1} + \Phi_{1,n}\in\mathcal{L}\big(a,a^{+}\big)
\quad,\quad \forall n\ge 2
\end{equation}
Since $\mathcal{L}\big(a,a^{+}\big)$ is a $*$--Lie algrebra, \eqref{LieAlg04-GAS-om2=2om1b} implies
that $\Phi_{2,1}\in\mathcal{L}\big(a,a^{+}\big)$.
Therefore, for any $n\ge 4$,
$$
\mathcal{L}\big(a,a^{+}\big)
\overset{\eqref{LieAlg04-GAS-om2=2om1c}}{\ni} \
[\Phi_{2,1}, \left(- (n-2)\omega_1\Phi_{0,n-1} + \Phi_{1,n}\right)]
= [\Phi_{2,1}, \Phi_{1,n}]
= \omega_1 \Phi_{2,n}
$$
So
\begin{equation}\label{LieAlg04-GAS-om2=2om1e}
\Phi_{2,n}\in\mathcal{L}\big(a,a^{+}\big) \quad,\quad\forall n\ge 4
\end{equation}
Therefore also
$$
[\Phi_{1,2}, \Phi_{2,n}] = \Phi_{1,2}\Phi_{2,n} - \Phi_{2,n}\Phi_{1,2}
= \omega_2!\Phi_{1,n}\in\mathcal{L}\big(a,a^{+}\big)
\quad,\quad\forall n\ge 4
$$
and this implies that
\begin{equation}\label{LieAlg04-GAS-om2=2om1f}
\Phi_{1,n} \ , \ \Phi_{n,1}\in\mathcal{L}\big(a,a^{+}\big)
\quad,\quad\forall n\ge 4
\end{equation}
Then \eqref{LieAlg04-GAS-om2=2om1d} implies that, for any $n\ge 4$,
$- (n-2)\omega_1\Phi_{0,n-1} \in\mathcal{L}\big(a,a^{+}\big)$, so that
\begin{equation}\label{LieAlg04-GAS-om2=2om1g}
\Phi_{0,n} \in\mathcal{L}\big(a,a^{+}\big)  \quad,\quad\forall n\ge 3
\end{equation}
\eqref{LieAlg04-GAS-om2=2om1f} and \eqref{LieAlg04-GAS-om2=2om1g} imply that
$$
\mathcal{L}\big(a,a^{+}\big)\ni
[\Phi_{0,n}, \Phi_{n,1}]
= \Phi_{0,n}\Phi_{n,1} - \Phi_{n,1}\Phi_{0,n}
= \omega_n!\Phi_{0,1}
$$
so that
\begin{equation}\label{LieAlg04-GAS-om2=2om1h}
\Phi_{0,1} \ , \ \Phi_{1,0} \in\mathcal{L}\big(a,a^{+}\big)
\end{equation}
This implies that
\begin{equation}\label{LieAlg04-GAS-om2=2om1i}
\mathcal{L}\big(a,a^{+}\big)\ni
[\Phi_{0,1} \ , \ \Phi_{1,0} ] = \omega_1\Phi_{0,0} - \Phi_{1,1}
\end{equation}
Adding \eqref{LieAlg04-GAS-om2=2om1} and \eqref{LieAlg04-GAS-om2=2om1i}, one concludes that
\begin{equation}\label{LieAlg04-GAS-om2=2om1j}
\mathcal{L}\big(a,a^{+}\big)\ni 2 \omega_1\Phi_{0,0}
\end{equation}
which is equivalent to $\Phi_{0,0} \in\mathcal{L}\big(a,a^{+}\big)$. The thesis now follows from
Theorem \ref{LieAlg03}.
\end{proof}

\subsubsection{The Arc--sine case}

\hfill\break

\begin{corollary}\label{cor:LieAlg-arc-sine}
In the Arc--sine case, i.e. $\omega_n=\omega_2={\omega_1\over 2}$ for any $n\ge2$,
\eqref{LieAlg03c} holds.
\end{corollary}
\begin{proof}
In this case, the conditions of Proposition \ref{prop:LieAlg-GAS} are satisfied.
\end{proof}

\appendix
\section{Vacuum distribution of $X_{As}=a_{As}+a_{As}^{+}$}\label{sec:Vac-distr-arcsin}

Let $\xi$ be a symmetric one dimensional random variable  with the quantum decomposition
$a+a^{+} $ and $\Phi$ be the vacuum vector in $\Gamma_\omega$. According to \cite{[AcLu16d-q-mom-prob]} one has, for any $k\in\mathbb{N}$,
\begin{align}
\mathbf{E}\left(\xi^{k}\right)     =\left\langle \Phi,\left(a+a^{+}\right)^{k}\Phi\right\rangle
&=\sum_{\varepsilon\in\left\{
-1,1\right\}_{+}^{k}}\left\langle \Phi,a^{\varepsilon\left(1\right)
}\ldots a^{\varepsilon\left(k\right)}\Phi\right\rangle \label{01}
\\& =\begin{cases}
	\sum_{\varepsilon\in\left\{-1,1\right\}_{+}^{2n}}\prod_{h=1}^{n}
	\omega_{2h-l_{h}^{\left(\varepsilon\right)}},& {\rm if\ k=2n}\\
	0,& {\rm otherwise}
\end{cases}\nonumber
\end{align}
where%
$$
\left\{-1,1\right\}_{+}^{2n}:=\left\{\varepsilon\in\left\{
-1,1\right\}^{2n}:\sum_{h=1}^{2n}\varepsilon\left(h\right)  =0\hbox{ and
}\sum_{h=1}^{m}\varepsilon\left(h\right)  \leq0\hbox{ for any }%
m\leq2n\right\}
$$%
$$
\left\{l_{h}^{\left(\varepsilon\right)}\right\}_{h=1}^{n}%
:=\varepsilon^{-1}\left(\left\{-1\right\}  \right)  \hbox{ with the order
}l_{1}^{\left(\varepsilon\right)}<\ldots<l_{n}^{\left(\varepsilon
\right)}%
$$
\begin{itemize}
\item In the arc--sine case, for any $\varepsilon\in\left\{-1,1\right\}_{+}^{2n}$%
\begin{equation}\prod_{h=1}^{n}
\omega_{2h-l_{h}^{\left(\varepsilon\right)}}=\omega_{1}^{m\left(
\varepsilon\right)}\omega_{2}^{n-m\left(\varepsilon\right)}\label{02}%
\end{equation}
with%
$$
m\left(\varepsilon\right) :=\left\vert \left\{m\in\left\{1,\ldots
,n\right\} :\sum_{h=1}^{2m}\varepsilon\left(h\right)  =0\right\}
\right\vert
$$
It is clear that, for any $n\in\mathbb{N}^{\ast}$ and
$\varepsilon\in\left\{-1,1\right\}_{+}^{2n}$

\begin{itemize}
\item $1\leq m\left(\varepsilon\right)  \leq n;$

\item $m\left(\varepsilon\right)  =n$ if and only if $\varepsilon\left(2h\right)
=-1$ and $\varepsilon\left(2h-1\right)  =-1$ for all $h\in\left\{
1,\ldots,n\right\}  ;$

\item $m\left(\varepsilon\right)  =1$ if and only if the unique non--crossing pair
partition $\left\{\left(l_{h}^{\left(\varepsilon\right)}%
,r_{h}^{\left(\varepsilon\right)}\right)  \right\}_{h=1}^{n}$ verifies
$r_{1}^{\left(\varepsilon\right)}=2n.$
\end{itemize}

We are interested now in the calculation of
\begin{equation}
\sum_{\varepsilon\in\left\{
	-1,1\right\}_{+}^{2n}}\prod_{h=1}^{n}\omega_{2h-l_{h}^{\left(
		\varepsilon\right)}}=\sum_{\varepsilon\in\left\{ -1,1\right\}_{+}^{2n}}\omega_{1}^{m\left(
\varepsilon\right)}\omega_{2}^{n-m\left(\varepsilon\right)}
\label{as21914-01a}%
\end{equation}

\begin{theorem}\label{as21914-01}
	Let, for any $m\in\left\{1,\ldots,n\right\}  ,$
$$
\left\{-1,1\right\}_{+,m}^{2n}
:=\left\{\varepsilon\in\left\{-1,1\right\}_{+}^{2n}:m\left(\varepsilon\right)  =m\right\}
$$
Then

1) $\left\{-1,1\right\}_{+,1}^{2n},\ldots,\left\{-1,1\right\}_{+,n}^{2n}$
are pairwise disjoint and
$\left\{-1,1\right\}_{+}^{2n}=\bigcup_{m=1}^{n}\left\{ -1,1\right\}_{+,m}^{2n}$.

2) $\left\vert \left\{-1,1\right\}_{+,1}^{2n}\right\vert =C_{n-1}$ and for
any $m\in\left\{2,\ldots,n\right\}  ,$
\begin{equation}\label{as21914-01b}
\left\vert \left\{-1,1\right\}_{+,m}^{2n}\right\vert
=\sum_{\substack{k_{1},\ldots,k_{m-1}\geq1\\k_{1}+\ldots+k_{m-1}\leq n-1}}
C_{n-k_{1}-\ldots-k_{m-1}-1}\prod_{h=1}^{m-1}C_{k_{h}-1}
\end{equation}
where we recall that
\begin{align*}
C_{k}=&	\left\vert \left\{
\hbox{non--crossing pair partitions on }\left\{1,2,\ldots,2k\right\}
\right\}  \right\vert\\
=&{\frac {1}{k+1}}{2k \choose k}={\frac {(2k)!}{(k+1)!k!}},\quad k\geqslant 0
\end{align*}
 is the $k-th$ Catalan number  and $C_{0}:=1.$

3) the expression (\ref{as21914-01a}) is equal to
\begin{equation}\label{as21914-01c}
\omega_{1}\omega_{2}^{n-1}C_{n-1}
+\sum_{m=2}^{n}\omega_{1}^{m}\omega_{2}^{n-m}\sum_{\substack{k_{1},\ldots,k_{m-1}\geq1\\
k_{1}+\ldots+k_{m-1}\leq n-1}}C_{n-k_{1}-\ldots-k_{m-1}-1}\prod_{h=1}^{m-1}C_{k_{h}-1}
\end{equation}
\end{theorem}
\begin{proof}
Statement 1) follows from the definition of $\left\{-1,1\right\}_{+,m}^{2n}$.\\
In particular, one finds
\begin{align}
&\sum_{\varepsilon\in\left\{-1,1\right\}_{+}^{2n}}\omega_{1}^{m\left(
\varepsilon\right)}\omega_{2}^{n-m\left(\varepsilon\right)}
=\sum_{m=1}^{n}\sum_{\varepsilon\in\left\{-1,1\right\}_{+,m}^{2n}}%
\omega_{1}^{m\left(\varepsilon\right)}\omega_{2}^{n-m\left(
\varepsilon\right)}\\
=& \sum_{m=1}^{n}\omega_{1}^{m}\omega_{2}^{n-m}%
\sum_{\varepsilon\in\left\{-1,1\right\}_{+,m}^{2n}}1\label{as21914-01d} =\sum_{m=1}^{n}\omega_{1}^{m}\omega_{2}^{n-m}\left\vert \left\{
-1,1\right\}_{+,m}^{2n}\right\vert \nonumber
\end{align}
So, if statement 2) is proved, one gets the (\ref{as21914-01c}). Now we turn to see the affirmation 2). In order to prove statement 2), one considers \textbf{semi--circle} 1MIFS $\Gamma_{\left\{1\right\}_{n}}\left(\mathbb{C}\right) .$
On this space, denoting $b$ and $b^{+}$ the annihilation--creation operators, one has
$$
\left\langle \Phi,b^{\varepsilon\left(1\right)}\ldots b^{\varepsilon
\left(2n\right)}\Phi\right\rangle =1,\ \ \forall\varepsilon\in\left\{
-1,1\right\}_{+}^{2n}%
$$

First of all, we see $\left\vert \left\{-1,1\right\}_{+,1}^{2n}%
\right\vert ,$ i.e.%
$$
\sum_{\varepsilon\in\left\{-1,1\right\}_{+,1}^{2n}}\left\langle
\Phi,b^{\varepsilon\left(1\right)}\ldots b^{\varepsilon\left(2n\right)
}\Phi\right\rangle
$$
As said before Theorem \ref{as21914-01}, $\varepsilon\in\left\{-1,1\right\}_{+,1}^{2n}$
{ if and only if} $r_{1}^{\left(\varepsilon\right)}=2n.$
For any  $\varepsilon\in\left\{-1,1\right\}_{+,1}^{2n},$ one defines $\varepsilon^{\prime}:=\left\{
1,2,\ldots,2n-2\right\}  \longmapsto\left\{-1,1\right\}  $ as%
\begin{equation}\label{as21914-01e}
\varepsilon^{\prime}\left(k\right) :=\varepsilon\left(k+1\right)
,\ \ \forall k\in\left\{1,2,\ldots,2n-2\right\}
\end{equation}
then, as $\varepsilon$ runs over $\left\{-1,1\right\}_{+,1}^{2n},$
$\varepsilon^{\prime}$ runs over $\left\{-1,1\right\}_{+}^{2\left(n-1\right)}$ and
\begin{equation}
\left\langle \Phi,b^{\varepsilon\left(1\right)}\ldots b^{\varepsilon
\left(2n\right)}\Phi\right\rangle =\left\langle \Phi,bb^{\varepsilon
^{\prime}\left(1\right)}\ldots b^{\varepsilon^{\prime}\left(2n-2\right)
}b^{+}\Phi\right\rangle =\left\langle \Phi,b^{\varepsilon^{\prime}\left(
1\right)}\ldots b^{\varepsilon^{\prime}\left(2n-2\right)}\Phi
\right\rangle \label{as21914-01e0}%
\end{equation}
and the unique non--crossing $\left\{\left(l_{h}^{\left(\varepsilon
^{\prime}\right)},r_{h}^{\left(\varepsilon^{\prime}\right)}\right)
\right\}_{h=1}^{n-1}$ determined by $\varepsilon^{\prime}$ is in fact
$\left(l_{h}^{\left(\varepsilon^{\prime}\right)},r_{h}^{\left(
\varepsilon^{\prime}\right)}\right)  =\left(l_{h+1}^{\left(
\varepsilon\right)}-1,r_{h+1}^{\left(\varepsilon\right)}-1\right)  $ for
any $h\in\left\{1,\ldots,n-1\right\}  $. So%
\begin{align}
&\left\vert \left\{-1,1\right\}_{+,1}^{2n}\right\vert  =\sum
_{\varepsilon\in\left\{-1,1\right\}_{+,1}^{2n}}\left\langle \Phi
,b^{\varepsilon\left(1\right)}\ldots b^{\varepsilon\left(2n\right)
}\Phi\right\rangle \label{as21914-01e1}\\
=&\sum_{\varepsilon'\in\left\{-1,1\right\}_{+}^{2\left(n-1\right)}}
\left\langle \Phi,b^{\varepsilon^{\prime}\left(1\right)
}\ldots b^{\varepsilon^{\prime}\left(2n-2\right)}\Phi\right\rangle=\sum_{\varepsilon\in\left\{ -1,1\right\}_{+}^{2\left(n-1\right)}%
}1=\left\vert \left\{-1,1\right\}_{+}^{2\left(n-1\right)}\right\vert
=C_{n-1}\nonumber
\end{align}

For any $m\geq2$ and $\varepsilon\in\left\{-1,1\right\}_{+,m}^{2n},$ let
$k_{1}:=\min\left\{k:\sum_{h=1}^{2k}\varepsilon\left(h\right)  =0\right\}
$ and%
\begin{align}
\varepsilon^{\prime}\left(k\right)    &:=\varepsilon\left(k+1\right)
,\ \ \forall k\in\left\{1,2,\ldots,2k_{1}-2\right\}  ;\label{as21914-01e3}\\
\varepsilon^{\prime\prime}\left(k\right)    &:=\varepsilon\left(
2k_{1}+k\right)  ,\ \ \forall k\in\left\{1,2,\ldots,2n-2k_{1}\right\}
\nonumber
\end{align}
then

\begin{itemize}
\item $k_{1}\in\left\{1,\ldots,n-m\right\}  ;$

\item as $\varepsilon$ running over $\left\{-1,1\right\}_{+,m}^{2n},$
$\varepsilon^{\prime}$ runs over $\left\{-1,1\right\}_{+}^{2\left(
k_{1}-1\right)};$

\item as $\varepsilon$ running over $\left\{-1,1\right\}_{+,m}^{2n},$
$\varepsilon^{\prime\prime}$ runs over $\left\{-1,1\right\}_{+,m-1}%
^{2\left(n-k_{1}\right)};$

\item for any $\varepsilon\in\left\{-1,1\right\}_{+,m}^{2n},$ one has,
since $\sum_{h=1}^{2k_{1}}\varepsilon\left(h\right)  =0$
\begin{align}
&\left\langle \Phi,b^{\varepsilon\left(1\right)}\ldots b^{\varepsilon
\left(2n\right)}\Phi\right\rangle   =\left\langle \Phi,b^{\varepsilon
\left(1\right)}\ldots b^{\varepsilon\left(2k_{1}\right)}%
\Phi\right\rangle \left\langle \Phi,b^{\varepsilon\left(2k_{1}+1\right)
}\ldots b^{\varepsilon\left(2n\right)}\Phi\right\rangle\notag\\
=& \left\langle \Phi,b^{\varepsilon^{\prime}\left(1\right)}\ldots
b^{\varepsilon^{\prime}\left(2k_{1}-2\right)}\Phi\right\rangle
\left\langle \Phi,b^{\varepsilon^{\prime\prime}\left(1\right)}\ldots
b^{\varepsilon^{\prime\prime}\left(2n-2k_{1}\right)}\Phi\right\rangle
\label{as21914-01e4}
\end{align}

\end{itemize}

So
\begin{align}
\left\vert \left\{-1,1\right\}_{+,m}^{2n}\right\vert  & =\sum
_{\varepsilon\in\left\{-1,1\right\}_{+,m}^{2n}}\left\langle \Phi
,b^{\varepsilon\left(1\right)}\ldots b^{\varepsilon\left(2n\right)
}\Phi\right\rangle =\sum_{\varepsilon\in\left\{-1,1\right\}_{+,m}^{2k_{1}%
}}\left\langle \Phi,b^{\varepsilon\left(1\right)}\ldots b^{\varepsilon
\left(2n\right)}\Phi\right\rangle \label{as21914-01f}\\
& =\sum_{1\leq k_{1}\leq n-m}\left\vert \left\{-1,1\right\}_{+}^{2\left(
n-k_{1}\right)}\right\vert \left\vert \left\{-1,1\right\}_{+,m-1}%
^{2\left(n-k_{1}\right)}\right\vert =\sum_{1\leq k_{1}\leq n-m}C_{k_{1}%
-1}\left\vert \left\{-1,1\right\}_{+,m-1}^{2\left(n-k_{1}\right)
}\right\vert \nonumber
\end{align}
By combining this and  the induction argument, one gets (\ref{as21914-01b}).

\end{proof}

Now, using the first identity in Lemma \ref{Catalan} and the fact that $\omega_{1}=2\omega_{n}>0$ for any $n\geq2$, the expression \eqref{as21914-01c} reduces to
\begin{equation*} \omega_2^n\sum_{m=1}^n\frac{m2^m}{2n-m}\binom{2n-m}{n-m}=\left(\frac{\omega_1}{2}\right)^n\binom{2n}{n}
\end{equation*}
which is the $2n-th$ moment of the arc--sine distribution on the interval $\left( -\sqrt{2\omega_{1}},\sqrt{2\omega_{1}}\right)$.
\item In the semi--circle case ($\omega_{1}=\omega_{j}$ for all $j\geq
1$), for any $\varepsilon\in\left\{-1,1\right\}_{+}^{2n}$
$$\prod\limits_{h=1}^{n}
\omega_{2h-l_{h}^{\left(\varepsilon\right)}}=\omega_{1}^{n}
$$
Thus,
\begin{align*}
	\sum_{\varepsilon\in\left\{
		-1,1\right\}_{+}^{2n}}\prod_{h=1}^{n}\omega_{2h-l_{h}^{\left(
			\varepsilon\right)}}
		=\sum_{\varepsilon\in\left\{-1,1\right\}_{+}^{2n}}\omega_{1}^{n}
		&=\omega_{1}^{n}\left\vert \left\{-1,1\right\}_{+}^{2n}\right\vert
		\\&=\omega_{1}^{n}\sum_{m=1}^n\left\vert \left\{-1,1\right\}_{+,m}^{2n}\right\vert
		\\&=\omega_1^n\sum_{m=1}^n\frac{m}{2n-m}\binom{2n-m}{n-m}=\omega_{1}^{n}C_n
\end{align*}
which is nothing else than the $2n-th$ moment of the semi--circle distribution on the interval $\left(-2\sqrt{\omega_{1}},2\sqrt{\omega_{1}}\right)$.
\end{itemize}

\section{Neumann-type series and integral representations of $e^{itP}$}

The objective of this appendix is to formulate the action of the evolution $e^{itP}$, associated to the semi--circle and arc--sine cases, in terms of Neumann series of Bessel functions. In addition, an integral representation of this time evolution operator will be furnished.

The Neumann series of Bessel functions, as defined in \cite[Chapter XVI]{Watson}, is given by the equation
\begin{equation*}
	\mathcal{R}_\nu(z)=\sum_{m\ge 0}\alpha_mJ_{m+\nu}(z),\quad z\in\mathbb{C}	
\end{equation*}
where $\nu$ and $(\alpha_n)_{n\in\mathbb{N}}$ are constants.
\subsection{Semi--circle case}
Recall from \cite{[AcHamLu22a]} that
\begin{align}\label{Neumann1}
	e^{itP_{Sc}}\Phi_{h}(x) =e^{-tH_{\mu}}\Phi_{h}(x) =&\sum_{m\ge 0}(-1)^m\sum_{n=0}^{h}\left(J_{m+n}(2t)+J_{m+n+2}(2t)\right)\Phi_{h+m-n}(x)\nonumber
	\\=&\sum_{n=0}^{h+2}\sum_{m\ge 0}(-1)^m\alpha_m^{(n)}(x,h)J_{m+n}(2t)
\end{align}
where
\begin{align*}
	\alpha_m^{(n)}(x,h)=
	\begin{cases}
	\Phi_{m+h-n}(x)	,	& n\in\{0,1\}\\
	x\,\Phi_{m+h+1-n}(x),  &n\in\{2,3,\dots,h\} \\
	\Phi_{m+h+2-n}(x),	& n\in\{h+1,h+2\}
	\end{cases}
\end{align*}
Thus, \eqref{Neumann1} rewrites
\begin{align*}
e^{-tH_{\mu}}\Phi_{h}(x) =\sum_{n=0}^{h+2}\mathcal{R}_n(t,x,h)
\end{align*}
where
\begin{align*}
\mathcal{R}_n(t,x,h):=\sum_{m\ge 0}(-1)^m\alpha_m^{(n)}(x,h)J_{m+n}(2t)
\end{align*}
To streamline the notation, we will omit the variables $x$ and $h$ in $\mathcal{R}_n(t,x,h)$ and $\alpha_m^{(n)}(x,h)$. Instead, they will be represented simply as $\mathcal{R}_n(t)$ and $\alpha_m^{(n)}$, as long as there is no ambiguity.
As stated in  \cite{PS09}, for all $t$ in the interval
$$\left(0,\min\left\{1,\left(e\limsup_{m\rightarrow\infty}\frac{\sqrt[m]{|\alpha_m^{(n)}|}}{m}\right)^{-1}\right\}\right)=(0,1),$$
the Neumann-type series $\mathcal{R}_n(t)$ can be rewritten as:

\begin{equation*}
	\mathcal{R}_n=\sqrt{\frac{t}{\pi}}\int_1^{\infty}K_n(t,u)A_n(u)du
\end{equation*}
where the kernel $K_n(t,u)$ is given by
\begin{align}\label{kernel}
	K_n(t,u)&=-2\int_0^1\cos(2st)\left(t(1-s^2)\right)^{n+u-1/2}\ln\left(t(1-s^2)\right)ds
	\\&=-\sqrt{\frac{\pi}{t}}\frac{\partial}{\partial u}\left(\Gamma(n+u+\frac{1}{2})J_{n+u}(2t)\right)\nonumber
\end{align}
and its density is given by
\begin{align*}
	A_n(u)=\sum_{j=1}^{[u]}\frac{\alpha_j^{(n)}}{\Gamma(n+j+\frac{1}{2})}=\int_0^{[u]} \mathcal{D}_v\left(\frac{\alpha_v^{(n)}(t,x,h)}{\Gamma(n+v+\frac{1}{2})}\right)dv.
\end{align*}
Here the operator $\mathcal{D}_v$ is defined as:

\begin{equation*}
	\mathcal{D}_v=1+\{v\}\frac{d}{dx},
\end{equation*}

and for any real number $x$, $[x]$ represents the integer part of $x$, while $\{x\}$ represents its fractional part.
Consequently, we have
\begin{align*}
	e^{-tH_{\mu}}\Phi_{h}(x) =&\sqrt{\frac{t}{\pi}}\sum_{n=0}^{h+2}\int_1^{\infty}K_n(t,u)A_n(u)du
	\\=&-2\sqrt{\frac{t}{\pi}}\sum_{n=0}^{h+2}\int_1^{\infty}\int_0^1\cos(2st)\left(t(1-s^2)\right)^{n+u-1/2}\ln\left(t(1-s^2)\right)\sum_{j=1}^{[u]}\frac{\alpha_j^{(n)}}{\Gamma(n+j+\frac{1}{2})}dsdu
\end{align*}
Using the fact that,
\begin{equation*}
	\left|\cos(2st)(1-s^2)\right|\le 1,\quad s,t\in(0,1)
\end{equation*}
we have
\begin{multline*}
	\int_1^{\infty}\int_0^1\left|\cos(2st)\left(t(1-s^2)\right)^{n+u-1/2}\ln\left(t(1-s^2)\right)\sum_{j=1}^{[u]}\frac{\alpha_j^{(n)}}{\Gamma(n+j+\frac{1}{2})}\right|dsdu\\
	\le \int_0^1\left(t(1-s^2)\right)^{n-1/2}\left|\ln\left(t(1-s^2)\right)\right|ds\int_1^{\infty}\sum_{j=1}^{[u]}\frac{\left|\alpha_j^{(n)}\right|}{\Gamma(n+j+\frac{1}{2})}du
\end{multline*}
where
\begin{equation*}
	\int_0^1\left(t(1-s^2)\right)^{n-1/2}\left|\ln\left(t(1-s^2)\right)\right|ds<\infty
\end{equation*}
and
\begin{equation*}
	\int_1^{\infty}\sum_{j=1}^{[u]}\frac{\left|\alpha_j^{(n)}\right|}{\Gamma(n+j+\frac{1}{2})}du=\sum_{j=1}^{\infty}\frac{\left|\alpha_j^{(n)}\right|}{\Gamma(n+j+\frac{1}{2})}<\infty
\end{equation*}
Then, Fubini's theorem implies that

\begin{align*}
	e^{-tH_{\mu}}\Phi_{h}(x) =&-2\sqrt{\frac{t}{\pi}}\int_0^1\cos(2st)\ln\left(t(1-s^2)\right)\tilde{A}(s)ds\\
	=&-\sqrt{\frac{t}{\pi}}\int_{-1}^1\cos(2st)\ln\left(t(1-s^2)\right)\tilde{A}(s)ds
\end{align*}
where
\begin{align*}
	\tilde{A}(s)=\sum_{n=0}^{h+2}\int_1^{\infty}\sum_{j=1}^{[u]}\frac{\alpha_j^{(n)}\left(t(1-s^2)\right)^{n+u-1/2}}{\Gamma(n+j+\frac{1}{2})}du.
\end{align*}
Or equivalently, by performing the variable change $s\mapsto\sin\varphi$:
\begin{align*}
	e^{-tH_{\mu}}\Phi_{h}(x) =&-\sqrt{\frac{t}{\pi}}\int_{-\pi/2}^{\pi/2}\cos(2t\sin\varphi)\ln\left(t\cos^2\varphi\right)\tilde{A}(\varphi)d\varphi
\end{align*}
where
\begin{align*}
	\tilde{A}(\varphi)=\sum_{n=0}^{h+2}\int_1^{\infty}\sum_{j=1}^{[u]}\frac{\alpha_j^{(n)}\left(t\cos^2\varphi\right)^{n+u-1/2}}{\Gamma(n+j+\frac{1}{2})}du.
\end{align*}
\subsection{Arc--sine case}
Performing the index changes $m\mapsto m-n$ and $n\mapsto n+1$ respectively in the first sum of \eqref{momentum}, we can reformulate the latter as follows:
\begin{align*}
		e^{itP_{As}}T_{h}(x) =&\sum_{n=1}^{h}\sum_{m\ge 0}\frac{\left(-1\right)^{m+n-1}}{t} (m+n)T_{h+m-n+1}(x)J_{m+n}(2t)
		+2\sum_{m\ge0} \left(-1\right)^{m}T_{m}(x)\, J_{m+h}(2t)
		\\:=& \sum_{n=0}^h\mathcal{R}_n(t,x,h)
\end{align*}
where
\begin{align*}
\mathcal{R}_n(t,x,h):=\sum_{m\ge 0}\beta_m^{(n)}(t,x,h)J_{m+n}(2t)
\end{align*}
and
\begin{align*}
	\beta_m^{(n)}(t,x,h)=
	\begin{cases}
		\frac{\left(-1\right)^{m+n-1}}{t}(m+n)T_{h+m-n+1}(x)\chi_{\{h\ge2\}},&1\le n\le h-1\\
	2\left(-1\right)^{m}  T_{m}(x)+\frac{\left(-1\right)^{m+h-1}}{t}(m+h)T_{m+1}(x)\chi_{\{h\ge1\}},	&0\le n=h\\
	0,& otherwise
	\end{cases}
\end{align*}

Then, as outlined in \cite{PS09}, for all $t$ in the interval
$$\left(0,\min\left\{1,\left(e\limsup_{m\rightarrow\infty}\frac{\sqrt[m]{|\beta_m^{(n)}|}}{m}\right)^{-1}\right\}\right)=(0,1),$$
the Neumann-type series $\mathcal{R}_n$ can be rewritten as:

\begin{equation*}
	\mathcal{R}_n=\sqrt{\frac{t}{\pi}}\int_1^{\infty}K_n(t,u)B_n(u)du
\end{equation*}
where the kernel $K_n(t,u)$ is as in \eqref{kernel}
and
\begin{align*}
	B_n(u)=\sum_{j=1}^{[u]}\frac{\beta_j^{(n)}}{\Gamma(n+j+\frac{1}{2})}
\end{align*}
Consequently, we have
\begin{align*}
	e^{itP_{As}}T_{h}(x) =&\sqrt{\frac{t}{\pi}}\sum_{n=0}^{h}\int_1^{\infty}K_n(t,u)B_n(u)du
	\\=&-2\sqrt{\frac{t}{\pi}}\sum_{n=0}^{h}\int_1^{\infty}\int_0^1\cos(2st)\left(t(1-s^2)\right)^{n+u-1/2}\ln\left(t(1-s^2)\right)\sum_{j=1}^{[u]}\frac{\beta_j^{(n)}}{\Gamma(n+j+\frac{1}{2})}dsdu
\end{align*}
and similarly to the previous case, we can express this equation in integral form as follows
\begin{align*}
	e^{itP_{As}}T_{h}(x) =&-\sqrt{\frac{t}{\pi}}\int_{-\pi/2}^{\pi/2}\cos(2t\sin\varphi)
\ln\left(t\cos^2\varphi\right)\tilde{B}(\varphi)d\varphi
\end{align*}
where its density is given in terms of the coefficients of the Neumann series as follows
\begin{align*}
	\tilde{B}(\varphi)
=\sum_{n=0}^{h}\int_1^{\infty}\sum_{j=1}^{[u]}
\frac{\beta_j^{(n)}\left(t\cos^2\varphi\right)^{n+u-1/2}}{\Gamma(n+j+\frac{1}{2})}du.
\end{align*}

\end{document}